\documentclass[12pt]{article}
\usepackage[english]{babel}
\usepackage[cp1251]{inputenc}
\usepackage{doc}
\usepackage{amsmath}
\usepackage{amsfonts}   
\usepackage{cite} 
\usepackage{fullpage}                           % Makes the page margins smaller to a predefined one.
\usepackage[lmargin=0.6in,rmargin=0.6in,tmargin=0.6in,headsep=.2in]{geometry}
\usepackage{graphicx}
\usepackage{forloop}
\usepackage{etoolbox}
\usepackage{calc}
\usepackage{wrapfig,lipsum,booktabs} % To produce the empty space on the left
\usepackage{xtab} % To produce the empty space on the left
\usepackage[dvipsnames]{xcolor}
\usepackage[normalem]{ulem}
\usepackage{sectsty}% http://ctan.org/pkg/sectsty
\usepackage{amsmath,amsfonts,amssymb,amsthm,epsfig,epstopdf,url,array}
 \usepackage[retainorgcmds]{IEEEtrantools}
 \usepackage{makeidx,epsfig,lscape}
 \usepackage{xcolor,pict2e}
\usepackage{upgreek}
   
\makeatother
\usepackage{makeidx}
\makeindex
\makeatletter
\renewcommand{\@biblabel}[1]{#1.}

\newtheorem{note1}{Note}          % [section]
\newtheorem{te}{Theorem}           %[section]
\newtheorem{leme}{Lemma}         % [section]
\newtheorem{remark}{Remark}   %[section]
\newtheorem{defin}{Definition}  % [section]

\begin{document}

\centerline{ \bf\Large{  Derivatives Pricing in Non-Arbitrage Market}}

\vskip 5mm
{\bf \centerline {\Large  N.S. Gonchar \footnote{This work was partially supported   by  the Program of Fundamental Research of the Department of Physics and Astronomy of  the National Academy of Sciences of Ukraine  "Mathematical models of non equilibrium processes in open systems" N 0120U100857.}}     }

\vskip 5mm
\centerline{\bf {Bogolyubov Institute for Theoretical Physics of NAS of Ukraine.}}
\vskip 2mm

\begin{abstract}

The general method is proposed for constructing a family of martingale measures for a wide class of evolution of risky assets. The sufficient conditions are formulated for the evolution of risky assets under which the family of equivalent  martingale measures to the original measure is a non-empty set. The set of martingale measures is constructed from a set of strictly nonnegative random variables, satisfying certain conditions. The inequalities are obtained for the non-negative random variables satisfying certain conditions. Using these inequalities, a new simple  proof of   optional decomposition  theorem for the nonnegative super-martingale is proposed. The family of spot measures is introduced and the representation is found for them. The conditions are found under which each martingale measure is an integral over the set of spot measures. On the basis of nonlinear processes such as ARCH and GARCH, the parametric family of random processes is introduced for which the interval of non-arbitrage prices are found.
The formula is obtained for the fair price of the contract with  option of  European type for the considered parametric processes. The parameters of the introduced random processes are estimated and the estimate is found at which the fair price of  contract with  option is the least.
\end{abstract}

\centerline{{\bf Keywords:}  Random process; Spot set of measures; } 
\centerline{ Optional Doob decomposition; Super-martingale;}
 \centerline{Martingale; Assessment of derivatives.}

\section{Introduction.}
The study of non-arbitrage markets was begun for the first time in Bachelier's work \cite{Bachelier}. Then, in the famous works of Black F. and Scholes M. \cite{Black} and Merton R. S. \cite{Merton} the formula was found for the fair price of the standard call option  of European type.
The absence of arbitrage in the financial market has a very transparent economic sense, since it can be considered reasonably arranged.
The concept of non arbitrage in  financial market is associated with the fact that one cannot earn money without risking, that is,
to make money you need to invest in  risky or risk-free assets. 
The exact mathematical substantiation of the concept of non arbitrage   was first made in the papers
 \cite{HK79}, \cite{HP81} for the finite probability space and in the general case in the paper \cite{DMW90}. In the continuous time evolution of risky asset,  the proof of  absent of arbitrage possibility see in \cite{WalterSchacher}. The value of the established Theorems is that they make it possible to value assets.
They got a special name "The First and The Second Fundamental Asset Pricing   Theorems." Generalizations of these Theorems are contained in papers \cite{Rogers}, \cite{Shiryaev},  \cite{Shiryaev1}.

If the martingale measure is not the only one for a given evolution of a risky asset, then a rather difficult problem  of describing all martingale measures arises in order to evaluate, for example, derivatives.

Assessment of risk in various systems was begun in papers \cite{Gonchar2},\cite{Gonchar555}, \cite{Gonchar557}, \cite{Honchar100}.

Statistical studies of the time series of the logarithm of the price ratio of risky assets contain heavy tails in distributions with strong elongation in the central region. The temporal behavior of these quantities exhibits the property of clustering and a strong dependence on the past. All this should be taken into account when building models for the evolution of risky assets. 

In this paper, we generalize the results of the papers \cite{GoncharSimon},   \cite{GoncharNick1},  \cite{GoncharNick} and   construct the  evolution    of risky assets for which we completely describe the set of equivalent martingale measures. 

The aim of this study is to describe the family of martingale measures for a wide class of risky asset evolutions. The paper proposes the general concept for constructing the family of martingale measures equivalent to a given measure for a wide class of evolutions of risky assets.
In particular, it also contains the description of the family of martingale measures for the evolution of risky assets given  by the ARCH \cite{Engle} and  GARCH  \cite{Bollerslev}, \cite{EngleBollerslev} models.
In  section 2,  we formulate the conditions relative to the evolution of risky assets
and give the examples of risky asset evolution satisfying these conditions. 
 Section 3 contains the  construction of measures by recurrent relations.
It is shown that under the conditions relative to the  evolution of risky assets such construction is meaningful. It is proved that the constructed set of measures is equivalent to an initial measure. In theorem \ref{witka2}, we are proved that under  certain integrability conditions of risky asset evolution the set of constructed measures is a set of martingale measures relative to this evolution of risky asset. In Section 4 we prove the inequalities for the nonnegative random values  very useful for the proof of optional decomposition for the non negative super-martingales relative to the set of all martingale measures. 

First, we show an integral inequality for a nonnegative random variable under the inequality for this nonnegative random variable with respect to the constructed family of measures. Further, using this integral inequality for the non-negative random variable, a pointwise system of inequalities is obtained for this non-negative random variable for a particular  case. After that, the pointwise system of inequalities  is obtained for the non-negative random variable in the general case. Then, using the resulting pointwise system of inequalities, an inequality is established for this non-negative random variable whose right-hand side is such that its conditional mathematical expectation is equal to one.

On the basis of the results of Section 4, in Section 5, we prove the optional decomposition for the non negative super-martingales.
In Section 6, we introduce the spot measures by the recurrent relations and find the representation for them. Using these facts under  certain conditions  we prove 
integral representation for every martingale measure over the set of spot  measures. 

First, the optional decomposition for   diffusion processes super-martingale was opened by  by  El Karoui N. and  Quenez M. C. \cite{KarouiQuenez}. After that, Kramkov D. O. and Follmer H. \cite{Kramkov}, \cite{FolmerKramkov1} proved the optional decomposition for the nonnegative bounded super-martingales.  Folmer H. and Kabanov Yu. M.  \cite{FolmerKabanov1},  \cite{FolmerKabanov}  proved analogous result for an arbitrary super-martingale. Recently, Bouchard B. and Nutz M. \cite{Bouchard1} considered a class of discrete models and proved the necessary and sufficient conditions for the validity of the optional decomposition.

Section 7 contains applications of the results obtained.
 A class of random processes is considered, which contains well-known processes of the type  ARCH and GARCH ones. Two types of random processes are considered, those for which the price of an asset cannot go down to zero and those for which the price can go down to zero during the period under consideration. The first class of processes describes the evolution of well-managed assets. We will call these assets relatively stable.
 For the evolution of relatively stable assets in the period under consideration, the family of martingale measures is one and the same.
The family of martingale measures for the evolution of risky assets whose price can drop to zero is contained in the family of martingale measures for the evolution of relatively stable assets. Each of the martingale measures for the considered class of evolutions is an integral over the set of spot martingale measures.

The interval of non-arbitrage prices is found for a wide class of payoff functions in the case when evolution describes relatively unstable assets.
This range is quite wide for the payoff functions of standard put and call options. The fair price of the super hedge is in this case the starting price of the underlying asset. 
The estimates are found for the fair price of the super-hedge for the introduced class of evolutions with respect to stable assets.
The formulas are found for the fair price of contracts with call and put options for the evolution of assets described by parametric processes.

 The same formulas are found for Asian-type put and call options. A characteristic feature of these estimates is that for the evolution of relatively stable assets the fair price of the super hedge is less than the price of the underlying asset.

In Section 8, the estimates of the parameters of risky assets included in the evolution are obtained. 
The formulas are found for the fair price of contracts with call and put options for the obtained parameter estimates, and the interval of non-arbitrage prices for different statistics is found. The same results are obtained for Asian-style call and put options.

\section{Evolutions of risky assets.}

Let $\{\Omega_N, {\cal F}_N, P_N\}$ be a  direct  product 
of the probability spaces $\{\Omega_i^0, {\cal F}_i^0, P_i^0\}, \ i=\overline{1, N}, $ 
$\Omega_N=\prod\limits_{i=1}^N\Omega_i^0,$ $P_N=\prod\limits_{i=1}^N P_i^0,$
${\cal F}_N=\prod\limits_{i=1}^N  {\cal F}_i^0,$ where
  the $\sigma$-algebra ${\cal F}_N$ is a minimal $\sigma$-algebra, generated by the sets $\prod\limits_{i=1}^N G_i, \  G_i \in {\cal F}_i^0.$ 
 On the measurable space $\{\Omega_N, {\cal F}_N\},$ under the filtration  ${\cal F}_n, \ n=\overline{1, N},$  we understand the minimal 
$\sigma$-algebra generated by the sets $\prod\limits_{i=1}^N G_i, \  G_i \in {\cal F}_i^0,$ where $G_i=\Omega_i^0$ for  $ i>n.$
We also introduce  the probability spaces $\{\Omega_n, {\cal F}_n, P_n\}, n=\overline{1, N}, $ where $\Omega_n=\prod\limits_{i=1}^n\Omega_i^0,$ ${\cal F}_n=\prod\limits_{i=1}^n  {\cal F}_i^0,$  $P_n=\prod\limits_{i=1}^n P_i^0.$
There is a one-to-one correspondence between the sets of the $\sigma$-algebra ${\cal F}_n,$ belonging to the introduced  filtration, and the sets of the  $\sigma$-algebra ${\cal F}_n=\prod\limits_{i=1}^n  {\cal F}_i^0$  of the measurable space $\{\Omega_n, {\cal F}_n\}, n=\overline{1, N}. $  Therefore, we don't  introduce new denotation for the $\sigma$-algebra  ${\cal F}_n$ of the measurable space  $\{\Omega_n, {\cal F}_n\},$ since  it always will be clear the difference between 
 the above introduced  $\sigma$-algebra  ${\cal F}_n$ of  filtration  on the measurable  space $\{\Omega_N, {\cal F}_N\}$ and the  $\sigma$-algebra ${\cal F}_n$  of the measurable space $\{\Omega_n, {\cal F}_n\}, n=\overline{1, N}. $

We assume that the evolution of risky asset $\{S_n\}_{n=0}^N, $ 
 given  on the probability space  $\{\Omega_N, {\cal F}_N, P_N\},$  is consistent with the filtration $ {\cal F}_n$, that is,  $S_n$ is a ${\cal F}_n$-measurable.
Due to the above one-to-one correspondence between the sets of the $\sigma$-algebra ${\cal F}_n,$ belonging to the introduced  filtration, and the sets of the  $\sigma$-algebra ${\cal F}_n$  of the measurable  space $\{\Omega_n, {\cal F}_n\}, n=\overline{1, N},$ we  give the evolution of risky assets in the form $\{S_n(\omega_1, \ldots, \omega_n)\}_{n=0}^N, $ where
$S_n(\omega_1, \ldots, \omega_n)$ is an ${\cal F}_n$-measurable random variable, given on the measurable space  $\{\Omega_n, {\cal F}_n\}.$  It is evident that such evolution is consistent with the filtration $ {\cal F}_n$ on the measurable space $\{\Omega_N, {\cal F}_N, P_N\}.$

Further, we assume that 
$$P_n( (\omega_1, \ldots,\omega_n) \in \Omega_n, \  \Delta S_n>0)>0, $$
\begin{eqnarray}\label{vitasja1}
 P_n( (\omega_1, \ldots,\omega_n) \in \Omega_n, \  \Delta S_n<0)>0, \quad n=\overline{1,N},
\end{eqnarray}
where $\Delta S_n=S_n(\omega_1, \ldots,\omega_n)- S_{n-1}(\omega_1, \ldots, \omega_{n-1}), \  n=\overline{1,N}.$

Let us introduce the denotations
\begin{eqnarray}\label{vitasja2}
\Omega_n^-=\{(\omega_1, \ldots,\omega_n) \in \Omega_n, \ \Delta S_n\leq 0\}, \quad \Omega_n^+=\{(\omega_1, \ldots,\omega_n) \in \Omega_n, \ \Delta S_n > 0\},
\end{eqnarray}
\begin{eqnarray}\label{vitasja3}
\Delta S_n^-=-\Delta S_n \chi_{\Omega_n^-}(\omega_1, \ldots,\omega_n), \quad  \Delta S_n^+=\Delta S_n \chi_{\Omega_n^+}(\omega_1, \ldots,\omega_n),
\end{eqnarray}

$$ V_n(\omega_1, \ldots,\omega_{n-1}, \omega_n^1,  \omega_n^2)=\Delta S_n^-(\omega_1, \ldots,\omega_{n-1}, \omega_n^1)+\Delta S_n^+(\omega_1, \ldots,\omega_{n-1}, \omega_n^2), $$
\begin{eqnarray}\label{vitasja4}
(\omega_1, \ldots,\omega_{n-1}, \omega_n^1) \in \Omega_n^-, \quad (\omega_1, \ldots,\omega_{n-1}, \omega_n^2) \in \Omega_n^+.
\end{eqnarray}

 We use the following denotation $ \Omega_n^{a}, \ n=\overline{1,N}, $ where $a$ takes two values $-$ and $+.$ 

Our assumption,  in this paper, is that for $ \Omega_n^{a}, \ a=-,+,$ the  representations 
\begin{eqnarray}\label{100vitasikja7}
  \Omega_n^{-}=\bigcup\limits_{k=1}^{N_n}[A_n^{0,k-}\times V_{n-1}^k],
\quad   \Omega_n^{+}=\bigcup\limits_{k=1}^{N_n}[A_n^{0,k+}\times V_{n-1}^k], \quad  N_n\leq \infty,
\end{eqnarray}
are true, where
$$\Omega_{n-1}=\bigcup\limits_{k=1}^{N_n}V_{n-1}^k, \ A_n^{0,k-}, \ A_n^{0,k+} \in {\cal F}_n^0, \quad  A_n^{0,k-}\cup A_n^{0,k+}=\Omega_n^0,$$
\begin{eqnarray}\label{101vitasikja7}
 \quad A_n^{0,k-}\cap A_n^{0,k+}=\emptyset, \quad  V_{n-1}^k \cap  V_{n-1}^j=\emptyset,\quad  k \neq j, \quad  V_{n-1}^k \in {\cal F}_{n-1}.
\end{eqnarray}
The number $N_n$ may be finite or infinite. Since $\Omega_n^-\cup \Omega_n^+=\Omega_n,$ $\Omega_n^-\cap \Omega_n^+=\emptyset,$ and $P_n(\Omega_n^-)>0, \ P_n(\Omega_n^+)>0, $ we have

$$P_n(\Omega_n^-)=\sum\limits_{k=1}^{N_n} P_n^0( A_n^{0,k-}) P_{n-1}(V_{n-1}^k ), $$ 
\begin{eqnarray}\label{102vitasikja7}
P_n(\Omega_n^+)=\sum\limits_{k=1}^{N_n} P_n^0( A_n^{0,k+}) P_{n-1}(V_{n-1}^k ), \quad P_n^0( A_n^{0,k-})+P_n^0( A_n^{0,k+})=1.
\end{eqnarray}

Further, in this paper,  we assume that  $P_n^0( A_n^{0,k-})>0, \ P_n^0( A_n^{0,k+})>0, \ n=\overline{1,N}, \ k=\overline{1,N_n}.$
We also assume some technical suppositions: there exist subsets $ B_{n,i}^{0,k -} \in {\cal F}_n^0,$  $ \   i=\overline{1,I_n},\ I_n>1, $  and $ B_{n,s}^{0,k +} \in {\cal F}_n^0,$  $ \   s=\overline{1,S_n},\ S_n>1, $ 
satisfying the conditions
  $$   \ B_{n,i}^{0,k -}\cap B_{n,j}^{0,k -}=\emptyset,\  i \neq j, \quad  \ B_{n,s}^{0,k +}\cap B_{n,l}^{0,k +}=\emptyset,\  s\neq l, \quad k=\overline{1,N_n},$$ 
$$  P_n^0( B_{n,i}^{0,k -})>0,\    i=\overline{1,I_n},\   P_n^0( B_{n,s}^{0,k +})>0,\    s=\overline{1,S_n}, \quad k=\overline{1,N_n},$$
\begin{eqnarray}\label{1vitasika1ja9}
A_n^{0,k -}=\bigcup\limits_{i=1}^{I_n}B_{n,i}^{0,k -},  \quad A_n^{0,k +}=\bigcup\limits_{s=1}^{S_n}B_{n,s}^{0,k +}, \quad k=\overline{1,N_n}.
\end{eqnarray}

Below, we give the examples of  evolutions $\{S_n(\omega_1, \ldots,\omega_n)\}_{n=1}^N $ for which the representations  (\ref{100vitasikja7}) are true.
 
Suppose that  the random values $ a_i(\omega_1, \ldots,\omega_i), $ $\eta_i(\omega_i)$ satisfy the inequalities
$0 < a_i(\omega_1, \ldots,\omega_i)\leq 1,$ $1+\eta_i(\omega_i) \geq 0,$
$P_i^0(\eta_i(\omega_i)<0)>0,$ $P_i^0(\eta_i(\omega_i)>0)>0,$ 
$ i=\overline{1,N}.$
If $S_n(\omega_1, \ldots,\omega_n) $ is given by the formula
\begin{eqnarray}\label{tin1vitasika1ja9}
S_n(\omega_1, \ldots,\omega_n)=S_0\prod\limits_{i=1}^n(1+a_i(\omega_1, \ldots,\omega_i)\eta_i(\omega_i)), \quad n=\overline{1,N},
\end{eqnarray}
then
$$ \{\omega_i \in \Omega_i^0, \ 
\eta_i(\omega_i) \leq 0\}=A_i^{0,1-}, \quad \{\omega_i \in \Omega_i^0, \ 
\eta_i(\omega_i) > 0\}=A_i^{0,1+},$$
\begin{eqnarray}\label{tinna1vitasika1ja10}
V_{i-1}^1= \Omega_{i-1},  \quad \Omega_i^-=A_i^{0,1-}\times \Omega_{i-1}, \quad \Omega_i^+=A_i^{0,1+}\times \Omega_{i-1},\quad i=\overline{1,N}.
\end{eqnarray}
In general case, let us consider the evolution of risky asset $\{S_n(\omega_1, \ldots,\omega_n)\}_{n=1}^N, $   given by the formula
$$S_n(\omega_1, \ldots,\omega_n)=$$
\begin{eqnarray}\label{tin1vitasika1ja11}
S_0\prod\limits_{i=1}^n(1+\sum\limits_{k=1}^{N_i}\eta_i^k(\omega_i)
\chi_{V_{i-1}^k}(\omega_1, \ldots,\omega_{i-1}) a_i^k(\omega_1, \ldots,\omega_i)), \quad n=\overline{1,N},
\end{eqnarray}
where the random values $ a_i^k(\omega_1, \ldots,\omega_i), $ $\eta_i^k(\omega_i)$ satisfy the inequalities \\
$0 < a_i^k(\omega_1, \ldots,\omega_i)\leq 1,$ $1+\eta_i^k(\omega_i) \geq 0,$
$P_i^0(\eta_i^k(\omega_i)<0)>0,$ $P_i^0(\eta_i^k(\omega_i)>0)>0,$ 
$ i=\overline{1,N},$ \  $ k=\overline{1,N_n},$
 and
$\bigcup\limits_{k=1}^{N_i} V_{i-1}^k= \Omega_{i-1}, \ V_{i-1}^k\cap V_{i-1}^s,$ $ k \neq s .$
Then, if to put 
$$ \{\omega_i \in \Omega_i^0, \ 
\eta_i^k(\omega_i) \leq 0\}=A_i^{0,k-}, \quad \{\omega_i \in \Omega_i^0, \ 
\eta_i^k(\omega_i) > 0\}=A_i^{0,k+},$$
we obtain
\begin{eqnarray}\label{tin1vitasika1ja10}
 \Omega_i^-=\bigcup\limits_{k=1}^{N_i} [A_i^{0,k-}\times V_{i-1}^k], \quad \Omega_i^+=\bigcup\limits_{k=1}^{N_i} [A_i^{0,k+}\times V_{i-1}^k],\quad i=\overline{1,N}.
\end{eqnarray}

$$ \Delta S_n(\omega_1, \ldots,\omega_{n-1}, \omega_{n})\leq 0, \quad (\omega_1, \ldots,\omega_{n-1}, \omega_{n}) \in \Omega_n^-, ,\quad n=\overline{1,N}, $$
\begin{eqnarray}\label{tinkakrasunja3}
\Delta S_n(\omega_1, \ldots,\omega_{n-1}, \omega_{n})>0, \quad (\omega_1, \ldots,\omega_{n-1}, \omega_{n}) \in \Omega_n^+,\quad n=\overline{1,N}.
\end{eqnarray}

\section{Construction of the set of martingale measures.}

 In this section, we present the construction of the set of  measures  on the basis of evolution of risky assets given by the formulas (\ref{tin1vitasika1ja9}), (\ref{tin1vitasika1ja11}) on the measurable space $\{\Omega_N, {\cal F}_N\}.$ For this purpose,  we use the set  of nonnegative  random values $\alpha_n(\{\omega_1^1, \ldots,\omega_{n-1}^1, \omega_n^1\};\{\omega_1^2, \ldots,\omega_{n-1}^2, \omega_n^2\}),$  given on the probability space 
$\{\Omega_n^-\times \Omega_n^+, {\cal F}_n^-\times {\cal F}_n^+, P_n^-\times P_n^+\}, \ n=\overline{1,N},$ where $ {\cal F}_n^-={\cal F}_n\cap \Omega_n^-,  \ {\cal F}_n^+={\cal F}_n\cap \Omega_n^+.$ The measure $P_n^-$ is a contraction of the measure $P_n$ on the $\sigma$-algebra $ {\cal F}_n^-$ and
the measure $P_n^+$ is a contraction of the measure $P_n$ on the $\sigma$-algebra $ {\cal F}_n^+.$ 
   After that, we prove that this set of measures, defined the above set of random values,  is equivalent to the measure $P_N.$
At last, Theorem \ref{witka2} gives the sufficient conditions under that the constructed set of measures is a set of martingale measures for the considered evolution of risky assets.
Sometimes, we use the abbreviated denotations $\{\omega_1^1, \ldots,\omega_n^1\}=\{\omega\}_n^1, \{\omega_1^2, \ldots,\omega_n^2\}=\{\omega\}_n^2.$

We assume that  the set of  random values $\alpha_n(\{\omega_1^1, \ldots,\omega_n^1\};\{\omega_1^2, \ldots,\omega_n^2\})=\alpha_n(\{\omega\}_n^1;\{\omega\}_n^2),$ $  (\{\omega\}_n^1;\{\omega\}_n^2)  \in \Omega_n^-\times\Omega_n^+, $ $ \ n=\overline{1,N},$ satisfies  the following conditions: 

$$ P_n^-\times P_n^+( (  \{\omega\}_n^1;\{\omega\}_n^2)   \in \Omega_n^-\times \Omega_n^+,  \alpha_n(  \{\omega\}_n^1;\{\omega\}_n^2)>0)=$$
\begin{eqnarray}\label{1vitasja7}
P_n(\Omega_n^-)\times P_n(\Omega_n^+), \quad n=\overline{1,N};
\end{eqnarray}
 $$  \int\limits_{\Omega_n^0\times \Omega_n^0}\chi_{\Omega_n^-}(\omega_1^1, \ldots,\omega_{n-1}^1, \omega_n^1) \chi_{\Omega_n^+}(\omega_1^2, \ldots,\omega_{n-1}^2, \omega_n^2)\times$$
$$\alpha_n(\{\omega_1^1, \ldots,\omega_{n-1}^1, \omega_n^1\};\{\omega_1^2, \ldots, \omega_{n-1}^2,\omega_n^2\})\times $$
$$\frac{\Delta S_n^+(\omega_1, \ldots,\omega_{n-1}, \omega_n^2) \Delta S_n^-(\omega_1, \ldots,\omega_{n-1}, \omega_n^1)}{V_n(\omega_1, \ldots,\omega_{n-1}, \omega_n^1,  \omega_n^2)}d P_n^0(\omega_n^1) dP_n^0(\omega_n^2)< \infty, $$
 $$(\{\omega_1^1, \ldots,\omega_{n-1}^1\};\{\omega_1^2, \ldots,\omega_{n-1}^2\}) \in \Omega_{n-1}\times\Omega_{n-1}, $$
\begin{eqnarray}\label{2vitasja7}
 (\omega_1, \ldots,\omega_{n-1}) \in \Omega_{n-1},  \quad n=\overline{1,N};
\end{eqnarray}
$$  \int\limits_{\Omega_n^0\times \Omega_n^0}\chi_{\Omega_n^-}(\omega_1^1, \ldots,\omega_{n-1}^1, \omega_n^1)\chi_{\Omega_n^+}(\omega_1^2, \ldots,\omega_{n-1}^2, \omega_n^2)\times $$

$$\alpha_n(\{\omega_1^1, \ldots,\omega_{n-1}^1, \omega_n^1\};\{\omega_1^2, \ldots, \omega_{n-1}^2,\omega_n^2\})dP_n^0(\omega_n^1) dP_n^0(\omega_n^2)=1,$$
\begin{eqnarray}\label{3vitasja7} 
 (\{\omega_1^1, \ldots,\omega_{n-1}^1\};\{\omega_1^2, \ldots,\omega_{n-1}^2\}) \in \Omega_{n-1}\times\Omega_{n-1},  \quad n=\overline{1,N}.
\end{eqnarray}

In the next Lemma \ref{witka0}, we give the sufficient conditions under which the conditions (\ref{1vitasja7}) - (\ref{3vitasja7}) are valid.

\begin{leme}\label{witka0}
 Suppose that for  $ \Omega_n^{a}, a=-,+, \ n=\overline{1,N}, $   the representations (\ref{100vitasikja7})
are true. If the conditions 
$$ \inf\limits_{1\leq k\leq N_n} P_n^0(A_n^{0,k -}\setminus B_{n,i}^{0,k-})>0, \quad i=\overline{1,I_n}, \quad  I_n>1,\quad n=\overline{1,N}, $$ $$  \inf\limits_{1\leq k\leq N_n} P_n^0(A_n^{0,k +}\setminus B_{n,s}^{0,k+})>0,\quad s=\overline{1,S_n}, \quad S_n>1, \quad n=\overline{1,N},$$
$$ \inf\limits_{1\leq k\leq N_n} P_n^0( B_{n,i}^{0,k-})>0, \quad i=\overline{1,I_n}, \quad  I_n>1,\quad n=\overline{1,N}, $$ $$  \inf\limits_{1\leq k \leq N_n} P_n^0( B_{n,s}^{0,k+})>0,\quad s=\overline{1,S_n}, \quad S_n>1, \quad n=\overline{1,N},$$
\begin{eqnarray}\label{pupswitkapup0}
\int\limits_{\Omega_N} \Delta S_n^-(\omega_1, \ldots,\omega_{n-1}, \omega_n) d P_N <\infty, \quad \ n=\overline{1,N},
\end{eqnarray}
are true,
then the set of bounded  random values  $\alpha_n(\{\omega\}_n^1;\{\omega\}_n^2),$ satisfying the conditions  (\ref{1vitasja7}) - (\ref{3vitasja7}), is a  nonempty set.
\end{leme}
\begin{proof}
Let us put 
$$\alpha_{n}^{i-} (\omega_1^1, \ldots,\omega_{n}^1)=\sum\limits_{k=1}^{N_n}\alpha_{n,k, i}^-( \omega_n^1)\chi_{A_n^{0,k-}}(\omega_n^1) \chi_{V_{n-1}^k}(\omega_1^1, \ldots,\omega_{n-1}^1),  $$

$$\alpha_{n}^{s+} (\omega_1^2, \ldots,\omega_{n}^2)=\sum\limits_{k=1}^{N_n}\alpha_{n,k, s}^+( \omega_n^2)\chi_{A_n^{0,k+}}(\omega_n^2) \chi_{V_{n-1}^k}(\omega_1^2, \ldots,\omega_{n-1}^2),  $$
where
$$\alpha_{n,k, i}^-( \omega_n^1)=
(1-\delta_i^n)\frac{\chi_{B_{n,i}^{0,k-}}(\omega_n^1)}{P_n^0(B_{n,i}^{0,k-})}+\delta_i^n \frac{\chi_{A_n^{0,k-}\setminus B_{n,i}^{0,k-}}(\omega_n^1)}{P_n^0(A_n^{0,k-}\setminus B_{n,i}^{0,k-})},$$
\begin{eqnarray}\label{2vitasja9}
\ 0< \delta_i^n<1, \quad  i=\overline{1, I_n}, \quad  
 k=\overline{1,N_n},
\end{eqnarray}
$$\alpha_{n,k, s}^+( \omega_n^2)=
(1-\mu_s^n)\frac{\chi_{B_{n,s}^{0,k+}}(\omega_n^2)}{P_n^0(B_{n,s}^{0,k+})}+\mu_s^n \frac{\chi_{A_n^{0,k+}\setminus B_{n,s}^{0,k+}}(\omega_n^2)}{P_n^0(A_n^{0,k+}\setminus B_{n,s}^{0,k+})},$$
\begin{eqnarray}\label{29vitasikafja9}
 0< \mu_s^n<1,  \quad  s=\overline{1, S_n}, \quad
 k=\overline{1,N_n}.
\end{eqnarray}
If to introduce   the nonnegative set of real numbers
\begin{eqnarray}\label{89vitasikafja9}
\gamma_{i, s}\geq 0,  \quad  i=\overline{1, I_n}, \quad s=\overline{1, S_n}, \quad  \sum\limits_{ i, s=1}^{I_n,  S_n} \gamma_{i,s}=1, \quad n=\overline{1,N},
\end{eqnarray}
 then
$$\alpha_n(\{ \omega_1^1, \ldots,\omega_{n}^1\};\{ \omega_1^2, \ldots,\omega_{n}^2\})=$$
\begin{eqnarray}\label{4vitasja9}
\sum\limits_{i, s=1}^{ I_n, S_n} \gamma_{i,s} 
\alpha_{n}^{i-} (\omega_1^1, \ldots,\omega_{n}^1)
\alpha_{n}^{s+} (\omega_1^2, \ldots,\omega_{n}^2),  \quad n=\overline{1,N},
\end{eqnarray}
satisfies the condition (\ref{1vitasja7}) - (\ref{3vitasja7}).

Really, due to the  Lemma \ref{witka0} conditions,   the     random values $\alpha_n(\{\omega\}_n^1;\{\omega\}_n^2\}),$ $ \  n=\overline{1,N},$  are strictly positive by construction. Therefore, the conditions (\ref{1vitasja7}) are true.

Due to the boundedness of  $ \alpha_n(\{\omega\}_n^1;\{\omega\}_n^2\}) \leq C, \  n=\overline{1,N}, \ 0<C< \infty, $ the inequalities
 $$  \int\limits_{\Omega_n^0\times \Omega_n^0}\chi_{\Omega_n^-}(\omega_1^1, \ldots,\omega_{n-1}^1, \omega_n^1) \chi_{\Omega_n^+}(\omega_1^2, \ldots,\omega_{n-1}^2, \omega_n^2)\times$$
$$\alpha_n(\{\omega_1^1, \ldots,\omega_{n-1}^1, \omega_n^1\};\{\omega_1^2, \ldots, \omega_{n-1}^2,\omega_n^2\})\times $$
$$\frac{\Delta S_n^+(\omega_1, \ldots,\omega_{n-1}, \omega_n^2) \Delta S_n^-(\omega_1, \ldots,\omega_{n-1}, \omega_n^1)}{V_n(\omega_1, \ldots,\omega_{n-1}, \omega_n^1,  \omega_n^2)}d P_n^0(\omega_n^1) dP_n^0(\omega_n^2)\leq $$ 
\begin{eqnarray}\label{pip2vitasja7pip}
 C  \int\limits_{\Omega_n^0}\Delta S_n^-(\omega_1, \ldots,\omega_{n-1}, \omega_n^1)d P_n^0(\omega_n^1) < \infty, \quad n=\overline{1,N}, 
\end{eqnarray}
are true for almost everywhere 
 $(\omega_1, \ldots,\omega_{n-1}) \in \Omega_{n-1}, \ n=\overline{1,N},$
relative to the measure $P_{n-1},$ owing  to the inequalities (\ref{pupswitkapup0}) and Foubini Theorem. This proves the inequality (\ref{2vitasja7}).
The equality (\ref{3vitasja7}) is also satisfied due to the construction of $\alpha_n(\{\omega\}_n^1;\{\omega\}_n^2).$
 Lemma \ref{witka0} is proved.
\end{proof}

  The values, which  the random variables $\alpha_n(\{\omega\}_n^1;\{\omega\}_n^2\}),\ n=\overline{1,N},$ constructed in Lemma \ref{witka0}, take, are determined by the values at points $ \omega_n^1 \in \Omega_n^{0-}$  and $ \omega_n^2 \in \Omega_n^{0+}$ for all $(\omega_1, \ldots,\omega_{n-1}) \in \Omega_{n-1}.$

On the basis of the set of  random values  $\alpha_n(\{\omega\}_n^1;\{\omega\}_n^2), \ n=\overline{1,N},$ constructed in Lemma \ref{witka0},  let us introduce into consideration  the family of measure $\mu_0(A)$ on the measurable space $\{\Omega_N, {\cal F}_N\}$ by the recurrent relations 

$$\mu_{N}^{(\omega_1, \ldots,\omega_{N-1})}(A)=\int\limits_{\Omega_N^0\times \Omega_N^0} \chi_{\Omega_N^-}(\omega_1, \ldots, \omega_{N-1},\omega_N^1) \chi_{\Omega_N^+}(\omega_1, \ldots, \omega_{N-1}, \omega_N^2)\times$$
$$\alpha_N(\{\omega_1, \ldots, \omega_{N-1},\omega_N^1\};\{\omega_1, \ldots, \omega_{N-1}, \omega_N^2\})\times $$ $$ \left[\frac{\Delta S_N^+(\omega_1, \ldots,\omega_{N-1}, \omega_N^2)}{V_N(\omega_1, \ldots,\omega_{N-1}, \omega_N^1,  \omega_N^2)}\mu_{N}^{(\omega_1, \ldots,\omega_{N-1}, \omega_N^1)}(A)+ \right.$$
\begin{eqnarray}\label{vitasja6}
\left. \frac{\Delta S_N^-(\omega_1, \ldots,\omega_{N-1}, \omega_N^1)}{V_N(\omega_1, \ldots,\omega_{N-1}, \omega_N^1,  \omega_N^2)}\mu_N^{(\omega_1, \ldots,\omega_{N-1}, \omega_N^2)}(A)\right]dP_N^0(\omega_N^1)dP_N^0(\omega_N^2),
\end{eqnarray}

$$\mu_{n-1}^{(\omega_1, \ldots,\omega_{n-1})}(A)=\int\limits_{\Omega_n^0\times \Omega_n^0} \chi_{\Omega_n^-}(\omega_1, \ldots,\omega_{n-1}, \omega_n^1) \chi_{\Omega_n^+}(\omega_1, \ldots,\omega_{n-1}, \omega_n^2)\times$$

$$\alpha_n(\{\omega_1, \ldots,\omega_{n-1}, \omega_n^1\};\{\omega_1, \ldots,\omega_{n-1}, \omega_n^2\})\times $$ $$ \left[\frac{\Delta S_n^+(\omega_1, \ldots,\omega_{n-1}, \omega_n^2)}{V_n(\omega_1, \ldots,\omega_{n-1}, \omega_n^1,  \omega_n^2)}\mu_{n}^{(\omega_1, \ldots,\omega_{n-1}, \omega_n^1)}(A)+ \right.$$
\begin{eqnarray}\label{vitasja5}
\left. \frac{\Delta S_n^-(\omega_1, \ldots,\omega_{n-1}, \omega_n^1)}{V_n(\omega_1, \ldots,\omega_{n-1}, \omega_n^1,  \omega_n^2)}\mu_{n}^{(\omega_1, \ldots,\omega_{n-1}, \omega_n^2)}(A)\right]dP_n^0(\omega_n^1)dP_n^0(\omega_n^2), \quad n=\overline{2,N},
\end{eqnarray}
$$\mu_{0}(A)=\int\limits_{\Omega_1^0\times \Omega_1^0} \chi_{\Omega_1^-}(\omega_1^1) \chi_{\Omega_1^+}( \omega_1^2)
\alpha_1( \omega_1^1; \omega_1^2)\times $$
\begin{eqnarray}\label{pupecvitasja5}
 \left[\frac{\Delta S_1^+(\omega_1^2)}{V_1( \omega_1^1,  \omega_1^2)}\mu_{1}^{( \omega_1^1)}(A)+ 
\frac{\Delta S_1^-(\omega_1^1)}{V_1( \omega_1^1,  \omega_1^2)}\mu_{1}^{( \omega_1^2)}(A)\right]dP_1^0(\omega_1^1)dP_1^0(\omega_1^2), 
\end{eqnarray}
where we put
\begin{eqnarray}\label{pupecvitasjapuc5}
 \mu_N^{(\omega_1, \ldots,\omega_{N-1}, \omega_N)}(A)=\chi_{A}(\omega_1, \ldots,\omega_{N-1}, \omega_N), \quad A \in {\cal F}_N.
\end{eqnarray}

\begin{leme}\label{witka1}
Suppose that the conditions of  Lemma \ref{witka0} are true.
 For the measure $\mu_0(A),\  A \in {\cal F}_N,$ constructed by the recurrent relations (\ref{vitasja6}) - (\ref{pupecvitasja5}),  the representation 
\begin{eqnarray}\label{vitasja10}
\mu_0(A)=\int\limits_{\Omega_N}\prod\limits_{n=1}^N\psi_n(\omega_1, \ldots,\omega_n)\chi_{A}(\omega_1, \ldots,\omega_N)\prod\limits_{i=1}^Nd P_i^0(\omega_i) 
\end{eqnarray}
is true and $\mu_0(\Omega_N)=1,$ that is, the measure $\mu_0(A)$ is a probability measure being equivalent to the measure $P_N,$ where we put
$$ \psi_n(\omega_1, \ldots,\omega_n)=\chi_{\Omega_n^-}(\omega_1, \ldots,\omega_{n-1}, \omega_n) \psi_n^1(\omega_1, \ldots,\omega_n)+$$
\begin{eqnarray}\label{vitasja11}
\chi_{\Omega_n^+}(\omega_1, \ldots,\omega_{n-1}, \omega_n) \psi_n^2(\omega_1, \ldots,\omega_n),
\end{eqnarray}
$$\psi_n^1(\omega_1, \ldots,\omega_{n-1},\omega_n)=$$ $$\int\limits_{\Omega_n^0}\chi_{\Omega_n^+}(\omega_1, \ldots,\omega_{n-1}, \omega_n^2)\alpha_n(\{\omega_1, \ldots, \omega_{n-1},\omega_n^1\};\{\omega_1, \ldots,\omega_{n-1}, \omega_n^2\}) \times $$
\begin{eqnarray}\label{vitasja12}
\frac{\Delta S_n^+(\omega_1, \ldots,\omega_{n-1}, \omega_n^2)}{V_n(\omega_1, \ldots,\omega_{n-1}, \omega_n^1,  \omega_n^2)}d P_n^0(\omega_n^2), \quad (\omega_1, \ldots,\omega_{n-1}) \in \Omega_{n-1},
\end{eqnarray} 

$$\psi_n^2(\omega_1, \ldots,\omega_{n-1},\omega_n)=$$ $$\int\limits_{\Omega_n^0}\chi_{\Omega_n^-}(\omega_1, \ldots,\omega_{n-1}, \omega_n^1)\alpha_n(\{\omega_1, \ldots,\omega_{n-1},\omega_n^1\};\{\omega_1, \ldots,\omega_{n-1},\omega_n^2\}) \times $$
\begin{eqnarray}\label{vitasja13}
\frac{\Delta S_n^-(\omega_1, \ldots,\omega_{n-1}, \omega_n^1)}{V_n(\omega_1, \ldots,\omega_{n-1}, \omega_n^1,  \omega_n^2)}d P_n^0(\omega_n^1), \quad  (\omega_1, \ldots,\omega_{n-1}) \in \Omega_{n-1}.
\end{eqnarray} 
\end{leme}
\begin{proof} Due to Lemma \ref{witka0} conditions, the set of the strictly positive 
bounded random values $\alpha_n(\{\omega\}_n^1; \{\omega\}_n^2),$ $  n=\overline{1,N}, $ satisfying the conditions (\ref{1vitasja7}) - (\ref{3vitasja7}),  is a non empty set.
 We prove  Lemma \ref{witka1} by induction  down. 
Let us denote
\begin{eqnarray}\label{13vitasja13}
 \mu_N^{(\omega_1, \ldots,\omega_{N-1}, \omega_N)}(A)=\chi_{A}(\omega_1, \ldots,\omega_N).
\end{eqnarray}
Then,
$$
\int\limits_{\Omega_N^0}\psi_N(\omega_1, \ldots,\omega_{N-1},\omega_N)\mu_N^{(\omega_1, \ldots,\omega_{N-1}, \omega_N)}(A) d P_N^0(\omega_N)=
$$
$$\int\limits_{\Omega_N^0}\chi_{\Omega_N^-}(\omega_1, \ldots,\omega_{N-1}, \omega_N)\psi_N^1(\omega_1, \ldots,\omega_{N-1},\omega_N) \mu_N^{(\omega_1, \ldots,\omega_{N-1}, \omega_N)}(A)d P_N^0(\omega_ N)+   $$

$$\int\limits_{\Omega_N^0}\chi_{\Omega_N^+}(\omega_1, \ldots,\omega_{N-1}, \omega_N)\psi_N^2(\omega_1, \ldots,\omega_{N-1},\omega_N) \mu_N^{(\omega_1, \ldots,\omega_{N-1}, \omega_N)}(A)d P_N^0(\omega_N)=   $$

$$\int\limits_{\Omega_N^0}\chi_{\Omega_N^-}(\omega_1, \ldots,\omega_{N-1}, \omega_N^1)\psi_N^1(\omega_1, \ldots,\omega_{N-1},\omega_N^1) \mu_N^{(\omega_1, \ldots,\omega_{N-1}, \omega_N^1)}(A)d P_N^0(\omega_N^1)+   $$
\begin{eqnarray}\label{11vitasja14}
\int\limits_{\Omega_N^0}\chi_{\Omega_N^+}(\omega_1, \ldots,\omega_{N-1}, \omega_N^2)\psi_N^2(\omega_1, \ldots,\omega_{N-1},\omega_N^2) \mu_N^{(\omega_1, \ldots,\omega_{N-1}, \omega_N^2)}(A)d P_N^0(\omega_N^2).   \end{eqnarray}
Substituting $\psi_N^1(\omega_1, \ldots,\omega_{N-1},\omega_N^1), \  \psi_N^2(\omega_1, \ldots,\omega_{N-1},\omega_N^2) $ into (\ref{11vitasja14}), we obtain
$$
\int\limits_{\Omega_N^0}\psi_N(\omega_1, \ldots,\omega_{N-1},\omega_N) \mu_N^{(\omega_1, \ldots,\omega_{N-1}, \omega_N)}(A)d P_N^0(\omega_N)=$$
$$\int\limits_{\Omega_N^0\times \Omega_N^0} \chi_{\Omega_N^-}(\omega_1, \ldots,\omega_{N-1}, \omega_N^1) \chi_{\Omega_N^+}(\omega_1, \ldots,\omega_{N-1}, \omega_N^2)\times$$

$$\alpha_N(\{\omega_1, \ldots,\omega_{N-1},\omega_N^1\};\{\omega_1, \ldots,\omega_{N-1}, \omega_N^2\})\times $$ 
$$ \left[\frac{\Delta S_N^+(\omega_1, \ldots,\omega_{N-1}, \omega_N^2)}{V_N(\omega_1, \ldots,\omega_{N-1}, \omega_N^1,  \omega_N^2)}\mu_{N}^{(\omega_1, \ldots,\omega_{N-1}, \omega_N^1)}(A)+ \right.$$
$$\left. \frac{\Delta S_N^-(\omega_1, \ldots,\omega_{N-1}, \omega_N^1)}{V_N(\omega_1, \ldots,\omega_{N-1}, \omega_N^1,  \omega_N^2)}\mu_{N}^{(\omega_1, \ldots,\omega_{N-1}, \omega_N^2)}(A)\right]dP_N^0(\omega_N^1)dP_N^0(\omega_N^2)=$$
\begin{eqnarray}\label{15vitasja15}
 \mu_{N-1}^{(\omega_1, \ldots,\omega_{N-1})}(A).
\end{eqnarray}
Suppose that we are proved that
$$ \mu_n^{(\omega_1, \ldots,\omega_{n-1}, \omega_n)}(A)=$$
\begin{eqnarray}\label{vitasja14}
\int\limits_{\prod\limits_{i=n+1}^N \Omega_i^0}
\prod\limits_{i=n+1}^N\psi_i(\omega_1, \ldots,\omega_i)\chi_{A}(\omega_1, \ldots,\omega_N)\prod\limits_{i=n+1}^Nd P_i^0(\omega_i).
\end{eqnarray}
Let us calculate
$$
\int\limits_{\Omega_n^0}\psi_n(\omega_1, \ldots,\omega_{n-1},\omega_n) \mu_n^{(\omega_1, \ldots,\omega_{n-1}, \omega_n)}(A)d P_n^0(\omega_n)=
$$
$$\int\limits_{\Omega_n^0}\chi_{\Omega_n^-}(\omega_1, \ldots,\omega_{n-1}, \omega_n)\psi_n^1(\omega_1, \ldots,\omega_{n-1},\omega_n) \mu_n^{(\omega_1, \ldots,\omega_{n-1}, \omega_n)}(A)d P_n^0(\omega_n)+   $$

$$\int\limits_{\Omega_n^0}\chi_{\Omega_n^+}(\omega_1, \ldots,\omega_{n-1}, \omega_n)\psi_n^2(\omega_1, \ldots,\omega_{n-1},\omega_n) \mu_n^{(\omega_1, \ldots,\omega_{n-1}, \omega_n)}(A)d P_n^0(\omega_n)=   $$

$$\int\limits_{\Omega_n^0}\chi_{\Omega_n^-}(\omega_1, \ldots,\omega_{n-1}, \omega_n^1)\psi_n^1(\omega_1, \ldots,\omega_{n-1},\omega_n^1) \mu_n^{(\omega_1, \ldots,\omega_{n-1}, \omega_n^1)}(A)d P_n^0(\omega_n^1)+   $$
\begin{eqnarray}\label{10vitasja14}
\int\limits_{\Omega_n^0}\chi_{\Omega_n^+}(\omega_1, \ldots,\omega_{n-1}, \omega_n^2)\psi_n^2(\omega_1, \ldots,\omega_{n-1},\omega_n^2) \mu_n^{(\omega_1, \ldots,\omega_{n-1}, \omega_n^2)}(A)d P_n^0(\omega_n^2).   \end{eqnarray}
Substituting $\psi_n^1(\omega_1, \ldots,\omega_{n-1},\omega_n^1),   \psi_n^2(\omega_1, \ldots,\omega_{n-1},\omega_n^2) $ into (\ref{10vitasja14}), we obtain
$$
\int\limits_{\Omega_n^0}\psi_n(\omega_1, \ldots,\omega_{n-1},\omega_n) \mu_n^{(\omega_1, \ldots,\omega_{n-1}, \omega_n)}(A)d P_n^0(\omega_n)=$$
$$\int\limits_{\Omega_n^0\times \Omega_n^0} \chi_{\Omega_n^-}(\omega_1, \ldots,\omega_{n-1}, \omega_n^1) \chi_{\Omega_n^+}(\omega_1, \ldots,\omega_{n-1}, \omega_n^2)\times$$

$$\alpha_n(\{\omega_1, \ldots, \omega_{n-1},\omega_n^1\};\{\omega_1, \ldots,\omega_{n-1}, \omega_n^2\}) \times $$ 
$$\left[\frac{\Delta S_n^+(\omega_1, \ldots,\omega_{n-1}, \omega_n^2)}{V_n(\omega_1, \ldots,\omega_{n-1}, \omega_n^1,  \omega_n^2)}\mu_{n}^{(\omega_1, \ldots,\omega_{n-1}, \omega_n^1)}(A)+ \right.$$
\begin{eqnarray}\label{tinvitasja15}
\left. \frac{\Delta S_n^-(\omega_1, \ldots,\omega_{n-1}, \omega_n^1)}{V_n(\omega_1, \ldots,\omega_{n-1}, \omega_n^1,  \omega_n^2)}\mu_{n}^{(\omega_1, \ldots,\omega_{n-1}, \omega_n^2)}(A)\right]dP_n^0(\omega_n^1)dP_n^0(\omega_n^2).
\end{eqnarray}
From the recurrent relations (\ref{vitasja6}) - (\ref{pupecvitasja5}), we have
$$\mu_{n-1}^{(\omega_1, \ldots,\omega_{n-1})}(A)=\int\limits_{\Omega_n^0\times \Omega_n^0} \chi_{\Omega_n^-}(\omega_1, \ldots,\omega_{n-1}, \omega_n^1) \chi_{\Omega_n^+}(\omega_1, \ldots,\omega_{n-1}, \omega_n^2)\times$$

$$\alpha_n(\{\omega_1, \ldots,\omega_{n-1}, \omega_n^1\};\{\omega_1, \ldots,\omega_{n-1}, \omega_n^2\})\times $$ 
$$ \left[\frac{\Delta S_n^+(\omega_1, \ldots,\omega_{n-1},\omega_{n-1}, \omega_n^2)}{V_n(\omega_1, \ldots,\omega_{n-1}, \omega_n^1,  \omega_n^2)}\mu_{n}^{(\omega_1, \ldots,\omega_{n-1}, \omega_n^1)}(A)+ \right.$$
\begin{eqnarray}\label{vitasja15}
\left. \frac{\Delta S_n^-(\omega_1, \ldots,\omega_{n-1}, \omega_n^1)}{V_n(\omega_1, \ldots,\omega_{n-1}, \omega_n^1,  \omega_n^2)}\mu_{n}^{(\omega_1, \ldots,\omega_{n-1}, \omega_n^2)}(A)\right]dP_n^0(\omega_n^1)dP_n^0(\omega_n^2), \quad n=\overline{1,N}.
\end{eqnarray}
From the last equality, we have
\begin{eqnarray}\label{11vitasja15}
\mu_{n-1}^{(\omega_1, \ldots,\omega_{n-1})}=
\int\limits_{\Omega_n^0}\psi_n(\omega_1, \ldots,\omega_{n-1},\omega_n) \mu_n^{(\omega_1, \ldots,\omega_{n-1}, \omega_n)}(A)d P_n^0(\omega_n), \quad n=\overline{1,N}.
\end{eqnarray}
Substituting into (\ref{11vitasja15}) the induction supposition (\ref{vitasja14}),
we obtain
$$ \mu_{n-1}^{(\omega_1, \ldots,\omega_{n-1})}(A)=$$
\begin{eqnarray}\label{vitasja16}
 \int\limits_{\prod\limits_{i=n}^N \Omega_i^0}
\prod\limits_{i=n}^N\psi_i(\omega_1, \ldots,\omega_i)\chi_{A}(\omega_1, \ldots,\omega_N)\prod\limits_{i=n}^Nd P_i^0(\omega_i).
\end{eqnarray}
To prove that $\mu_0(\Omega_N)=1,$ let us prove  the equality
\begin{eqnarray}\label{vitasja23}
\int\limits_{\Omega_n^0}\psi_n(\omega_1, \ldots,\omega_n)d P_n^0(\omega_n)=1,  \quad  (\omega_1, \ldots,\omega_{n-1}) \in \Omega_{n-1}, \quad 
 \quad n=\overline{1,N}.
\end{eqnarray}
We  have
$$\int\limits_{\Omega_n^0}\psi_n(\omega_1, \ldots,\omega_n)d P_n^0(\omega_n)=$$

$$\int\limits_{\Omega_n^0}\int\limits_{\Omega_n^0}
\chi_{\Omega_n^-}(\omega_1, \ldots,\omega_{n-1}, \omega_n^1) \chi_{\Omega_n^+}(\omega_1, \ldots,\omega_{n-1}, \omega_n^2)\times $$
$$\alpha_n(\{\omega_1, \ldots,\omega_{n-1},\omega_n^1\};\{\omega_1, \ldots,\omega_{n-1}, \omega_n^2\}) \times $$ $$\left[\frac{\Delta S_n^+(\omega_1, \ldots,\omega_{n-1}, \omega_n^2)}{V_n(\omega_1, \ldots,\omega_{n-1}, \omega_n^1,  \omega_n^2)}+\right.$$
$$\left.\frac{\Delta S_n^-(\omega_1, \ldots,\omega_{n-1}, \omega_n^1)}{V_n(\omega_1, \ldots,\omega_{n-1}, \omega_n^1,  \omega_n^2)}\right]d P_n^0(\omega_n^1)d P_n^0(\omega_n^2)=$$
$$\int\limits_{\Omega_n^0}\int\limits_{\Omega_n^0}
\chi_{\Omega_n^-}(\omega_1, \ldots,\omega_{n-1},\omega_n^1) \chi_{\Omega_n^+}(\omega_1, \ldots,\omega_{n-1},\omega_n^2)\times$$
\begin{eqnarray}\label{vitasja25}
\alpha_n(\{\omega_1, \ldots,\omega_{n-1},\omega_n^1\};\{\omega_1, \ldots,\omega_{n-1},\omega_n^2\})d P_n^0(\omega_n^1)d P_n^0(\omega_n^2)=1.
\end{eqnarray}
The last equality follows from the fact that the set of random values $\alpha_n(\{\omega_1\}_n^1;\{\omega_1\}_n^2), $ $ \ n=\overline{1,N},$ satisfies the condition (\ref{3vitasja7}).
The equalities (\ref{vitasja23}) proves that every  measure (\ref{vitasja10}), defined by the set of random values  $\alpha_n(\{\omega_1^1, \ldots,\omega_n^1\};\{\omega_1^2, \ldots,\omega_n^2\}),$ $  n=\overline{1,N},$  satisfying the conditions  (\ref{1vitasja7}),   (\ref{3vitasja7}), is a probability measure being equivalent to the measure $P_N.$

This proves Lemma \ref{witka1}
\end{proof}

\begin{note1}\label{pups_witka3}
Due to the equality (\ref{vitasja23}), the contraction of measure $\mu_0(A), A \in {\cal F}_N,$ on the $\sigma$-algebra  ${\cal F}_n$ of filtration we denote by  $\mu_0^n.$ If  $ A$ belongs to the $\sigma$-algebra   ${\cal F}_n$ of filtration, then $A=B\times \prod\limits_{i=n+1}^N\Omega_i^0,$ where $B$ belongs to the $\sigma$-algebra ${\cal F}_n$ of the measurable space $\{\Omega_n, {\cal F}_n\},$ therefore, for this contraction we obtain the formula
\begin{eqnarray}\label{pups_witka4}
\mu_0^n(A)=\int\limits_{\Omega_n}\prod\limits_{i=1}^n\psi_i(\omega_1, \ldots,\omega_i)\chi_{B}(\omega_1, \ldots,\omega_n)\prod\limits_{i=1}^n d P_i^0(\omega_i), \quad  B \in {\cal F}_n.  
\end{eqnarray}
Further, we also use the probability spaces $\{\Omega_n, {\cal F}_n, \mu_0^n\}, \ n=\overline{1,N},$ where under the measure  $ \mu_0^n(B), B \in {\cal F}_n,$ we understand the measure, given by the formula 
\begin{eqnarray}\label{pups_witka5}
\mu_0^n(B)=\int\limits_{\Omega_n}\prod\limits_{i=1}^n\psi_i(\omega_1, \ldots,\omega_i)\chi_{B}(\omega_1, \ldots,\omega_n)\prod\limits_{i=1}^n d P_i^0(\omega_i), \quad  B \in {\cal F}_n. 
\end{eqnarray}
\end{note1}

\begin{note1}\label{witka3}
Assume that for $\alpha_n(\{\omega_1^1, \ldots,\omega_{n-1}^1,\omega_n^1\};\{\omega_1^2, \ldots,\omega_{n-1}^2,\omega_n^2\}),$ constructed in Lemma \ref{witka0}, the inequalities  
 $$0< c_n \leq \alpha_n(\{\omega_1^1, \ldots,\omega_{n-1}^1,\omega_n^1\};\{\omega_1^2, \ldots,\omega_{n-1}^2,\omega_n^2\})\leq C_n<\infty, $$
are true. Suppose that   the conditions
\begin{eqnarray}\label{vitasja28} 
\Delta S_n^-(\omega_1, \ldots,\omega_{n-1}, \omega_n) \leq B_n< \infty, \quad n=\overline{1,N},  
\end{eqnarray}
are valid, where $c_n, \ C_n, \ B_n $ are constant,
then
the set of equivalent  measures to the measure $P_N,$  described in Lemma \ref{witka1}, is nonempty one. 
\end{note1}
\begin{proof}
Due to  Lemma \ref{witka1} conditions,   the equality (\ref{1vitasja7}) is true.
Further,
 $$ \int\limits_{\Omega_n^0}\int\limits_{\Omega_n^0}\chi_{\Omega_n^-}(\omega_1^1, \ldots,\omega_{n-1}^1, \omega_n^1) \chi_{\Omega_n^+}(\omega_1^2, \ldots,\omega_{n-1}^2, \omega_n^2) \times $$ $$\alpha_n(\{\omega_1^1, \ldots,\omega_{n-1}^1, \omega_n^1\};\{\omega_1^2, \ldots, \omega_{n-1}^2,\omega_n^2\})\times $$
$$
\frac{\Delta S_n^+(\omega_1, \ldots,\omega_{n-1}, \omega_n^2) \Delta S_n^-(\omega_1, \ldots,\omega_{n-1}, \omega_n^1)}{V_n(\omega_1, \ldots,\omega_{n-1}, \omega_n^1,  \omega_n^2)}dP_n^0(\omega_n^1) dP_n^0(\omega_n^2)\leq B_n,$$

$$(\{\omega_1^1, \ldots,\omega_{n-1}^1\};  \{\omega_1^2, \ldots,\omega_{n-1}^2\}) \in \Omega_{n-1}\times \Omega_{n-1},\quad  (\omega_1, \ldots,\omega_{n-1}) \in \Omega_{n-1},$$

$$  \int\limits_{\Omega_n^0\times \Omega_n^0}\chi_{\Omega_n^-}(\omega_1^1, \ldots,\omega_{n-1}^1, \omega_n^1)\chi_{\Omega_n^+}(\omega_1^2, \ldots,\omega_{n-1}^2, \omega_n^2)\times $$

$$\alpha_n(\{\omega_1^1, \ldots,\omega_{n-1}^1, \omega_n^1\};\{\omega_1^2, \ldots, \omega_{n-1}^2,\omega_n^2\})dP_n^0(\omega_n^1) dP_n^0(\omega_n^2)=1,$$
\begin{eqnarray}\label{vitasja32}
(\{\omega_1^1, \ldots,\omega_{n-1}^1\};  \{\omega_1^2, \ldots,\omega_{n-1}^2\}) \in \Omega_{n-1}\times \Omega_{n-1}.
\end{eqnarray}

The last  inequality and the equality (\ref{vitasja32}) means that the  conditions (\ref{1vitasja7}) - (\ref{3vitasja7}) are satisfied.
Note \ref{witka3} is proved.
\end{proof}

For the nonnegative random value  $\alpha_n(\{\omega_1^1, \ldots,\omega_n^1\};\{\omega_1^2, \ldots,\omega_n^2\}),$ given
  on the measurable space $\{\Omega_n^-\times\Omega_n^+,{\cal F}_n^-\times {\cal F}_n^+\},$   ${\cal F}_n^-={\cal F}_n \cap\Omega_n^- ,$   ${\cal F}_n^+={\cal F}_n \cap\Omega_n^+, 
n=\overline{1,N},$
let us define the integral for the nonnegative  random value $f_N(\omega_1, \ldots,\omega_N)$ relative to the measure $\mu_0(A)$  using the recurrent relations
$$\mu_{n-1}^{f_N}(\omega_1, \ldots,\omega_{n-1})=$$ $$\int\limits_{\Omega_n^0\times \Omega_n^0} \chi_{\Omega_n^-}(\omega_1, \ldots,\omega_{n-1},\omega_n^1) \chi_{\Omega_n^+}(\omega_1, \ldots,\omega_{n-1},\omega_n^2)\times$$

$$\alpha_n(\{\omega_1, \ldots,\omega_{n-1},\omega_n^1\};\{\omega_1, \ldots,\omega_{n-1}, \omega_n^2\}) \times $$ $$\left[\frac{\Delta S_n^+(\omega_1, \ldots,\omega_{n-1}, \omega_n^2)}{V_n(\omega_1, \ldots,\omega_{n-1}, \omega_n^1,  \omega_n^2)}\mu_{n}^{f_N}(\omega_1, \ldots,\omega_{n-1}, \omega_n^1)+ \right.$$
\begin{eqnarray}\label{vitasja18}
\left. \frac{\Delta S_n^-(\omega_1, \ldots,\omega_{n-1}, \omega_n^1)}{V_n(\omega_1, \ldots,\omega_{n-1}, \omega_n^1,  \omega_n^2)}\mu_{n}^{f_N}(\omega_1, \ldots,\omega_{n-1}, \omega_n^2)\right]d P_n^0(\omega_n^1)d P_n^0(\omega_n^2), \quad n=\overline{1,N},
\end{eqnarray}

$$\mu_{N-1}^{f_N}(\omega_1, \ldots,\omega_{N-1})=\int\limits_{\Omega_N^0\times \Omega_N^0} \chi_{\Omega_N^-}(\omega_1, \ldots, \omega_{N-1},\omega_N^1) \chi_{\Omega_N^+}(\omega_1, \ldots, \omega_{N-1}, \omega_N^2)\times$$
$$\alpha_N(\{\omega_1, \ldots, \omega_{N-1},\omega_N^1\};\{\omega_1, \ldots, \omega_{N-1}, \omega_N^2\})\times$$ 
$$ \left[\frac{\Delta S_N^+(\omega_1, \ldots,\omega_{N-1}, \omega_N^2)}{V_N(\omega_1, \ldots,\omega_{N-1}, \omega_N^1,  \omega_N^2)}f_{N}(\omega_1, \ldots,\omega_{N-1}, \omega_N^1)+ \right.$$
\begin{eqnarray}\label{vitasja19}
\left. \frac{\Delta S_N^-(\omega_1, \ldots,\omega_{N-1}, \omega_N^1)}{V_N(\omega_1, \ldots,\omega_{N-1}, \omega_N^1,  \omega_N^2)}f_N(\omega_1, \ldots,\omega_{N-1}, \omega_N^2)\right]dP_N^0(\omega_N^1)dP_N^0(\omega_N^2).
\end{eqnarray}

From the formula (\ref{vitasja10})  of  Lemma \ref{witka1},  it follows that
\begin{eqnarray}\label{vitasja21}
E^{\mu_0}f_N=\int\limits_{\Omega_N}\prod\limits_{n=1}^N\psi_n(\omega_1, \ldots,\omega_n)f_N(\omega_1, \ldots,\omega_{N-1}, \omega_N)\prod\limits_{i=1}^Nd P_i^0(\omega_i)
\end{eqnarray}
for every nonnegative  ${\cal F}_ N$-measurable random value $f_N(\omega_1, \ldots,\omega_{N-1}, \omega_N).$

\begin{te}\label{witka2}
Suppose that  the conditions of Lemma \ref{witka0} 
are true.
Then, the set of nonnegative random values  $\alpha_n(\{\omega\}_n^1; \{\omega\}_n^2), n=\overline{1,N},$ satisfying the conditions 
$$E^{\mu_0}|\Delta S_n(\omega_1, \ldots,\omega_{n-1}, \omega_n)|=$$
\begin{eqnarray}\label{vitasja22}
\int\limits_{\Omega_N}\prod\limits_{i=1}^N\psi_i(\omega_1, \ldots,\omega_i) |\Delta S_n(\omega_1, \ldots,\omega_{n-1}, \omega_n)| \prod\limits_{i=1}^Nd P_i^0(\omega_i)<\infty, \quad n=\overline{1,N},
\end{eqnarray}
is a nonempty one and the convex linear span of the  set of  measures (\ref{vitasja10}), defined by the random values  $\alpha_n(\{\omega_1^1, \ldots,\omega_n^1\};\{\omega_1^2, \ldots,\omega_n^2\}),$ $  n=\overline{1,N},$  satisfying the conditions (\ref{vitasja22}), is a set of  martingale measures being equivalent to the measure $P_N.$ 
\end{te}
\begin{proof}

Taking into account the equality (\ref{vitasja23}), the conditions (\ref{vitasja22})  can be written in the form

$$\int\limits_{\Omega_N}\prod\limits_{i=1}^N\psi_i(\omega_1, \ldots,\omega_i) |\Delta S_n(\omega_1, \ldots,\omega_{n-1}, \omega_n)| \prod\limits_{i=1}^Nd P_i^0(\omega_i)=$$
$$\int\limits_{\Omega_n}\prod\limits_{i=1}^n\psi_i(\omega_1, \ldots,\omega_i) |\Delta S_n(\omega_1, \ldots,\omega_{n-1}, \omega_n)| \prod\limits_{i=1}^n d P_i^0(\omega_i)=$$
$$2 \int\limits_{\Omega_{n-1}}\prod\limits_{i=1}^{n-1}\psi_i(\omega_1, \ldots,\omega_i) 
 \int\limits_{\Omega_n^0}\int\limits_{\Omega_n^0}\chi_{\Omega_n^-}(\omega_1, \ldots,\omega_{n-1}, \omega_n^1) \chi_{\Omega_n^+}(\omega_1, \ldots,\omega_{n-1}, \omega_n^2) \times $$ $$\alpha_n(\{\omega_1, \ldots,\omega_{n-1}, \omega_n^1\};\{\omega_1, \ldots, \omega_{n-1},\omega_n^2\})\times $$
$$\frac{\Delta S_n^+(\omega_1, \ldots,\omega_{n-1}, \omega_n^2) \Delta S_n^-(\omega_1, \ldots,\omega_{n-1}, \omega_n^1)}{V_n(\omega_1, \ldots,\omega_{n-1}, \omega_n^1,  \omega_n^2)} \times$$
\begin{eqnarray}\label{vitunjazajka2}
d P_n^0(\omega_n^1)  d P_n^0(\omega_n^2)  \prod\limits_{i=1}^{n-1} d P_i^0(\omega_i), \quad n=\overline{1,N}.
\end{eqnarray}
Since the conditions  of Lemma \ref{witka0} are true, then the the set of  bounded  random values  $\alpha_n(\{\omega_1^1, \ldots,\omega_n^1\};\{\omega_1^2, \ldots,\omega_n^2\}),$ $  n=\overline{1,N},$ constructed in Lemma \ref{witka0},
satisfy the conditions (\ref{1vitasja7}) - (\ref{3vitasja7}).

From the equality (\ref{vitunjazajka2}) for the set of bounded random values   $\alpha_n(\{\omega\}_n^1;\{\omega\}_n^2),$ $  n=\overline{1,N},$ satisfying the conditions (\ref{1vitasja7}) - (\ref{3vitasja7}),  we obtain the inequality
$$\int\limits_{\Omega_N}\prod\limits_{i=1}^N\psi_i(\omega_1, \ldots,\omega_i) |\Delta S_n(\omega_1, \ldots,\omega_{n-1}, \omega_n)| \prod\limits_{i=1}^Nd P_i^0(\omega_i) \leq$$
\begin{eqnarray}\label{vitunjazajka1}
 C \int\limits_{\Omega_N}\Delta S_n^-(\omega_1, \ldots,\omega_{n-1}, \omega_n^1)d P_N<\infty, \quad n=\overline{1,N},
\end{eqnarray}
for a certain constant $0<  C<\infty.$ This proves that the set of nonnegative random values $\alpha_n(\{\omega_1^1, \ldots,\omega_n^1\};\{\omega_1^2, \ldots,\omega_n^2\}),$ $  n=\overline{1,N},$  satisfying the conditions (\ref{vitasja22}), is a non empty set.

Let us prove that
$$\int\limits_{\Omega_n^0}\psi_n(\omega_1, \ldots,\omega_n)\Delta S_n(\omega_1, \ldots,\omega_n) d P_n^0(\omega_n)=0,$$
\begin{eqnarray}\label{vitasja24}
  (\omega_1, \ldots,\omega_{n-1}) \in \Omega_{n-1}, \quad  n=\overline{1,N}.
\end{eqnarray}
Really, 
$$\int\limits_{\Omega_n^0}\psi_n(\omega_1, \ldots,\omega_n)\Delta S_n(\omega_1, \ldots,\omega_n) d P_n^0(\omega_n)=$$
$$\int\limits_{\Omega_n^0}\int\limits_{\Omega_n^0}
\chi_{\Omega_n^-}(\omega_1, \ldots,\omega_{n-1}, \omega_n^1) \chi_{\Omega_n^+}(\omega_1, \ldots,\omega_{n-1}, \omega_n^2)\times $$
$$\alpha_n(\{\omega_1, \ldots,\omega_{n-1}, \omega_n^1\};\{\omega_1, \ldots,\omega_{n-1}, \omega_n^2\}) \times $$ $$\left[- \frac{ \Delta S_n^+(\omega_1, \ldots,\omega_{n-1}, \omega_n^2)}{V_n(\omega_1, \ldots,\omega_{n-1}, \omega_n^1,  \omega_n^2)} \Delta S_n^-(\omega_1, \ldots,\omega_{n-1}, \omega_n^1)+\right.$$
\begin{eqnarray}\label{vitasja26}
\left.\frac{ \Delta S_n^-(\omega_1, \ldots,\omega_{n-1}, \omega_n^1)}{V_n(\omega_1, \ldots,\omega_{n-1}, \omega_n^1,  \omega_n^2)}\Delta S_n^+(\omega_1, \ldots,\omega_{n-1}, \omega_n^2)\right]d P_n^0(\omega_n^1)d P_n^0(\omega_n^2)=0, 
\end{eqnarray}
due to the condition (\ref{2vitasja7}).

To complete the proof of Theorem \ref{witka2},  let $ A$ belongs to the filtration $ {\cal F}_{n-1},$ then $A=B\times\prod\limits_{i=n}^N\Omega_i^0,$ where $B$ belongs to the $\sigma$-algebra ${\cal F}_{n-1}$ of the measurable space $\{\Omega_{n-1}, {\cal F}_{n-1}\}.$ Taking into account the equality (\ref{vitasja25}), (\ref{vitasja26}), we have, due to Foubini theorem,

$$\int\limits_{\Omega_N}\prod\limits_{i=1}^N\psi_i(\omega_1, \ldots,\omega_i)\chi_{A}(\omega_1, \ldots, \omega_N)\Delta S_n(\omega_1, \ldots, \omega_n)\prod\limits_{i=1}^N d P_i^0(\omega_i)=$$

$$ \int\limits_{\Omega_n}\prod\limits_{i=1}^n\psi_i(\omega_1, \ldots,\omega_i)
\chi_{B}(\omega_1, \ldots, \omega_{n-1})\Delta S_n(\omega_1, \ldots, \omega_n)\prod\limits_{i=1}^n d P_i^0(\omega_i)=$$
$$ \int\limits_{\Omega_{n-1}}\prod\limits_{i=1}^{n-1}\psi_i(\omega_1, \ldots,\omega_i)
\chi_{B}(\omega_1, \ldots, \omega_{n-1})\prod\limits_{i=1}^{n-1} d P_i^0(\omega_i)\times $$
\begin{eqnarray}\label{vitasja27} 
\int\limits_{\Omega_n^0}\psi_i(\omega_1, \ldots,\omega_n) \Delta S_n(\omega_1, \ldots, \omega_n)d P_n^0(\omega_n)=0.
\end{eqnarray}
The last means that $E^{\mu_0}\{S_n(\omega_1, \ldots, \omega_n)|{\cal F}_{n-1}\}=S_{n-1}(\omega_1, \ldots, \omega_{n-1}).$
Since every measure, belonging to the convex linear span of the measures considered above,  is a finite  sum of such measures, then it is a martingale measure being equivalent to the measure $P_N.$ 
Theorem \ref{witka2} is proved.
\end{proof}
Our aim is to describe this convex span of martingale measures in particular cases.

\section{Inequalities for the nonnegative random values.}

In this section, we prove some inequalities, which will be very useful for to prove optional decomposition for super-martingale relative to all martingale measures.
First, we prove an integral inequality for a nonnegative random variable under the fulfillment of the inequality for this nonnegative random variable with respect to the constructed family of measures $\mu_0(A).$ Further, using this integral inequality for the non-negative random variable, a pointwise system of inequalities is obtained for this non-negative random variable for a particular  case. After that, the pointwise system of inequalities  is obtained for the non-negative random variable in the general case. Then, using the resulting pointwise system of inequalities, the inequality is established for this non-negative random variable whose right-hand side is such that its conditional mathematical expectation is equal to one.

\begin{defin}\label{vitnick1}
Let $\{\Omega_1, {\cal F}_1\}$ be a measurable space. The  decomposition $A_{n,k}, \  n,k=\overline{1, \infty},$ of the space $\Omega_1$  we call exhaustive one, if the following conditions are valid:\\
1) $A_{n,k} \in {\cal F}_1,$ \ $ A_{n,k}\cap A_{n,s}=\emptyset, \ k\neq s,$ \ $\bigcup\limits_{k=1}^\infty A_{n,k}=\Omega_1, \ n=\overline{1, \infty};$\\
2) the $(n+1)$-th decomposition is a sub-decomposition of the $n$-th one, that is, for every $j,$ $A_{n+1,j} \subseteq  A_{n,k}$ for a certain $k=k(j);$\\
3) the minimal $\sigma$-algebra containing all $A_{n,k}, \ n,k=\overline{1, \infty},$ coincides with ${\cal F}_1.$
\end{defin}
\begin{leme}\label{2vitanick1tin1}
Let $\{\Omega_1, {\cal F}_1\}$ be  a measurable space with a complete separable metric space  $\Omega_1$ and Borel $\sigma$-algebra  ${\cal F}_1$ on it. Then,  
$\{\Omega_1, {\cal F}_1\}$ has an exhaustive decomposition.
\end{leme}
The proof of Lemma \ref{2vitanick1tin1} see, for example,   in \cite{GoncharSimon},
\cite{GoncharNick1}.

For the proof of integral  inequalities,  
 we cannot require  the fulfillment  for the random values $\alpha_n(\{\omega_1^1, \ldots,\omega_n^1\};\{\omega_1^2, \ldots,\omega_n^2\}),$ $  n=\overline{1,N},$  the condition (\ref{2vitasja7}) in the Lemma \ref{witka4}.

\begin{leme}\label{witka4}
Suppose that $\Omega_n^0$ is a complete separable metric space, 
${\cal F}_n^0$ is a corresponding Borel $\sigma$-algebra on  $\Omega_n^0, \ n=\overline{1,N},$ and the conditions of Lemma \ref{witka0} are valid. 
If, on the probability space  
$\{ \Omega_{n-1},    {\cal F}_{n-1}, \mu_0^{n-1}\},$ for each  $B \in {\cal F}_{n-1},$ $\mu_0^{n-1}(B)>0,$   the nonnegative  random value $f_n(\omega_1, \ldots,\omega_{n-1}, \omega_n)$ satisfies
the inequality
\begin{eqnarray}\label{vitasja33}
\frac{1}{\mu_0^{n-1}(B)}\int\limits_{B}\int\limits_{\Omega_n^0} \prod\limits_{i=1}^n\psi_i(\omega_1, \ldots,\omega_i)f_n(\omega_1, \ldots, \omega_n)\prod\limits_{i=1}^n d P_i^0(\omega_i)\leq 1, \quad B \in {\cal F}_ {n-1},
\end{eqnarray}
 then the inequality

$$\int\limits_{\Omega_n^0}
\psi_n(\omega_1, \ldots,\omega_n)f_n(\omega_1, \ldots, \omega_n) d P_n^0(\omega_n)\leq 1,  $$
\begin{eqnarray}\label{1vitasja33}
 \{\omega_1, \ldots, \omega_{n-1}\} \in \Omega_{n-1}, \ n=\overline{1, N},
\end{eqnarray}
is true almost everywhere relative to the measure $P_{n-1}.$
\end{leme}
\begin{proof}
The metric space $\Omega_{n-1}$   is a complete separable metric space  with the metric $\rho(x,y)=\sum\limits_{i=1}^{n-1}\rho_i(x_i,y_i), $  where $x=(x_1, \ldots, x_{n-1}),  y=(y_1, \ldots, y_{n-1}) \in \Omega_{n-1},$  \ $\ (x_i, y_i) \in \Omega_i^0, $ \  $\rho_i(x_i,y_i)$ is a metric in $\Omega_{i}^0.$ This means that the metric space $\Omega_{n-1}$ has an  exhaustive decomposition$\{B_{m k}\}_{m,k=1}^\infty.$ Suppose that $(\omega_1, \ldots, \omega_{n-1}) \in B_{m, k}$ for a certain $k,$ depending on $m,$ and there exists  an infinite number of $m$ for which $ \mu_0^{n-1}( B_{m, k})>0.$ 
On the probability space  
$\{ \Omega_{n-1},    {\cal F}_{n-1}, \mu_0^{n-1}\},$   for every integrable finite valued random value $ \varphi_{n-1}(\omega_1, \ldots,\omega_{n-1})$  the sequence 
$E^{ \mu_0^{n-1}}\{ \varphi_{n-1}(\omega_1,\ldots, \omega_{n-1})|\bar {\cal F}_m\}$ converges to $ \varphi_{n-1}(\omega_1, \ldots, \omega_{n-1})$ with probability one, as $m \to \infty,$ since it is a regular martingale. Here, we denoted  $\bar {\cal F}_m$  the $\sigma$-algebra, generated by the sets $B_{m,k}, k=\overline{1,\infty}.$

 It is evident that for those $B_{m,k},$ for which $ \mu_0^{n-1}(B_{m,k})\neq 0,$

$$E^{ \mu_0^{n-1}}\{ \varphi_{n-1}(\omega_1,\ldots, \omega_n)|\bar {\cal F}_m\}=$$
\begin{eqnarray}\label{mykolatina1}
\frac{\int\limits_{B_{m,k}} \varphi_{n-1}(\omega_1,\ldots, \omega_{n-1}) d \mu_0^{n-1}}{ \mu_0^{n-1}(B_{m,k})}, \quad (\omega_1,\ldots, \omega_n) \in B_{m,k}.
\end{eqnarray}

 Denote $A_m=A_m(\omega_1,\ldots, \omega_{n-1})$ those sets  $B_{m,k}$   for which  $(\omega_1,\ldots, \omega_n) \in B_{m,k}$ for a certain $k,$ depending on $m,$  and $\mu_0^{n-1}(A_m)>0$. Then, for every integrable finite valued $ \varphi_{n-1}(\omega_1,\ldots, \omega_{n-1})$
\begin{eqnarray}\label{mykolatina22}
\lim\limits_{m \to \infty} \frac{\int\limits_{A_m} \varphi_{n-1}(\omega_1,\ldots, \omega_{n-1})d \mu_0^{n-1}}{ \mu_0^{n-1}(A_m)}= \varphi_{n-1}(\omega_1,\ldots, \omega_{n-1})
\end{eqnarray}
almost everywhere relative to the measure $\mu_0^{n-1}.$
If to put 
$$  \varphi_{n-1}(\omega_1,\ldots, \omega_{n-1})=$$
\begin{eqnarray}\label{mykolatina23}
\int\limits_{\Omega_n^0}\psi_n(\omega_1, \ldots,\omega_n)f_n(\omega_1, \ldots, \omega_n) d P_n^0(\omega_n),  \quad (\omega_1, \ldots, \omega_{n-1}) \in \Omega_{n-1}, 
\end{eqnarray}
then we obtain the proof of  Lemma \ref{witka4}.
\end{proof}

In Theorem \ref{witka5}, we assume that for $\Delta S_n(\omega_1, \ldots,\omega_{n-1}, \omega_n), \  n=\overline{1,N},$ the representation
$$\Delta S_n(\omega_1, \ldots,\omega_{n-1}, \omega_n)=$$
$$S_{n-1}(\omega_1, \ldots,\omega_{n-1}) a_n (\omega_1, \ldots,\omega_{n-1}, \omega_n) \eta_n( \omega_n)=$$
\begin{eqnarray}\label{vitasja34}
d_n (\omega_1, \ldots,\omega_{n-1}, \omega_n) \eta_n( \omega_n), \quad n=\overline{1,N},\quad S_0>0,
\end{eqnarray}
is true, where the random values 
$d_n (\omega_1, \ldots,\omega_{n-1}, \omega_n), $
$a_n(\omega_1, \ldots,\omega_{n-1}, \omega_n),$ $\eta_n( \omega_n), $ $ \  n=\overline{1,N},$ given  on the probability space  $\{\Omega_n, {\cal F}_n, P_n\},$  satisfy the conditions
 $$ 0 < a_n(\omega_1, \ldots,\omega_{n-1}, \omega_n)\leq 1, \quad 1+a_n(\omega_1,\quad \ldots,\omega_{n-1}, \omega_n)\eta_n( \omega_n)> 0, $$ 
\begin{eqnarray}\label{pupsvitasjapop34}
 d_n (\omega_1, \ldots,\omega_{n-1}, \omega_n)>0, \quad P_n^0(\eta_n( \omega_n)>0)>0, \quad  P_n^0(\eta_n( \omega_n)<0)>0.
\end{eqnarray}
From these conditions we obtain
$\Omega_n^-=\Omega_n^{0-}\times \Omega_{n-1}, \ \Omega_n^+=\Omega_n^{0+}\times \Omega_{n-1},$
where $\Omega_n^{0-}=\{\omega_n \in \Omega_n^0, \eta_n( \omega_n)\leq 0\},  \ \Omega_n^{0+}=\{\omega_n \in \Omega_n^0, \eta_n( \omega_n)> 0\}.$

From the suppositions above, it follows that $P_n^0(\Omega_n^{0-})>0, \ P_n^0(\Omega_n^{0+})>0.$
The measure $ P_n^{0-}$ is a contraction of the measure $P_n^0$ on the $\sigma$-algebra ${\cal F}_n^{0-}=\Omega_n^{0-}\cap {\cal F}_n^0,$  $ P_n^{0+}$ is a contraction of the measure $P_n^0$ on the $\sigma$-algebra ${\cal F}_n^{0+}=\Omega_n^{0+}\cap {\cal F}_n^0.$

\begin{te}\label{witka5} Let $\Omega_i^0$ be a complete separable metric space and let  ${\cal F}_i^0$ be a Borell $\sigma$-algebra on $\Omega_i^0, \ i=\overline{1,N}.$
Suppose that for $\Delta S_n(\omega_1, \ldots,\omega_{n-1}, \omega_n), \  n=\overline{1,N},$ the representation (\ref{vitasja34}) is valid and
Lemma \ref{witka4}  conditions are true.
Then, for the nonnegative  random value  $f_n(\omega_1, \ldots,\omega_{n-1}, \omega_n)$  the inequalities 
$$\chi_{\Omega_n^{0-}}(\omega_n^1)\chi_{\Omega_n^{0+}}(\omega_n^2)\left[\frac{\Delta S_n^+(\omega_1, \ldots,\omega_{n-1}, \omega_n^2)}{V_n(\omega_1, \ldots,\omega_{n-1}, \omega_n^1,  \omega_n^2)} f_n(\omega_1, \ldots,\omega_{n-1}, \omega_n^1)+\right.$$

$$\left.\frac{\Delta S_n^-(\omega_1, \ldots,\omega_{n-1}, \omega_n^1)}{V_n(\omega_1, \ldots,\omega_{n-1}, \omega_n^1,  \omega_n^2)}f_n(\omega_1, \ldots,\omega_{n-1}, \omega_n^2)\right]\leq 1,$$
\begin{eqnarray}\label{vitasja36}
(\omega_1, \ldots,\omega_{n-1}) \in \Omega_{n-1}, \quad (\omega_n^1, \omega_n^2) \in \Omega_n^{0-}\times \Omega_n^{0+}, \quad n=\overline{1,N},
\end{eqnarray}
are true almost everywhere relative to the measure $P_{n-1}\times P_n^{0-}\times P_n^{0+}$ on the measurable space $\{\Omega_{n-1}\times \Omega_n^{0-}\times \Omega_n^{0+}, {\cal F}_{n-1}\times {\cal F}_{n}^{0-} \times {\cal F}_{n}^{0+}\}$.
\end{te}
\begin{proof} Under Theorem \ref{witka5} conditions, the set of martingale measures is a nonempty one.
Due to the equality (\ref{vitasja23}), we obtain 
$$\int\limits_{\Omega_N}\prod\limits_{i=1}^N\psi_i(\omega_1, \ldots,\omega_i)f_n(\omega_1, \ldots, \omega_n)\prod\limits_{i=1}^Nd P_i^0(\omega_i)=$$
\begin{eqnarray}\label{vitasja38}
\int\limits_{\Omega_n}\prod\limits_{i=1}^n\psi_i(\omega_1, \ldots,\omega_i)f_n(\omega_1, \ldots, \omega_n)\prod\limits_{i=1}^nd P_i^0(\omega_i).
\end{eqnarray}
 Further, 
$$\int\limits_{\Omega_n^0}\psi_n(\omega_1, \ldots,\omega_n)f_n(\omega_1, \ldots,\omega_n) d P_n^0(\omega_n)=$$
$$\int\limits_{\Omega_n^0}\int\limits_{\Omega_n^0}
\chi_{\Omega_n^-}(\omega_1, \ldots, \omega_{n-1}, \omega_n^1) \chi_{\Omega_n^+}(\omega_1, \ldots,\omega_{n-1}, \omega_n^2)\times $$
$$\alpha_n(\{\omega_1, \ldots,\omega_{n-1},\omega_n^1\};\{\omega_1, \ldots,\omega_{n-1},\omega_n^2\})\times $$ $$ \left[\frac{\Delta S_n^+(\omega_1, \ldots,\omega_{n-1}, \omega_n^2)}{V_n(\omega_1, \ldots,\omega_{n-1}, \omega_n^1,  \omega_n^2)}  f_n(\omega_1, \ldots,\omega_{n-1}, \omega_n^1)+\right.$$
\begin{eqnarray}\label{vitasja39}
\left.\frac{\Delta S_n^-(\omega_1, \ldots,\omega_{n-1}, \omega_n^1)}{V_n(\omega_1, \ldots,\omega_{n-1}, \omega_n^1,  \omega_n^2)}f_n(\omega_1, \ldots,\omega_{n-1}, \omega_n^2)\right]d P_n^0(\omega_n^1)d P_n^0(\omega_n^2). 
\end{eqnarray}

$$\chi_{\Omega_n^-}(\omega_1, \ldots,\omega_n^1)=\chi_{\Omega_{n-1}}(\omega_1, \ldots,\omega_{n-1}) \chi_{\Omega_n^{0-}}(\omega_n^1),$$
\begin{eqnarray}\label{tinvitasja40}
\chi_{\Omega_n^+}(\omega_1, \ldots,\omega_n^2)=\chi_{\Omega_{n-1}}(\omega_1, \ldots,\omega_{n-1}) \chi_{\Omega_n^{0+}}(\omega_n^2).
\end{eqnarray}
Due to Lemma \ref{witka4},  the inequality
$$ \int\limits_{\Omega_n^{0}}\int\limits_{\Omega_n^{0}}\chi_{\Omega_n^{0-}}(\omega_n^1)\chi_{\Omega_n^{0+}}(\omega_n^2)\alpha_n(\{\omega_1, \ldots, \omega_{n-1}, \omega_n^1\};\{\omega_1, \ldots, \omega_{n-1}, \omega_n^2\})
\times $$
$$ \left[\frac{\Delta S_n^+(\omega_1, \ldots,\omega_{n-1}, \omega_n^2)}{V_n(\omega_1, \ldots,\omega_{n-1}, \omega_n^1,  \omega_n^2)}  f_n(\omega_1, \ldots,\omega_{n-1}, \omega_n^1)+\right.$$
\begin{eqnarray}\label{vitasja40}
\left.\frac{\Delta S_n^-(\omega_1, \ldots,\omega_{n-1}, \omega_n^1)}{V_n(\omega_1, \ldots,\omega_{n-1}, \omega_n^1,  \omega_n^2)}f_n(\omega_1, \ldots,\omega_{n-1}, \omega_n^2)\right]d P_n^0(\omega_n^1)d P_n^0(\omega_n^2) \leq 1, 
\end{eqnarray}
 is true almost everywhere relative to the measure $P_{n-1}$ on the $\sigma$-algebra ${\cal F}_{n-1}.$
Let us put
\begin{eqnarray}\label{vitasja42}
 \alpha_n(\{\omega_1, \ldots,\omega_{n-1},\omega_n^1\};\{\omega_1, \ldots,\omega_{n-1}, \omega_n^2\})=
 \alpha_n(\omega_n^1;\omega_n^2),
\end{eqnarray}
where $  \alpha_n(\omega_n^1;\omega_n^2)$ satisfy the condition
\begin{eqnarray}\label{vitasja43}
 \int\limits_{\Omega_n^{0-}}\int\limits_{\Omega_n^{0+}} \alpha_n(\omega_n^1;\omega_n^2)d P_n^0(\omega_n^1)d P_n^0(\omega_n^2)=1.
\end{eqnarray}
Since, on the probability space $\{\Omega_n^{0-}\times \Omega_n^{0+}, $${\cal F}_n^{0-}\times {\cal F}_n^{0+}, P_n^{0-}\times  P_n^{0+}\},$ there  exists an exhaustive decomposition $\{A_{m,k}\}_{m,k=1}^\infty,$ let us put
\begin{eqnarray}\label{vitasja44}
 \alpha_n(\omega_n^1;\omega_n^2)=(1-\varepsilon) \frac{\chi_{A_{m,k}}(\omega_n^1;\omega_n^2)}{\mu_n(A_{m,k})}+ \varepsilon \frac{ \chi_{\Omega_n^{0-}\times \Omega_n^{0+}\setminus A_{m,k}}(\omega_n^1;\omega_n^2)}{\mu_n(\Omega_n^{0-}\times \Omega_n^{0+}\setminus A_{m,k})},
\end{eqnarray}
where $\mu_n(A)=[P_n^{0-}\times  P_n^{0+}](A), \ A \in {\cal F}_n^{0-}\times {\cal F}_n^{0+},$ and we assume that $\mu_n(A_{m,k})>0, $ $\mu_n(\Omega_n^{0-}\times \Omega_n^{0+}\setminus A_{m,k}) >0. $
Suppose that $(\omega_n^1;\omega_n^2) \in A_{m,k}$ and $\mu_n( A_{m,k})>0$ for the infinite number of $m$ and $k.$
Then,
$$ \int\limits_{\Omega_n^{0}}\int\limits_{\Omega_n^{0}}
\chi_{\Omega_n^{0-}}(\omega_n^1)\chi_{\Omega_n^{0+}}(\omega_n^2)\left[(1-\varepsilon) \frac{\chi_{A_{m,k}}(\omega_n^1;\omega_n^2)}{\mu_n(A_{m,k})}+ \varepsilon \frac{ \chi_{\Omega_n^{0-}\times \Omega_n^{0+}\setminus A_{m,k}}(\omega_n^1;\omega_n^2)}{\mu_n(\Omega_n^{0-}\times \Omega_n^{0+}\setminus A_{m,k})}\right]
\times $$
$$ \left[\frac{\Delta S_n^+(\omega_1, \ldots,\omega_{n-1}, \omega_n^2)}{V_n(\omega_1, \ldots,\omega_{n-1}, \omega_n^1,  \omega_n^2)}  f_n(\omega_1, \ldots,\omega_{n-1}, \omega_n^1)+\right.$$
\begin{eqnarray}\label{vitasja45}
\left.\frac{\Delta S_n^-(\omega_1, \ldots,\omega_{n-1}, \omega_n^1)}{V_n(\omega_1, \ldots,\omega_{n-1}, \omega_n^1,  \omega_n^2)}f_n(\omega_1, \ldots,\omega_{n-1}, \omega_n^2)\right]d P_n^0(\omega_n^1)d P_n^0(\omega_n^2) \leq 1.
\end{eqnarray}
Going to the limit as $m,k \to \infty$ and then as $\varepsilon \to 0,$ we obtain the inequality
$$\chi_{\Omega_n^{0,-}}(\omega_n^1) \chi_{\Omega_n^{0,+}}(\omega_n^2)\left[\frac{\Delta S_n^+(\omega_1, \ldots,\omega_{n-1}, \omega_n^2)}{V_n(\omega_1, \ldots,\omega_{n-1}, \omega_n^1,  \omega_n^2)}  f_n(\omega_1, \ldots,\omega_{n-1}, \omega_n^1)+\right.$$
\begin{eqnarray}\label{vitasja46}
\left.\frac{\Delta S_n^-(\omega_1, \ldots,\omega_{n-1}, \omega_n^1)}{V_n(\omega_1, \ldots,\omega_{n-1}, \omega_n^1,  \omega_n^2)}f_n(\omega_1, \ldots,\omega_{n-1}, \omega_n^2)\right] \leq 1, \quad (\omega_1, \ldots,\omega_{n-1}) \in \Omega_{n-1}.
\end{eqnarray}
which is valid almost everywhere relative to the measure $\mu_n.$
Theorem \ref{witka5} is proved.
\end{proof}

\begin{leme}\label{witka6}
Let $\Omega_n^0$ be a complete separable metric space and let ${\cal F}_n^0$ be a Borel $\sigma$-algebra on  $\Omega_n^0, \ n=\overline{1,N}$.
 If the conditions of  Lemma \ref{witka4}  are true, then the inequality
$$\chi_{\Omega_n^-}(\omega_1, \ldots,\omega_{n-1},\omega_n^1) \chi_{\Omega_n^+}(\omega_1, \ldots,\omega_{n-1}, \omega_n^2)\times$$ $$ \left[\frac{\Delta S_n^+(\omega_1, \ldots,\omega_{n-1}, \omega_n^2)}{V_n(\omega_1, \ldots,\omega_{n-1}, \omega_n^1,  \omega_n^2)}  f_n(\omega_1, \ldots,\omega_{n-1}, \omega_n^1)+\right.$$
\begin{eqnarray}\label{2vitasja46}
\left.\frac{\Delta S_n^-(\omega_1, \ldots,\omega_{n-1}, \omega_n^1)}{V_n(\omega_1, \ldots,\omega_{n-1}, \omega_n^1,  \omega_n^2)}f_n(\omega_1, \ldots,\omega_{n-1}, \omega_n^2)\right] \leq 1, \quad (\omega_1, \ldots,\omega_{n-1}) \in \Omega_{n-1},
\end{eqnarray}
 is valid almost everywhere relative to the measure $P_{n-1}\times [P_n^0\times P_n^0]$  on the measurable space $\{\Omega_{n-1}\times \Omega_n^{0}\times \Omega_n^{0}, {\cal F}_{n-1}\times {\cal F}_{n}^0 \times {\cal F}_{n}^0\}.$ 
\end{leme}
\begin{proof}
Due to the conditions  for $\Omega_n^a, a=-,+,$ the representation 
\begin{eqnarray}\label{1vitasja46}
\Omega_n^a=\bigcup\limits_{k=1}^{N_n}[A_{n}^{0, k a}\times V_{n-1}^{k}]
\end{eqnarray}
is true. Owing to Lemma \ref{witka6}  conditions, there exists an exhaustive decomposition $D_{m i}^n, \  m, i=\overline{1,\infty}, $ such that  $\bigcup\limits_{i=1}^\infty D_{m i}^n=\Omega_n^0, \ m=\overline{1,\infty}. $
Let us denote $A_{n}^{0, k a}\cap  D_{m i}^n=  E_{m i}^{n k a}.$
It is evident that $ E_{m i}^{n k a}$ forms an exhaustive decomposition of sets $A_{n}^{0, k a}, \ n=\overline{1,N}, \ k=\overline{1,\infty}, \  a=-,+, $ correspondingly.
Due to  Lemma \ref{witka4},  the inequality
\begin{eqnarray}\label{5vitasja33}
\int\limits_{\Omega_n^0}
\psi_n(\omega_1, \ldots,\omega_n)f_n(\omega_1, \ldots, \omega_n) d P_n^0(\omega_n)\leq 1,  \quad (\omega_1, \ldots, \omega_{n-1}) \in \Omega_{n-1}, 
\end{eqnarray}
is true almost everywhere relative to the measure $P_{n-1}.$
The equality 
$$\int\limits_{\Omega_n^0}\psi_n(\omega_1, \ldots,\omega_n)f_n(\omega_1, \ldots,\omega_n) d P_n^0(\omega_n)=$$
$$\int\limits_{\Omega_n^0}\int\limits_{\Omega_n^0}
\chi_{\Omega_n^-}(\omega_1, \ldots,\omega_{n-1}, \omega_n^1) \chi_{\Omega_n^+}(\omega_1, \ldots,\omega_{n-1}, \omega_n^2)\times $$
$$\alpha_n(\{\omega_1, \ldots,\omega_{n-1}, \omega_n^1\};\{\omega_1, \ldots,\omega_{n-1}, \omega_n^2\})\times $$ $$ \left[\frac{\Delta S_n^+(\omega_1, \ldots,\omega_{n-1}, \omega_n^2)}{V_n(\omega_1, \ldots,\omega_{n-1}, \omega_n^1,  \omega_n^2)}  f_n(\omega_1, \ldots,\omega_{n-1}, \omega_n^1)+\right.$$
\begin{eqnarray}\label{1vitasja39}
\left.\frac{\Delta S_n^-(\omega_1, \ldots,\omega_{n-1}, \omega_n^1)}{V_n(\omega_1, \ldots,\omega_{n-1}, \omega_n^1,  \omega_n^2)}f_n(\omega_1, \ldots,\omega_{n-1}, \omega_n^2)\right]d P_n^0(\omega_n^1)d P_n^0(\omega_n^2) 
\end{eqnarray}
is valid.
From the equality (\ref{1vitasja39}) and Lemma \ref{witka4},  the inequality
$$\int\limits_{\Omega_n^0}\int\limits_{\Omega_n^0}
\chi_{\Omega_n^-}(\omega_1, \ldots,\omega_{n-1}, \omega_n^1) \chi_{\Omega_n^+}(\omega_1, \ldots,\omega_{n-1},\omega_n^2)\times $$
$$\alpha_n(\{\omega_1, \ldots,\omega_{n-1}, \omega_n^1\};\{\omega_1, \ldots,\omega_{n-1},\omega_n^2\})\times $$ 
$$ \left[\frac{\Delta S_n^+(\omega_1, \ldots,\omega_{n-1}, \omega_n^2)}{V_n(\omega_1, \ldots,\omega_{n-1}, \omega_n^1,  \omega_n^2)}  f_n(\omega_1, \ldots,\omega_{n-1}, \omega_n^1)+\right.$$
\begin{eqnarray}\label{1vitasja40}
\left.\frac{\Delta S_n^-(\omega_1, \ldots,\omega_{n-1}, \omega_n^1)}{V_n(\omega_1, \ldots,\omega_{n-1}, \omega_n^1,  \omega_n^2)}f_n(\omega_1, \ldots,\omega_{n-1}, \omega_n^2)\right]d P_n^0(\omega_n^1)d P_n^0(\omega_n^2) \leq 1, 
\end{eqnarray}
 is true almost everywhere relative to the measure $P_{n-1}$ on the $\sigma$-algebra ${\cal F}_{n-1}.$
Let us put 
$$\alpha_{n}^{r,s-} (\omega_1^1, \ldots,\omega_{n}^1)=\sum\limits_{k=1}^{N_n}\alpha_{n,k,r,s}^-( \omega_n^1)\chi_{A_n^{0,k-}}(\omega_n^1) \chi_{V_{n-1}^k}(\omega_1^1, \ldots,\omega_{n}^1),   $$

$$\alpha_{n}^{m,i+} (\omega_1^2, \ldots,\omega_{n}^2)=\sum\limits_{k=1}^{N_n}\alpha_{n,k,m,i}^+( \omega_n^2)\chi_{A_n^{0,k+}}(\omega_n^2) \chi_{V_{n-1}^k}(\omega_1^2, \ldots,\omega_{n-1}^2),   $$
\begin{eqnarray}\label{1vitatinnnna40}
 \alpha_{n}^{r,s,m,i} (\{\omega_1^1, \ldots,\omega_{n}^1\}; \{\omega_1^2, \ldots,\omega_{n}^2\}) = \alpha_{n}^{r,s-} (\omega_1^1, \ldots,\omega_{n}^1)\alpha_{n}^{m,i+} (\omega_1^2, \ldots,\omega_{n}^2),
\end{eqnarray}
where
$$ \alpha_{n,k,r,s}^-( \omega_n^1)=  \left[(1-\delta)\frac{\chi_{E_{r s}^{n k -}}(\omega_n^1)}{P_n^0(E_{r s}^{n k -})}+\delta \frac{\chi_{A_n^{0 k-} \setminus E_{r s}^{n k -}}(\omega_n^1)}{P_n^0(A_n^{0 k-} \setminus E_{r s}^{n k -})}\right],$$

\begin{eqnarray}\label{1vitatinnnna41}
\alpha_{n,k,m,i}^+( \omega_n^2)= \left[(1-\delta)\frac{\chi_{E_{m i}^{n k +}}(\omega_n^2)}{P_n^0(E_{m i}^{n k +})}+\delta \frac{\chi_{A_n^{0 k+} \setminus E_{m i}^{n k +}}(\omega_n^2)}{P_n^0(A_n^{0 k+} \setminus E_{m i}^{n k +})}\right], \quad 0<\delta <1.
\end{eqnarray}
In the formulas (\ref{1vitatinnnna41}), we assume  that the inequalities 
\begin{eqnarray}\label{supervitochka1}
  P_n^0(E_{r s}^{n k -})>0, \  P_n^0(A_n^{0 k-} \setminus E_{r s}^{n k -})>0, \   P_n^0(E_{m i}^{n k +})>0, \   P_n^0(A_n^{0 k+} \setminus E_{m i}^{n k +})>0,
\end{eqnarray}
are true.
Let us consider
$$ \alpha_{n}^{r,s,m,i} (\{\omega_1, \ldots,\omega_{n-1}, \omega_{n-1}^1\}; \{\omega_1, \ldots,\omega_{n-1},  \omega_{n}^2\}) =$$ 
\begin{eqnarray}\label{1vitatinnnna42}
\alpha_{n}^{r,s-} (\omega_1, \ldots,\omega_{n-1},\omega_{n}^1)\alpha_{n}^{m,i+} (\omega_1, \ldots,\omega_{n-1}, \omega_{n}^2).
\end{eqnarray}
Suppose that $(\omega_1, \ldots,\omega_{n-1}) \in V_{n-1}^{k}$ for a certain $k.$
Then,
$$ \alpha_{n}^{r,s,m,i} (\{\omega_1, \ldots,\omega_{n-1}, \omega_{n-1}^1\}; \{\omega_1, \ldots,\omega_{n-1},  \omega_{n}^2\}) =$$ 
$$\left[(1-\delta)\frac{\chi_{E_{r s}^{n k -}}(\omega_n^1)}{P_n^0(E_{r s}^{n k -})}+\delta \frac{\chi_{A_n^{0 k-} \setminus E_{r s}^{n k -}}(\omega_n^1)}{P_n^0(A_n^{0 k-} \setminus E_{r s}^{n k -})}\right]\times$$
\begin{eqnarray}\label{2vitasja40}
 \left[(1-\delta)\frac{\chi_{E_{m i}^{n k +}}(\omega_n^2)}{P_n^0(E_{m i}^{n k +})}+\delta \frac{\chi_{A_n^{0 k+} \setminus E_{m i}^{n k +}}(\omega_n^2)}{P_n^0(A_n^{0 k+} \setminus E_{m i}^{n k +})}\right].
\end{eqnarray}
 We assume that the point $ (\omega_n^1,\omega_n^2) \in E_{r s}^{n k -}\times  E_{m i}^{n k +}$ for the infinite number of $r,s$  and  $m,i $ , where $P_n^0( E_{r s}^{n k -})>0, \ P_n^0(  E_{m i}^{n k +})>0.$

  Substituting (\ref{2vitasja40}) into (\ref{1vitasja40}) and going to the limit  as  $m, k \to \infty$  $r, s  \to \infty$  and then as  $\delta \to 0,$ we obtain the needed inequality. Lemma \ref{witka6} is proved.
\end{proof}

\begin{te}\label{witka7}
Suppose that the conditions of Theorem \ref{witka5} are true.  If 
for a certain  $\omega_n^1 \in \Omega_n^{0-}$ and $\omega_n^2 \in \Omega_n^{0+}$ the inequalities 

$$\sup\limits_{(\omega_1, \ldots,  \omega_{n-1}) \in \Omega_{n-1}} \frac{1}{\Delta S_n^-(\omega_1, \ldots,  \omega_{n-1}, \omega_n^1)}< \infty,$$
\begin{eqnarray}\label{tinwickpussy1}
\sup\limits_{(\omega_1, \ldots,  \omega_{n-1}) \in \Omega_{n-1}} \frac{1}{\Delta S_n^+(\omega_1, \ldots,  \omega_{n-1}, \omega_n^2)}< \infty,  \quad  n=\overline{1,N}, 
\end{eqnarray}
are true, then  the  nonnegative random values $f_n(\omega_1, \ldots,\omega_{n-1}, \omega_n),\  n=\overline{1,N},$ satisfy the  inequalities
$$ f_n(\omega_1, \ldots,\omega_{n-1}, \omega_n) \leq $$ 
\begin{eqnarray}\label{1vitasjatinnaaa47}
(1+ \gamma_{n-1}(\omega_1, \ldots,\omega_{n-1})\Delta S_n(\omega_1, \ldots,\omega_{n-1}, \omega_n)), \quad  n=\overline{1,N},
\end{eqnarray} 
where $\gamma_{n-1}(\omega_1, \ldots,\omega_{n-1}) $ is a bounded ${\cal F}_{n-1}$-measurable random value.  
\end{te}
\begin{proof}
From the inequality (\ref{vitasja46}), it follows the inequality
$$f_n (\omega_1, \ldots,  \omega_{n-1}, \omega_n^2) \leq $$
\begin{eqnarray}\label{3vitamyk16}
 1+\frac{1- f_n (\omega_1, \ldots,  \omega_{n-1}, \omega_n^1) }{\Delta S_n^-(\omega_1, \ldots,  \omega_{n-1}, \omega_n^1)}\Delta S_n^+(\omega_1, \quad  \ldots,  \omega_{n-1}, \omega_n^2), \ \omega_n^1 \in \Omega_n^{0-}, \   \omega_n^2 \in \Omega_n^{0+}.
\end{eqnarray} 
Let us define
\begin{eqnarray}\label{31vitamyk16}
\gamma_{n-1}(\omega_1, \ldots,  \omega_{n-1})=
\inf_{\{\omega_n^1,  \eta_n^-(\omega_n^1) > 0\}}\frac{1- f_n (\omega_1, \ldots,  \omega_{n-1}, \omega_n^1) }{\Delta S_n^-(\omega_1, \ldots,  \omega_{n-1}, \omega_n^1)},
\end{eqnarray}
then, taking into account the inequality (\ref{3vitamyk16}), we  obtain the inequality
\begin{eqnarray}\label{3vitamyk17}
f_n (\omega_1, \ldots,  \omega_{n-1}, \omega_n^2) \leq 1+\gamma_{n-1}(\omega_1, \ldots,  \omega_{n-1})\Delta S_n^+(\omega_1, \ldots,  \omega_{n-1}, \omega_n^2). 
\end{eqnarray} 
From the definition of  $\gamma_{n-1}(\omega_1, \ldots,  \omega_{n-1}),$ we obtain the inequality
\begin{eqnarray}\label{3vitamyk18}
f_n (\omega_1, \ldots,  \omega_{n-1}, \omega_n^1) \leq 1-\gamma_{n-1}(\omega_1, \ldots,  \omega_{n-1})\Delta S_n^-(\omega_1, \ldots,  \omega_{n-1}, \omega_n^1).
\end{eqnarray} 
The inequalities (\ref{3vitamyk17}), (\ref{3vitamyk18}) give the inequality
\begin{eqnarray}\label{3vitamyk19}
f_n (\omega_1, \ldots,  \omega_{n-1}, \omega_n) \leq 1+\gamma_{n-1}(\omega_1, \ldots,  \omega_{n-1})\Delta S_n(\omega_1, \ldots,  \omega_{n-1}, \omega_n).
\end{eqnarray}
Let us prove the boundedness of $\gamma_{n-1}(\omega_1, \ldots,  \omega_{n-1}).$ From the inequalities (\ref{3vitamyk17}), (\ref{3vitamyk18}) we obtain

$$\frac{1}{\Delta S_n^-(\omega_1, \ldots,  \omega_{n-1}, \omega_n^1)} \geq $$
\begin{eqnarray}\label{3vitamyk20}
 \gamma_{n-1}(\omega_1, \ldots,  \omega_{n-1}) \geq -  \frac{1}{\Delta S_n^+(\omega_1, \ldots,  \omega_{n-1}, \omega_n^2)}.
\end{eqnarray}
Due to Theorem \ref{witka7} conditions, we obtain the boundedness of $\gamma_{n-1}(\omega_1, \ldots,  \omega_{n-1}).$
The ${\cal F}_{n-1}$ measurability of the random value $\gamma_{n-1}(\omega_1, \ldots,  \omega_{n-1})$  follows from the fact that $\Omega_n^0$ is separable metric space and infimum  is reached on the countable set, which is dense in  $\Omega_n^0.$ Theorem \ref{witka7} is proved.
\end{proof}

\begin{te}\label{1witka7}
Let the conditions of Lemma \ref{witka6} be valid. If there exist 
$ \omega_n^1 \in A_n^{0 k -}, \  \omega_n^2 \in A_n^{0 k +},$ and 
the real numbers $a_k, \ b_k, \ k=\overline{1, N_n},$ such that
$$ \sup\limits_{(\omega_1, \ldots,  \omega_{n-1}) \in V_{n-1}^k } \frac{1}{\Delta S_n^-(\omega_1, \ldots,  \omega_{n-1}, \omega_n^1)}=a_k^n<\infty, $$
$$\sup\limits_{(\omega_1, \ldots,  \omega_{n-1}) \in V_{n-1}^k }  \frac{1}{\Delta S_n^+(\omega_1, \ldots,  \omega_{n-1}, \omega_n^2)}=b_k^n<\infty,  \quad k=\overline{1, N_n},  \quad  n=\overline{1,N}, $$
\begin{eqnarray}\label{3vitamyktinaaa21}
\max\limits_{1\leq n \leq N}\sup\limits_{1\leq k\leq N_n}\max\{a_k^n, b_k^n\}<\infty,
\end{eqnarray}
then there exists  a bounded ${\cal F}_{n-1}$-measurable random value $\gamma_{n-1}(\omega_1, \ldots,\omega_{n-1}) $ such that the inequalities
$$ f_n(\omega_1, \ldots,\omega_{n-1}, \omega_n)) \leq$$
\begin{eqnarray}\label{2vitasjatinnaaa47}
(1+ \gamma_{n-1}(\omega_1, \ldots,\omega_{n-1})\Delta S_n(\omega_1, \ldots,\omega_{n-1}, \omega_n)), \quad  n=\overline{1,N},
\end{eqnarray}
are true. 
\end{te}
\begin{proof}
For $\omega_n^1 \in A_n^{0 k -}, \  \omega_n^2 \in A_n^{0 k +}$ and 
$(\omega_1, \ldots,\omega_{n-1}) \in V_{n-1}^k,$  we have that 
$(\omega_1, \ldots,\omega_{n-1}, \omega_n^1) \in \Omega_n^-,$ $ \ (\omega_1, \ldots,\omega_{n-1}, \omega_n^2) \in \Omega_n^+.$ 
Then, from the inequality (\ref{2vitasja46}), we obtain the inequality
 $$ \left[\frac{\Delta S_n^+(\omega_1, \ldots,\omega_{n-1}, \omega_n^2)}{V_n(\omega_1, \ldots,\omega_{n-1}, \omega_n^1,  \omega_n^2)}  f_n(\omega_1, \ldots,\omega_{n-1}, \omega_n^1)+\right.$$
\begin{eqnarray}\label{2vitasjatinnaa46}
\left.\frac{\Delta S_n^-(\omega_1, \ldots,\omega_{n-1}, \omega_n^1)}{V_n(\omega_1, \ldots,\omega_{n-1}, \omega_n^1,  \omega_n^2)}f_n(\omega_1, \ldots,\omega_{n-1}, \omega_n^2)\right] \leq 1.
\end{eqnarray}
  From the inequality (\ref{2vitasjatinnaa46}), it follows the inequality
\begin{eqnarray}\label{3vitamyktinnaaa16}
f_n (\omega_1, \ldots,  \omega_{n-1}, \omega_n^2) \leq 1+\frac{1- f_n (\omega_1, \ldots,  \omega_{n-1}, \omega_n^1) }{\Delta S_n^-(\omega_1, \ldots,  \omega_{n-1}, \omega_n^1)}\Delta S_n^+(\omega_1, \ldots,  \omega_{n-1}, \omega_n^2).
\end{eqnarray} 
Let us define
$$\gamma_{n-1}^k(\omega_1, \ldots,  \omega_{n-1})=$$
\begin{eqnarray}\label{31vitamyktinnaaa16}
\inf_{\{\omega_n^1 \in A_n^{0,k -} \}}\frac{1- f_n (\omega_1, \ldots,  \omega_{n-1}, \omega_n^1) }{\Delta S_n^-(\omega_1, \ldots,  \omega_{n-1}, \omega_n^1)}, \quad (\omega_1, \ldots,  \omega_{n-1}) \in V_{n-1}^k,
\end{eqnarray}
then, taking into account the inequality (\ref{3vitamyktinnaaa16}), we  have the inequality
\begin{eqnarray}\label{3vitamyktinnaaa17}
f_n (\omega_1, \ldots,  \omega_{n-1}, \omega_n^2) \leq 1+\gamma_{n-1}^k(\omega_1, \ldots,  \omega_{n-1})\Delta S_n^+(\omega_1, \ldots,  \omega_{n-1}, \omega_n^2). 
\end{eqnarray} 
From the definition of  $\gamma_{n-1}^k(\omega_1, \ldots,  \omega_{n-1}),$ we obtain the inequality 
\begin{eqnarray}\label{3vitamyktinnaaa18}
f_n (\omega_1, \ldots,  \omega_{n-1}, \omega_n^1) \leq 1-\gamma_{n-1}^k(\omega_1, \ldots,  \omega_{n-1})\Delta S_n^-(\omega_1, \ldots,  \omega_{n-1}, \omega_n^1).
\end{eqnarray} 
The inequalities (\ref{3vitamyktinnaaa17}), (\ref{3vitamyktinnaaa18}) give the inequality
\begin{eqnarray}\label{3vitamyktinnaaa19}
f_n (\omega_1, \ldots,  \omega_{n-1}, \omega_n) \leq 1+\gamma_{n-1}^k(\omega_1, \ldots,  \omega_{n-1})\Delta S_n(\omega_1, \ldots,  \omega_{n-1}, \omega_n).
\end{eqnarray}
Let us prove the boundedness of $\gamma_{n-1}^k(\omega_1, \ldots,  \omega_{n-1}).$ From the inequalities (\ref{3vitamyktinnaaa17}), (\ref{3vitamyktinnaaa18}), we obtain the inequalities

$$ a_k^n=\sup\limits_{(\omega_1, \ldots,  \omega_{n-1}) \in V_{n-1}^k} \frac{1}{\Delta S_n^-(\omega_1, \ldots,  \omega_{n-1}, \omega_n^1)} \geq $$
\begin{eqnarray}\label{3vitamyktinnaaa20}
 \gamma_{n-1}^k(\omega_1, \ldots,  \omega_{n-1}) \geq - \sup\limits_{(\omega_1, \ldots,  \omega_{n-1}) \in V_{n-1}^k} \frac{1}{\Delta S_n^+(\omega_1, \ldots,  \omega_{n-1}, \omega_n^2)}=- b_k^n.
\end{eqnarray}
From this, it follows the boundedness of $\gamma_{n-1}^k(\omega_1, \ldots,  \omega_{n-1}).$
The ${\cal F}_{n-1}$ measurability of the random value $\gamma_{n-1}^k(\omega_1, \ldots,  \omega_{n-1})$  follows from the fact that $\Omega_n^0$ is separable metric space and infimum  is reached on the countable set, which is dense in  $\Omega_n^0.$
To complete the proof of Theorem \ref{1witka7}, let us put
\begin{eqnarray}\label{kiss3vitamyktinnaaa20}
 \gamma_{n-1}(\omega_1, \ldots,  \omega_{n-1})=\sum\limits_{k=1}^{N_n}\chi_{V_{n-1}^k}((\omega_1, \ldots,  \omega_{n-1})\gamma_{n-1}^k(\omega_1, \ldots,  \omega_{n-1}), 
\end{eqnarray}
then for such $\gamma_{n-1}(\omega_1, \ldots,  \omega_{n-1})$ the inequality (\ref{2vitasjatinnaaa47})
are satisfied. Theorem \ref{1witka7} is proved.
\end{proof}

\section{Optional decomposition for super-martingales.}
In this section, we give simple proof of optional decomposition for the nonnegative super-martingale relative to the set of equivalent martingale measures. Such a proof first appeared  in the paper \cite{GoncharNick1}.
First, the optional decomposition for   diffusion processes super-martingale was opened by    El Karoui N. and  Quenez M. C. \cite{KarouiQuenez}. After that, Kramkov D. O. and Follmer H. \cite{Kramkov}, \cite{FolmerKramkov1} proved the optional decomposition for the nonnegative bounded super-martingales.  Folmer H. and Kabanov Yu. M.  \cite{FolmerKabanov1},  \cite{FolmerKabanov}  proved analogous result for an arbitrary super-martingale. Recently, Bouchard B. and Nutz M. \cite{Bouchard1} considered a class of discrete models and proved the necessary and sufficient conditions for the validity of the optional decomposition. 

\begin{te}\label{Tinnna1}
Let $\Omega_i^0$ be a complete separable metric space and let ${\cal F}_i^0$ be a Borell $\sigma$-algebra on $\Omega_i^0, \ i=\overline{1,N}.$ Suppose that the 
evolution $\{S_n(\omega_1, \ldots,\omega_{n})\}_{n=1}^N $ of risky assets satisfies the conditions of Theorems \ref{witka2}, \ref{witka5}, \ref{witka7},  \ref{1witka7}, then for every nonnegative super-martingale $\{f_n^1(\omega_1, \ldots,\omega_{n})\}_{n=0}^N$ relative to the set of martingale  measure $M,$ described in Theorem \ref{witka2}, the optional decomposition is true.
\end{te}
\begin{proof}
Without loss of generality, we assume that $f_n^1(\omega_1, \ldots,\omega_{n}) \geq a, $ where $ a $ is a real positive number.  If it is not so, then we can come to the super-martingale $f_n^1(\omega_1, \ldots,\omega_{n})+a.$ Let us consider the set of random values 
\begin{eqnarray}\label{Tinnna2}
f_n(\omega_1, \ldots,\omega_{n})=\frac{f_n^1(\omega_1, \ldots,\omega_{n})}{f_{n-1}^1(\omega_1, \ldots,\omega_{n-1})}, \quad n=\overline{1, N}.
\end{eqnarray}
Every random value $f_n(\omega_1, \ldots,\omega_{n})$ satisfies the conditions of Lemma \ref{witka4}. Due to Theorems \ref{witka7},  \ref{1witka7}, the inequalities 
\begin{eqnarray}\label{Tinnna3}
\frac{f_n^1(\omega_1, \ldots,\omega_{n})}{f_{n-1}^1(\omega_1, \ldots,\omega_{n-1})} \leq 1+\gamma_{n-1}(\omega_1, \ldots,\omega_{n-1})\Delta S_n(\omega_1, \ldots,\omega_{n}), \quad n=\overline{1, N},
\end{eqnarray}
are true, where  $\gamma_{n-1}(\omega_1, \ldots,\omega_{n-1})$ is a bounded ${\cal F}_{n-1}$-measurable random value. Since $E^Q|\Delta S_n(\omega_1, \ldots,\omega_{n})|< \infty, \ Q\in M,$ we have 
\begin{eqnarray}\label{Tinnna4}
E^Q\{ \gamma_{n-1}(\omega_1, \ldots,\omega_{n-1})\Delta S_n(\omega_1, \ldots,\omega_{n}) |{\cal F}_{n-1}\}=0,  \quad Q \in M, \quad n=\overline{1, N}.
 \end{eqnarray}
Let us denote
\begin{eqnarray}\label{Tinnna5}
\xi_n^0(\omega_1, \ldots,\omega_{n})=1+\gamma_{n-1}(\omega_1, \ldots,\omega_{n-1})\Delta S_n(\omega_1, \ldots,\omega_{n}), \quad n=\overline{1, N}.
 \end{eqnarray}
Then, from the inequalities (\ref{Tinnna3}), we obtain the inequalities
$$f_n^1(\omega_1, \ldots,\omega_{n}) \leq$$
\begin{eqnarray}\label{Tinnna6}
 f_{n-1}^1(\omega_1, \ldots,\omega_{n-1}) +
 f_{n-1}^1(\omega_1, \ldots,\omega_{n-1})[\xi_n^0(\omega_1, \ldots,\omega_{n})-1], \quad n=\overline{1, N}.
 \end{eqnarray}
Introduce the denotations
$$g_n(\omega_1, \ldots,\omega_{n})=$$
\begin{eqnarray}\label{Tinnna7}
-f_n^1(\omega_1, \ldots,\omega_{n})+ f_{n-1}^1(\omega_1, \ldots,\omega_{n-1})\xi_n^0(\omega_1, \ldots,\omega_{n}), \quad n=\overline{1, N}.
 \end{eqnarray}
Then, $g_n(\omega_1, \ldots,\omega_{n})\geq 0, \ n=\overline{1, N}, $ and
\begin{eqnarray}\label{Tinnna8}
 E^Q g_n(\omega_1, \ldots,\omega_{n})\leq E^Q f_n^1(\omega_1, \ldots,\omega_{n})+ E^Q f_n^1(\omega_1, \ldots,\omega_{n-1}).
 \end{eqnarray}
The equalities (\ref{Tinnna7}) give the equalities
$$f_n^1(\omega_1, \ldots,\omega_{n})=$$
\begin{eqnarray}\label{Tinnna9}
f_0^1+\sum\limits_{i=1}^n f_{i-1}^1(\omega_1, \ldots,\omega_{n-1})[\xi_i^0(\omega_1, \ldots,\omega_{i})-1] - \sum\limits_{i=1}^n g_i(\omega_1, \ldots,\omega_{i}),
\ n=\overline{1, N}.
 \end{eqnarray}
Let us put 
\begin{eqnarray}\label{Tinnna10}
M_n(\omega_1, \ldots,\omega_{n})=f_0^1+\sum\limits_{i=1}^n f_{i-1}^1(\omega_1, \ldots,\omega_{i-1})[\xi_i^0(\omega_1, \ldots,\omega_{i})-1],
\quad n=\overline{1, N},
 \end{eqnarray}
then $E^Q\{M_n(\omega_1, \ldots,\omega_{n})|{\cal F}_{n-1}\}=M_{n-1}(\omega_1, \ldots,\omega_{n-1}).$ Theorem \ref{Tinnna1} is proved.
\end{proof}

\section{Spot measures and integral representation for martingale measures.}

In this section, we introduce the family of spot measures. After that, we obtain the representations for the family of  spot measures and define integral over these set of measures. The sufficient conditions are found, under which the integral over these set of measures is a set of martingale measures being equivalent to the initial measure. The introduced  family of spot measures  is a family
of extreme points for these set of equivalent  measures.

We give an evident construction of the set of martingale measures for risky assets evolution, given by the formula (\ref{tin1vitasika1ja9}). First  of all, to do that we consider  a simple case as the measures  $P_n^0$ is concentrated at two points  $\omega_n^1, \omega_n^2 \in \Omega_n^0, $  where $\omega_n^1 \in A_n^{0 k-},
\omega_n^2 \in A_n^{0 k+}$ for a certain $k,$ depending on $n,$ for the representation $\Omega_n^-$ and  $\Omega_n^+,$ given by the formula (\ref{100vitasikja7}). Let us put $P_n^0(\omega_n^1)=p_n^k,$  $P_n^0(\omega_n^2)=1-p_n^k,$ where $0< p_n^k< 1.$ Then, to satisfy the conditions
(\ref{1vitasja7}) - (\ref{3vitasja7}), we need to put 
\begin{eqnarray}\label{pupsvitasjapups1}
\alpha_n(\{\omega_1^1, \ldots, \omega_n^1\} ; \{\omega_1^2, \ldots, \omega_n^2\})=\frac{1}{ p_n^k(1-p_n^k)}, \quad n=\overline{1,N},
\end{eqnarray}
and to require  that 
$$\Delta S_n^-(\omega_1, \ldots,\omega_{n-1}, \omega_n^1)< \infty, \quad
(\omega_1, \ldots,\omega_{n-1}, \omega_n^1) \in \Omega_n^-,$$
\begin{eqnarray}\label{pupsvitasjapups2}
 \Delta S_n^+(\omega_1, \ldots,\omega_{n-1}, \omega_n^2)<\infty, \quad  (\omega_1, \ldots,\omega_{n-1}, \omega_n^2) \in \Omega_n^+.
\end{eqnarray}
Let us denote $\mu_{\{\omega_n^1,\omega_n^2\},\ldots ,\{\omega_N^1, \omega_N^2\}}(A)$ the measure, generated by the recurrent relations (\ref{vitasja6}) - (\ref{pupecvitasja5}), for the  measures $P_n^0, \ n=\overline{1,N},$ concentrated at two points.
For the  point  $\{\omega_n^1,\omega_n^2\},\ldots ,\{\omega_N^1, \omega_N^2\} \in   \Omega_N \times\Omega_N,$
 the recurrent relations (\ref{vitasja6}) - (\ref{pupecvitasja5}) is converted 
relative to the set of measures $\mu_{\{\omega_n^1,\omega_n^2\},\ldots ,\{\omega_N^1, \omega_N^2\}}^{(\omega_1, \ldots,\omega_{n-1})}(A)$ into the recurrent relations 

$$\mu_{\{\omega_N^1, \omega_N^2\}}^{(\omega_1, \ldots,\omega_{N-1})}(A)= \chi_{\Omega_N^-}(\omega_1, \ldots,\omega_{N-1}, \omega_N^1) \chi_{\Omega_N^+}(\omega_1, \ldots,\omega_{N-1}, \omega_N^2)\times$$
$$ \left[\frac{\Delta S_N^+(\omega_1, \ldots,\omega_{N-1}, \omega_N^2)}{V_N(\omega_1, \ldots,\omega_{N-1}, \omega_N^1,  \omega_N^2)}\mu_{N}^{(\omega_1, \ldots,\omega_{N-1}, \omega_N^1)}(A)+ \right.$$
\begin{eqnarray}\label{2vitasja6}
\left. \frac{\Delta S_N^-(\omega_1, \ldots,\omega_{N-1}, \omega_N^1)}{V_N(\omega_1, \ldots,\omega_{N-1}, \omega_N^1,  \omega_N^2)}\mu_N^{(\omega_1, \ldots,\omega_{N-1}, \omega_N^2)}(A)\right], \quad A \in {\cal F}_N,
\end{eqnarray}

$$\mu_{\{\omega_n^1,\omega_n^2\},\ldots ,\{\omega_N^1, \omega_N^2\}}^{(\omega_1, \ldots,\omega_{n-1})}(A)= \chi_{\Omega_n^-}(\omega_1, \ldots,\omega_{n-1},\omega_n^1) \chi_{\Omega_n^+}(\omega_1, \ldots,\omega_{n-1}, \omega_n^2)\times$$

$$ \left[\frac{\Delta S_n^+(\omega_1, \ldots,\omega_{n-1}, \omega_n^2)}{V_n(\omega_1, \ldots,\omega_{n-1}, \omega_n^1,  \omega_n^2)}\mu_{\{\omega_{n+1}^1,\omega_{n+1}^2\},\ldots ,\{\omega_N^1, \omega_N^2\}}^{(\omega_1, \ldots,\omega_{n-1}, \omega_n^1)}(A)+ \right.$$
\begin{eqnarray}\label{1vitasja5}
\left. \frac{\Delta S_n^-(\omega_1, \ldots,\omega_{n-1}, \omega_n^1)}{V_n(\omega_1, \ldots,\omega_{n-1}, \omega_n^1,  \omega_n^2)}\mu_{\{\omega_{n+1}^1,\omega_{n+1}^2\},\ldots ,\{\omega_N^1, \omega_N^2\}}^{(\omega_1, \ldots,\omega_{n-1}, \omega_n^2)}(A)\right], \quad n=\overline{2,N}, \quad A \in {\cal F}_N,
\end{eqnarray}

$$\mu_{\{\omega_n^1,\omega_n^2\},\ldots ,\{\omega_N^1, \omega_N^2\}}(A)= \chi_{\Omega_1^-}(\omega_1^1) \chi_{\Omega_1^+}(\omega_1^2)\times$$
\begin{eqnarray}\label{pupsja1vitasja5}
 \left[\frac{\Delta S_1^+(\omega_n^2)}{V_1(\omega_1^1,  \omega_1^2)}\mu_{\{\omega_{2}^1,\omega_{2}^2\},\ldots ,\{\omega_N^1, \omega_N^2\}}^{(\omega_1^1)}(A)+
 \frac{\Delta S_1^-(\omega_1^1)}{V_1(\omega_1^1,  \omega_1^2)}\mu_{\{\omega_{2}^1,\omega_{2}^2\},\ldots ,\{\omega_N^1, \omega_N^2\}}^{(\omega_1^2)}(A)\right], 
\end{eqnarray}
where we put
\begin{eqnarray}\label{pupsja19vitasja5}
 \mu_N^{(\omega_1, \ldots,\omega_{N-1}, \omega_N)}(A)=\chi_{A}(\omega_1, \ldots,\omega_{N-1}, \omega_N), \quad A \in {\cal F}_N.
\end{eqnarray}
The recurrent relations  (\ref{2vitasja6}) -  (\ref{pupsja1vitasja5}) we call the recurrent relations for the spot measures $\mu_{\{\omega_n^1,\omega_n^2\},\ldots ,\{\omega_N^1, \omega_N^2\}}(A).$

Let us consider  the random values
$$ \psi_n(\omega_1, \ldots,\omega_n)=\chi_{\Omega_n^-}(\omega_1, \ldots,\omega_{n-1}, \omega_n) \psi_n^1(\omega_1, \ldots,\omega_n)+$$
\begin{eqnarray}\label{vitasja111}
\chi_{\Omega_n^+}(\omega_1, \ldots,\omega_{n-1}, \omega_n) \psi_n^2(\omega_1, \ldots,\omega_n),
\end{eqnarray}
where
$$\psi_n^1(\omega_1, \ldots,\omega_{n-1},\omega_n)=\chi_{\Omega_n^+}(\omega_1, \ldots,\omega_{n-1}, \omega_n^2) \times $$
\begin{eqnarray}\label{vitasja112}
\frac{\Delta S_n^+(\omega_1, \ldots,\omega_{n-1}, \omega_n^2)}{V_n(\omega_1, \ldots,\omega_{n-1}, \omega_n^1,  \omega_n^2)}, \quad (\omega_1, \ldots,\omega_{n-1}) \in \Omega_{n-1},
\end{eqnarray} 

$$\psi_n^2(\omega_1, \ldots,\omega_{n-1},\omega_n)=\chi_{\Omega_n^-}(\omega_1, \ldots,\omega_{n-1}, \omega_n^1) \times $$
\begin{eqnarray}\label{vitasja113}
\frac{\Delta S_n^-(\omega_1, \ldots,\omega_{n-1}, \omega_n^1)}{V_n(\omega_1, \ldots,\omega_{n-1}, \omega_n^1,  \omega_n^2)}, \quad  (\omega_1, \ldots,\omega_{n-1}) \in \Omega_{n-1}.
\end{eqnarray}

\begin{leme}\label{vitasja114}
For the spot measure $\mu_{\{\omega_1^1,\omega_1^2\},\ldots ,\{\omega_N^1, \omega_N^2\}}(A)$ the representation 

$$\mu_{\{\omega_1^1,\omega_1^2\},\ldots ,\{\omega_N^1, \omega_N^2\}}(A)=$$
\begin{eqnarray}\label{vitasja115}
\sum\limits_{i_1=1}^2\ldots \sum\limits_{i_N=1}^2\prod\limits_{j=1}^N\psi_j(\omega_1^{i_1}, \ldots, \omega_j^{i_j})\chi_{A}(\omega_1^{i_1}, \ldots, \omega_N^{i_N}), \quad A \in {\cal F}_N,
\end{eqnarray} 
is true.
\end{leme}
\begin{proof}
The proof of Lemma \ref{vitasja114} we lead by induction  down. Let us prove the equality
$$\mu_{\{\omega_N^1, \omega_N^2\}}^{(\omega_1, \ldots,\omega_{N-1})}(A)= $$
\begin{eqnarray}\label{2119vitasja119}
 \sum\limits_{i_N=1}^2 \psi_N(\omega_1, \ldots, \omega_{N-1}, \omega_N^{i_N})\chi_{A}(\omega_1, \ldots, \omega_{N-1}, \omega_N^{i_N}). 
\end{eqnarray}
Really,
$$ \psi_N(\omega_1, \ldots, \omega_{N-1}, \omega_N^{1})\chi_{A}(\omega_1, \ldots, \omega_{N-1}, \omega_N^{1})+$$ 
$$  \psi_N(\omega_1, \ldots, \omega_{N-1}, \omega_N^{2})\chi_{A}(\omega_1, \ldots, \omega_{N-1}, \omega_N^{2})=$$

$$\left[\chi_{\Omega_N^-}(\omega_1, \ldots,\omega_{N-1},\omega_N^1) \chi_{\Omega_N^+}(\omega_1, \ldots,\omega_{N-1},\omega_N^2) \frac{\Delta S_N^+(\omega_1, \ldots,\omega_{N-1}, \omega_N^2)}{V_N(\omega_1, \ldots,\omega_{N-1}, \omega_N^1,  \omega_N^2)}+\right.  $$

$$\left.\chi_{\Omega_N^-}(\omega_1, \ldots, \omega_{N-1},\omega_N^1) \chi_{\Omega_N^+}(\omega_1, \ldots,\omega_{N-1},\omega_N^1) \frac{\Delta S_N^-(\omega_1, \ldots,\omega_{N-1}, \omega_N^1)}{V_N(\omega_1, \ldots,\omega_{N-1}, \omega_N^1,  \omega_N^2)}\right]\times $$ $$\chi_{A}(\omega_1, \ldots, \omega_{N-1}, \omega_N^{1})+  $$

$$\left[\chi_{\Omega_N^-}(\omega_1, \ldots,\omega_{N-1},\omega_N^2) \chi_{\Omega_N^+}(\omega_1, \ldots,\omega_{N-1},\omega_N^2) \frac{\Delta S_N^+(\omega_1, \ldots,\omega_{N-1}, \omega_N^2)}{V_N(\omega_1, \ldots,\omega_{N-1}, \omega_N^1,  \omega_N^2)}+\right.  $$

$$\left.\chi_{\Omega_N^-}(\omega_1, \ldots,\omega_{N-1}, \omega_N^1) \chi_{\Omega_N^+}(\omega_1, \ldots, \omega_{N-1}, \omega_N^2) \frac{\Delta S_N^-(\omega_1, \ldots,\omega_{N-1}, \omega_N^1)}{V_N(\omega_1, \ldots,\omega_{N-1}, \omega_N^1,  \omega_N^2)}\right]\times$$ $$\chi_{A}(\omega_1, \ldots, \omega_{N-1}, \omega_N^{2})=  $$

$$ \chi_{\Omega_N^-}(\omega_1, \ldots, \omega_{N-1}, \omega_N^1) \chi_{\Omega_N^+}(\omega_1, \ldots, \omega_{N-1}, \omega_N^2)\times$$
 $$ \left[\frac{\Delta S_N^+(\omega_1, \ldots,\omega_{N-1}, \omega_N^2)}{V_N(\omega_1, \ldots,\omega_{N-1}, \omega_N^1,  \omega_N^2)}\chi_{A}(\omega_1, \ldots, \omega_{N-1}, \omega_N^{1})+ \right.$$
\begin{eqnarray}\label{2vitasja116}
\left. \frac{\Delta S_N^-(\omega_1, \ldots,\omega_{N-1}, \omega_N^1)}{V_N(\omega_1, \ldots,\omega_{N-1}, \omega_N^1,  \omega_N^2)}
\chi_{A}(\omega_1, \ldots, \omega_{N-1}, \omega_N^{2})\right], \quad A \in {\cal F}_N.
\end{eqnarray}
The last prove the needed. Suppose that we proved that the equality
$$\mu_{\{\omega_{n+1}^1,\omega_{n+1}^2\},\ldots ,\{\omega_N^1, \omega_N^2\}}^{(\omega_1, \ldots,\omega_{n-1}, \omega_n)}(A)=$$

$$\sum\limits_{i_{n+1}=1}^2\ldots \sum\limits_{i_N=1}^2\prod\limits_{j=n+1}^N\psi_j(\omega_1, \ldots, \omega_{n}, \omega_{n+1}^{i_{n+1}}, \ldots, \omega_j^{i_j})\chi_{A}(\omega_1, \ldots, \omega_{n},\omega_{n+1}^{i_{n+1}}, \ldots, \omega_N^{i_N}), $$
\begin{eqnarray}\label{vitasja117}
A \in {\cal F}_N,
\end{eqnarray} 
is true. By the same way as above, we have

$$\sum\limits_{i_n=1}^2\psi_n(\omega_1, \ldots, \omega_{n-1}, \omega_{n}^{i_{n}})\mu_{\{\omega_{n+1}^1,\omega_{n+1}^2\},\ldots ,\{\omega_N^1, \omega_N^2\}}^{(\omega_1, \ldots,\omega_{n-1}, \omega_n^{i_n})}(A)=$$

$$\chi_{\Omega_n^-}(\omega_1, \ldots,\omega_{n-1},\omega_n^1) \chi_{\Omega_n^+}(\omega_1, \ldots,\omega_{n-1},\omega_n^2)\times$$

$$ \left[\frac{\Delta S_n^+(\omega_1, \ldots,\omega_{n-1}, \omega_n^2)}{V_n(\omega_1, \ldots,\omega_{n-1}, \omega_n^1,  \omega_n^2)}\mu_{\{\omega_{n+1}^1,\omega_{n+1}^2\},\ldots ,\{\omega_N^1, \omega_N^2\}}^{(\omega_1, \ldots,\omega_{n-1}, \omega_n^1)}(A)+ \right.$$
$$\left. \frac{\Delta S_n^-(\omega_1, \ldots,\omega_{n-1}, \omega_n^1)}{V_n(\omega_1, \ldots,\omega_{n-1}, \omega_n^1,  \omega_n^2)}\mu_{\{\omega_{n+1}^1,\omega_{n+1}^2\},\ldots ,\{\omega_N^1, \omega_N^2\}}^{(\omega_1, \ldots,\omega_{n-1}, \omega_n^2)}(A)\right]=$$
\begin{eqnarray}\label{2vitasja118}
\mu_{\{\omega_n^1,\omega_n^2\},\ldots ,\{\omega_N^1, \omega_N^2\}}^{(\omega_1, \ldots,\omega_{n-1})}(A),  \quad A \in {\cal F}_N.
\end{eqnarray}
The last proves Lemma \ref{vitasja114}.
\end{proof}

Let us define the integral for the random value $f_N(\omega_1, \ldots,\omega_{N-1}, \omega_N)$ relative to the measure $\mu_{\{\omega_1^1,\omega_1^2\},\ldots ,\{\omega_N^1, \omega_N^2\}}(A)$ by the formula

$$ \int\limits_{\Omega_N}f_N(\omega_1, \ldots,\omega_{N-1}, \omega_N) d\mu_{\{\omega_1^1,\omega_1^2\},\ldots ,\{\omega_N^1, \omega_N^2\}}=
$$
\begin{eqnarray}\label{vitasja120}
\sum\limits_{i_1=1}^2\ldots \sum\limits_{i_N=1}^2\prod\limits_{j=1}^N\psi_j(\omega_1^{i_1}, \ldots, \omega_j^{i_j})f_N(\omega_1^{i_1}, \ldots, \omega_N^{i_N}).
\end{eqnarray} 

To describe the convex set of equivalent martingale measures, we introduce the family of  $\alpha$-spot measures, depending on the point $(\{\omega_1^{1},\{\omega_1^{2} \}, \ldots, \{\omega_N^{1},\{\omega_N^{2} \})$ belonging to $\Omega_N\times \Omega_N$ and  the set of strictly positive random values
\begin{eqnarray}\label{witakolja1}
\alpha_n(\{\omega_1^1, \ldots,\omega_{n-1}^1, \omega_{n}^1\};\{\omega_1^2, \ldots, \omega_{n-1}^2, \omega_{n}^{2}\}), \quad n=\overline{1,N},
\end{eqnarray}
at points $W_n=(\{\omega_1^{1}, \ldots, \omega_n^{1}\};\{\omega_1^{2}, \ldots, \omega_n^{2}\}),$ being  constructed by the point $(\{\omega_1^{1},\omega_1^{2} \}, \ldots, \{\omega_N^{1}, \omega_N^{2} \}).$ 

Further, in this section, we assume that the evolution of risky asset is  given by the formula  (\ref{tin1vitasika1ja9}). Therefore, in this case
\begin{eqnarray}\label{pupwitakoljapup1}
\Omega_n^-=\Omega_n^{0-}\times \Omega_{n-1}, \quad \Omega_n^+=\Omega_n^{0+}\times \Omega_{n-1}, \quad n=\overline{1,N},
\end{eqnarray}
and the condition (\ref{3vitasja7}) is formulated, as follows: 
$$ \int\limits_{\Omega_n^{0}\times \Omega_n^{0}}\chi_{\Omega_n^{0-}}
(\omega_n^1)\chi_{\Omega_n^{0+}}(\omega_n^2)\alpha_n(\{\omega_1^1, \ldots,\omega_{n-1}^1, \omega_{n}^1\};\{\omega_1^2, \ldots, \omega_{n-1}^2, \omega_{n}^{2}\})\times$$
\begin{eqnarray}\label{pupwitakolja1} 
d P_n^{0}(\omega_n^1)  d P_n^{0}(\omega_n^2)=1, \quad n=\overline{1,N}. 
\end{eqnarray}

Let us determine  the random values

$$ \psi_n^{\alpha}(\omega_1, \ldots,\omega_n)= 
\chi_{\Omega_n^-}(\omega_1, \ldots,\omega_{n-1}, \omega_n) \psi_n^1(\omega_1, \ldots,\omega_n)+$$
\begin{eqnarray}\label{vitakolja111}
\chi_{\Omega_n^+}(\omega_1, \ldots,\omega_{n-1}, \omega_n) \psi_n^2(\omega_1, \ldots,\omega_n),
\end{eqnarray}

$$\psi_n^1(\omega_1, \ldots,\omega_{n-1},\omega_n)=$$
 $$\alpha_n(\{\omega_1^1, \ldots,\omega_{n-1}^1, \omega_{n}^1\};\{\omega_1^2, \ldots, \omega_{n-1}^2, \omega_{n}^{2}\})\chi_{\Omega_n^+}(\omega_1, \ldots,\omega_{n-1}, \omega_n^2) \times $$
\begin{eqnarray}\label{vitakolja112}
\frac{\Delta S_n^+(\omega_1, \ldots,\omega_{n-1}, \omega_n^2)}{V_n(\omega_1, \ldots,\omega_{n-1}, \omega_n^1,  \omega_n^2)}, \quad (\omega_1, \ldots,\omega_{n-1}) \in \Omega_{n-1},
\end{eqnarray} 

$$\psi_n^2(\omega_1, \ldots,\omega_{n-1},\omega_n)=$$ $$\alpha_n(\{\omega_1^1, \ldots,\omega_{n-1}^1, \omega_{n}^1\};\{\omega_1^2, \ldots, \omega_{n-1}^2, \omega_{n}^{2}\})\chi_{\Omega_n^-}(\omega_1, \ldots,\omega_{n-1}, \omega_n^1)  \times $$
\begin{eqnarray}\label{vitakolja113}
\frac{\Delta S_n^-(\omega_1, \ldots,\omega_{n-1}, \omega_n^1)}{V_n(\omega_1, \ldots,\omega_{n-1}, \omega_n^1,  \omega_n^2)}, \quad  (\omega_1, \ldots,\omega_{n-1}) \in \Omega_{n-1}.
\end{eqnarray} 
Let us define the set of $\alpha$-spot measures on the $\sigma$-algebra ${\cal F}_N$ by the formula
\begin{eqnarray}\label{vitakolja114}
\mu^{\alpha}_{W_N}(A)=\sum\limits_{i_1=1}^2\ldots \sum\limits_{i_N=1}^2\prod\limits_{j=1}^N\psi_j^{\alpha}(\omega_1^{i_1}, \ldots, \omega_j^{i_j})\chi_A(\omega_1^{i_1}, \ldots, \omega_N^{i_N}), \quad A \in {\cal F}_N,
\end{eqnarray} 
and the set of the measures
$$\mu_0(A)=$$
\begin{eqnarray}\label{3vitakolja7}
\int\limits_{\Omega_N\times \Omega_N}\sum\limits_{i_1=1}^2\ldots \sum\limits_{i_N=1}^2\prod\limits_{j=1}^N\psi_j^{\alpha}(\omega_1^{i_1}, \ldots, \omega_j^{i_j})\chi_A(\omega_1^{i_1}, \ldots, \omega_N^{i_N})d P_N\times d P_N, \ A \in {\cal F}_N.
\end{eqnarray}

\begin{te}\label{vitasja121}
Suppose that the conditions of  Lemma \ref{witka0}
are true. If the strictly positive random values
\begin{eqnarray}\label{vitakolja115}
\alpha_n(\{\omega_1^1, \ldots,\omega_n^1\};\{\omega_1^2, \ldots,\omega_n^2\}), \ n=\overline{1,N}, 
\end{eqnarray} 
 given on the 
probability space $\{\Omega_n\times \Omega_n, {\cal F}_n \times {\cal F}_n,  P_n \times  P_n \}, \ n=\overline{1,N}, $
satisfy the conditions (\ref{pupwitakolja1}),
then for the  measure $\mu_0(A),$  given by the formula (\ref{3vitakolja7}),
the representation
$$\mu_0(A)=$$
\begin{eqnarray}\label{vitasja122}
\int\limits_{\Omega_N\times \Omega_N}\prod\limits_{i=1}^N\alpha_{i}
(\{\omega_1^1, \ldots ,\omega_i^1\}; \{\omega_1^2,\ldots ,\omega_i^2\}) \mu_{\{\omega_1^1,\omega_1^2\},\ldots ,\{\omega_N^1, \omega_N^2\}}(A) d P_N \times d P_N
\end{eqnarray}
is true.
\end{te}
\begin{proof} Due to Lemma \ref{witka0}, the set of random values  $ \alpha_n(\{\omega_1^1, \ldots,\omega_n^1\};\{\omega_1^2, \ldots,\omega_n^2\}),$  $ \ n=\overline{1,N}, $ satisfying the conditions  (\ref{pupwitakolja1}), is a non empty set.    

We prove Theorem \ref{vitasja121} by induction  down. For the   spot measure the  relation
$$\mu_{\{\omega_N^1, \omega_N^2\}}^{(\omega_1, \ldots,\omega_{N-1})}(A)= $$
$$ \chi_{\Omega_N^-}(\omega_1, \ldots, \omega_{N-1}, \omega_N^1) \chi_{\Omega_N^+}(\omega_1, \ldots, \omega_{N-1}, \omega_N^2)\times$$
 $$ \left[\frac{\Delta S_N^+(\omega_1, \ldots,\omega_{N-1}, \omega_N^2)}{V_N(\omega_1, \ldots,\omega_{N-1}, \omega_N^1,  \omega_N^2)}\chi_{A}(\omega_1, \ldots, \omega_{N-1}, \omega_N^{1})+ \right.$$
\begin{eqnarray}\label{2vitasja122}
\left. \frac{\Delta S_N^-(\omega_1, \ldots,\omega_{N-1}, \omega_N^1)}{V_N(\omega_1, \ldots,\omega_{N-1}, \omega_N^1,  \omega_N^2)}
\chi_{A}(\omega_1, \ldots, \omega_{N-1}, \omega_N^{2})\right], \quad A \in {\cal F}_N,
\end{eqnarray}
is true.
Multiplying the relation (\ref{2vitasja122}) on $\alpha_N(\{\omega_1^1, \ldots, \omega_{N-1}^1, \omega_N^1\}; \{\omega_1^2, \ldots, \omega_{N-1}^2, \omega_N^2\})$ and after that,
integrating  relative to the measure $P_N^0\times P_N^0$ on the set $ \Omega_N^0\times \Omega_N^0,$ we obtain
$$\int\limits_{\Omega_N^0} \int\limits_{\Omega_N^0}
\alpha_N(\{\omega_1^1, \ldots, \omega_{N-1}^1, \omega_N^1\}; \{\omega_1^2, \ldots, \omega_{N-1}^2, \omega_N^2\})\times$$
$$\mu_{\{\omega_N^1, \omega_N^2\}}^{(\omega_1, \ldots,\omega_{N-1})}(A) d P_N^0(\omega_N^1)d P_N^0(\omega_N^2)= $$
$$ \int\limits_{\Omega_N^0} \int\limits_{\Omega_N^0} \alpha_N(\{\omega_1^1, \ldots, \omega_{N-1}^1, \omega_N^1\}; \{\omega_1^2, \ldots, \omega_{N-1}^2, \omega_N^2\})\times$$ 
$$ \chi_{\Omega_N^-}(\omega_1, \ldots, \omega_{N-1}, \omega_N^1) \chi_{\Omega_N^+}(\omega_1, \ldots, \omega_{N-1}, \omega_N^2)\times$$
 $$ \left[\frac{\Delta S_N^+(\omega_1, \ldots,\omega_{N-1}, \omega_N^2)}{V_N(\omega_1, \ldots,\omega_{N-1}, \omega_N^1,  \omega_N^2)}\chi_{A}(\omega_1, \ldots, \omega_{N-1}, \omega_N^{1})+ \right.$$

$$\left. \frac{\Delta S_N^-(\omega_1, \ldots,\omega_{N-1}, \omega_N^1)}{V_N(\omega_1, \ldots,\omega_{N-1}, \omega_N^1,  \omega_N^2)}
\chi_{A}(\omega_1, \ldots, \omega_{N-1}, \omega_N^{2})\right] d P_N^0(\omega_N^1)d P_N^0(\omega_N^2)=$$
\begin{eqnarray}\label{5vitasja122}
\mu_{N-1}^{(\omega_1, \ldots,\omega_{N-1})}(A), \quad A \in {\cal F}_N.
\end{eqnarray}

Suppose that  we proved  the equality

$$\int\limits_{\prod\limits_{i=n+1}^N[\Omega_i^0\times \Omega_i^0]}   
\prod\limits_{i=n+1}^N \alpha_i(\{\omega_1^1, \ldots,\omega_n^1, \omega_{n+1}^1, \ldots, \omega_{i}^1\}; \{\omega_1^2, \ldots,\omega_n^2, \omega_{n+1}^2, \ldots, \omega_{i}^2\})\times$$
\begin{eqnarray}\label{5vitasja123}
 \mu_{\{\omega_{n+1}^1,\omega_{n+1}^2\},\ldots ,\{\omega_N^1, \omega_N^2\}}^{(\omega_1, \ldots,\omega_{n})}(A)\prod\limits_{i=n+1}^N d P_{i}^0(\omega_i^1)d P_{i}^0(\omega_i^2)=\mu_n^{(\omega_1, \ldots,\omega_n)}(A).
\end{eqnarray}
Then, using the induction supposition (\ref{5vitasja123}), the relation for the spot measure
$$\mu_{\{\omega_n^1,\omega_n^2\},\ldots ,\{\omega_N^1, \omega_N^2\}}^{(\omega_1, \ldots,\omega_{n-1})}(A)= $$
$$\chi_{\Omega_n^-}(\omega_1, \ldots,\omega_{n-1}, \omega_n^1) \chi_{\Omega_n^+}(\omega_1, \ldots,\omega_{n-1}, \omega_n^2)\times$$

$$ \left[\frac{\Delta S_n^+(\omega_1, \ldots,\omega_{n-1}, \omega_n^2)}{V_n(\omega_1, \ldots,\omega_{n-1}, \omega_n^1,  \omega_n^2)}\mu_{\{\omega_{n+1}^1,\omega_{n+1}^2\},\ldots ,\{\omega_N^1, \omega_N^2\}}^{(\omega_1, \ldots,\omega_{n-1}, \omega_n^1)}(A)+ \right.$$
\begin{eqnarray}\label{5vitasja124}
\left. \frac{\Delta S_n^-(\omega_1, \ldots,\omega_{n-1}, \omega_n^1)}{V_n(\omega_1, \ldots,\omega_{n-1}, \omega_n^1,  \omega_n^2)}\mu_{\{\omega_{n+1}^1,\omega_{n+1}^2\},\ldots ,\{\omega_N^1, \omega_N^2\}}^{(\omega_1, \ldots,\omega_{n-1}, \omega_n^2)}(A)\right],  \quad A \in {\cal F}_N,
\end{eqnarray}
and multiplying it  on 
$\prod\limits_{i=n}^N \alpha_i(\{\omega_1^1, \ldots,\omega_{n-1}^1, \omega_{n}^1, \ldots, \omega_{i}^1\}; \{\omega_1^2, \ldots,\omega_{n-1}^2, \omega_{n}^2, \ldots, \omega_{i}^2\}) $  and then integrating  relative to the measure $\prod\limits_{i=n}^N[P_i^0\times P_i^0]$  on the set $\prod\limits_{i=n}^N[\Omega_i^0\times \Omega_i^0],$ we obtain the equality

$$\int\limits_{\Omega_n^0\times \Omega_n^0} \chi_{\Omega_n^-}(\omega_1, \ldots,\omega_{n-1}, \omega_n^1) \chi_{\Omega_n^+}(\omega_1, \ldots,\omega_{n-1}, \omega_n^2)\times$$

$$\alpha_n(\{\omega_1^1, \ldots,\omega_n^1\};\{\omega_1^2, \ldots,\omega_n^2\}) \left[\frac{\Delta S_n^+(\omega_1, \ldots,\omega_{n-1}, \omega_n^2)}{V_n(\omega_1, \ldots,\omega_{n-1}, \omega_n^1,  \omega_n^2)}\mu_{n}^{(\omega_1, \ldots,\omega_{n-1}, \omega_n^1)}(A)+ \right.$$

$$ \left. \frac{\Delta S_n^-(\omega_1, \ldots,\omega_{n-1}, \omega_n^1)}{V_n(\omega_1, \ldots,\omega_{n-1}, \omega_n^1,  \omega_n^2)}\mu_{n}^{(\omega_1, \ldots,\omega_{n-1}, \omega_n^2)}(A)\right]dP_n^0(\omega_n^1)dP_n^0(\omega_n^2)=$$
\begin{eqnarray}\label{5vitasja125}
\mu_{n-1}^{(\omega_1, \ldots,\omega_{n-1})}(A), \quad n=\overline{1,N}.
\end{eqnarray}
Thus, we proved the following recurrent relations
$$\mu_{n-1}^{(\omega_1, \ldots,\omega_{n-1})}(A)=\int\limits_{\Omega_n^0\times \Omega_n^0} \chi_{\Omega_n^-}(\omega_1, \ldots, \omega_{n-1}, \omega_n^1) \chi_{\Omega_n^+}(\omega_1, \ldots,\omega_{n-1}, \omega_n^2)\times$$
$$\alpha_n(\{\omega_1^1, \ldots,\omega_n^1\};\{\omega_1^2, \ldots,\omega_n^2\}) \left[\frac{\Delta S_n^+(\omega_1, \ldots,\omega_{n-1}, \omega_n^2)}{V_n(\omega_1, \ldots,\omega_{n-1}, \omega_n^1,  \omega_n^2)}\mu_{n}^{(\omega_1, \ldots,\omega_{n-1}, \omega_n^1)}(A)+ \right.$$
\begin{eqnarray}\label{5vitasja126}
\left. \frac{\Delta S_n^-(\omega_1, \ldots,\omega_{n-1}, \omega_n^1)}{V_n(\omega_1, \ldots,\omega_{n-1}, \omega_n^1,  \omega_n^2)}\mu_n^{(\omega_1, \ldots,\omega_{n-1}, \omega_n^2)}(A)\right]dP_n^0(\omega_n^1)dP_n^0(\omega_n^2), \quad n=\overline{1,N}.
\end{eqnarray}
To finish the proof of Theorem \ref{vitasja121}, let us calculate
\begin{eqnarray}\label{119vitasja119}
 \int\limits_{\Omega_N^0\times \Omega_N^0} \sum\limits_{i_N=1}^2 \psi_N^{\alpha}(\omega_1, \ldots, \omega_{N-1}, \omega_N^{i_N})\chi_{A}(\omega_1, \ldots, \omega_{N-1}, \omega_N^{i_N}) d P_N^0(\omega_N^1) d P_N^0(\omega_N^2).
\end{eqnarray}
 Calculating the expression
$$ \sum\limits_{i_N=1}^2 \psi_N^{\alpha}(\omega_1, \ldots, \omega_{N-1}, \omega_N^{i_N})\chi_{A}(\omega_1, \ldots, \omega_{N-1}, \omega_N^{i_N})=$$
$$ \psi_N^{\alpha}(\omega_1, \ldots, \omega_{N-1}, \omega_N^{1})\chi_{A}(\omega_1, \ldots, \omega_{N-1}, \omega_N^{1})+$$ 
$$  \psi_N^{\alpha}(\omega_1, \ldots, \omega_{N-1}, \omega_N^{2})\chi_{A}(\omega_1, \ldots, \omega_{N-1}, \omega_N^{2})=$$

$$\alpha_N(\{\omega_1^1, \ldots,\omega_N^1\};\{\omega_1^2, \ldots,\omega_N^2\})\times $$
$$ \chi_{\Omega_N^-}(\omega_1, \ldots, \omega_{N-1}, \omega_N^1) \chi_{\Omega_N^+}(\omega_1, \ldots, \omega_{N-1}, \omega_N^2)\times$$
 $$ \left[\frac{\Delta S_N^+(\omega_1, \ldots,\omega_{N-1}, \omega_N^2)}{V_N(\omega_1, \ldots,\omega_{N-1}, \omega_N^1,  \omega_N^2)}\chi_{A}(\omega_1, \ldots, \omega_{N-1}, \omega_N^{1})+ \right.$$
\begin{eqnarray}\label{119vitasja116}
\left. \frac{\Delta S_N^-(\omega_1, \ldots,\omega_{N-1}, \omega_N^1)}{V_N(\omega_1, \ldots,\omega_{N-1}, \omega_N^1,  \omega_N^2)}
\chi_{A}(\omega_1, \ldots, \omega_{N-1}, \omega_N^{2})\right], \quad A \in {\cal F}_N,
\end{eqnarray}
and substituting (\ref{119vitasja116}) into (\ref{119vitasja119}), we obtain the equality
$$\int\limits_{\Omega_N^0\times \Omega_N^0} \sum\limits_{i_N=1}^2 \psi_N^{\alpha}(\omega_1, \ldots, \omega_{N-1}, \omega_N^{i_N})\chi_{A}(\omega_1, \ldots, \omega_{N-1}, \omega_N^{i_N}) d P_N^0(\omega_N^1) d P_N^0(\omega_N^2)=$$
\begin{eqnarray}\label{119vitasja120}
\mu_{N-1}^{(\omega_1, \ldots, \omega_{N-1})}(A).
\end{eqnarray}
Suppose that we already proved the equality
$$\int\limits_{\prod\limits_{i=n+1}^N\Omega_i^0\times \Omega_i^0} \sum\limits_{i_{n+1}=1}^2\ldots \sum\limits_{i_N=1}^2\prod\limits_{j=1}^N\psi_j^{\alpha}(\omega_1, \ldots, \omega_n, \omega_{n+1}^{i_{n+1}}\ldots, \omega_j^{i_j}) \prod\limits_{i=n+1}^N d P_i^0(\omega_i^1) d P_i^0(\omega_i^2)=$$
\begin{eqnarray}\label{119vitasja121}
\mu_{n}^{(\omega_1, \ldots, \omega_{n})}(A).
\end{eqnarray}
Then,  acting as above, we obtain the equalities
$$\int\limits_{\Omega_n^0\times \Omega_n^0} \sum\limits_{i_n=1}^2 \psi_n^{\alpha}(\omega_1, \ldots, \omega_{n-1}, \omega_n^{i_n})\mu_n^{(\omega_1, \ldots, \omega_{n-1}, \omega_n^{i_n})}(A)d P_n^0(\omega_n^1) d P_n^0(\omega_n^2)=$$
$$ \int\limits_{\Omega_n^0\times \Omega_n^0} \alpha_n(\{\omega_1^1, \ldots,\omega_N^1\};\{\omega_1^2, \ldots,\omega_N^2\})\times $$
$$ \chi_{\Omega_n^-}(\omega_1, \ldots, \omega_{n-1}, \omega_n^1) \chi_{\Omega_n^+}(\omega_1, \ldots, \omega_{n-1}, \omega_n^2)\times$$
 $$ \left[\frac{\Delta S_n^+(\omega_1, \ldots,\omega_{n-1}, \omega_n^2)}{V_n(\omega_1, \ldots,\omega_{n-1}, \omega_n^1,  \omega_n^2)}
\mu_n^{(\omega_1, \ldots, \omega_{n-1}, \omega_n^{1})}(A)
+ \right.$$

$$\left. \frac{\Delta S_n^-(\omega_1, \ldots,\omega_{n-1}, \omega_n^1)}{V_n(\omega_1, \ldots,\omega_{n-1}, \omega_n^1,  \omega_n^2)}
\mu_n^{(\omega_1, \ldots, \omega_{n-1}, \omega_n^{2})}(A)
\right] d P_n^0(\omega_n^1) d P_n^0(\omega_n^2)=$$
\begin{eqnarray}\label{119vitasja130}
\mu_{n-1}^{(\omega_1, \ldots, \omega_{n-1})}(A), \quad A \in {\cal F}_N.
\end{eqnarray}
We proved that the recurrent relations (\ref{119vitasja130}) are the same as the recurrent relations (\ref{5vitasja126}).
This proves Theorem \ref{vitasja121}.
\end{proof}
Let us introduce the denotations
$$\mu_{\{\omega_1^1,\omega_1^2\},\ldots ,\{\omega_N^1, \omega_N^2\}}(\Omega_N)=\sum\limits_{i_1=1}^2\ldots \sum\limits_{i_N=1}^2\prod\limits_{j=1}^N\psi_j(\omega_1^{i_1}, \ldots, \omega_j^{i_j}),$$
\begin{eqnarray}\label{500vitasja500}
 W_N= \{\omega_1^1, \ldots, \omega_N^1; \omega_1^2,\ldots, \omega_N^2\}=\{\{\omega\}_N^1, \{\omega\}_N^2\}.
\end{eqnarray}
Further,  only those points $(\{\omega_1^1, \omega_1^2\}, \ldots, \{\omega_N^1, \omega_N^2\}) \in \Omega_N \times \Omega_N$ play important role for which 
$\mu_{\{\omega_1^1,\omega_1^2\},\ldots ,\{\omega_N^1, \omega_N^2\}}(\Omega_N) \neq 0.$

Below, in the next two Theorems, we assume that
 the random value
\begin{eqnarray}\label{pupetsvitakolja115}
\alpha_n(\{\omega_1^1, \ldots,\omega_n^1\};\{\omega_1^2, \ldots,\omega_n^2\})
\end{eqnarray} 
 given on the 
probability space $\{\Omega_n\times \Omega_n, {\cal F}_n \times {\cal F}_n,  P_n \times  P_n \}, \ n=\overline{1,N}, $
satisfy the conditions (\ref{pupwitakolja1}).

Under the above conditions, for the measure $\mu_0(A),$ given by the formula (\ref{vitasja122}), the representation
 \begin{eqnarray}\label{2vitakolja10}
\mu_0(A)=\int\limits_{\Omega_N}\prod\limits_{n=1}^N\psi_n(\omega_1, \ldots,\omega_n)\chi_{A}(\omega_1, \ldots, \omega_N) \prod\limits_{i=1}^Nd P_i^0(\omega_i)
\end{eqnarray}
 is true, where 
$$ \psi_n(\omega_1, \ldots,\omega_n)=\chi_{\Omega_n^-}(\omega_1, \ldots,\omega_{n-1}, \omega_n) \psi_n^1(\omega_1, \ldots,\omega_n)+$$
\begin{eqnarray}\label{pupets2vitakolja11}
\chi_{\Omega_n^+}(\omega_1, \ldots,\omega_{n-1}, \omega_n) \psi_n^2(\omega_1, \ldots,\omega_n),
\end{eqnarray}
$$\psi_n^1(\omega_1, \ldots,\omega_{n-1},\omega_n)=\int\limits_{\Omega_n^0}\chi_{\Omega_n^+}(\omega_1, \ldots,\omega_{n-1}, \omega_n^2)\alpha_n(\{\omega_1^1, \ldots,\omega_n^1\};\{\omega_1^2, \ldots,\omega_n^2\}) \times $$
\begin{eqnarray}\label{pupets2vitakolja12}
\frac{\Delta S_n^+(\omega_1, \ldots,\omega_{n-1}, \omega_n^2)}{V_n(\omega_1, \ldots,\omega_{n-1}, \omega_n^1,  \omega_n^2)}d P_n^0(\omega_n^2), \quad (\omega_1, \ldots,\omega_{n-1}) \in \Omega_{n-1},
\end{eqnarray} 

$$\psi_n^2(\omega_1, \ldots,\omega_{n-1},\omega_n)=\int\limits_{\Omega_n^0}\chi_{\Omega_n^-}(\omega_1, \ldots,\omega_{n-1}, \omega_n^1)\alpha_n(\{\omega_1^1, \ldots,\omega_n^1\};\{\omega_1^2, \ldots,\omega_n^2\}) \times $$
\begin{eqnarray}\label{pupets2vitakolja13}
\frac{\Delta S_n^-(\omega_1, \ldots,\omega_{n-1}, \omega_n^1)}{V_n(\omega_1, \ldots,\omega_{n-1}, \omega_n^1,  \omega_n^2)}d P_n^0(\omega_n^1), \quad  (\omega_1, \ldots,\omega_{n-1}) \in \Omega_{n-1}.
\end{eqnarray} 
Due to the conditions (\ref{pupwitakolja1}) relative to the random values $\alpha_n(\{\omega\}_n^1;\{\omega\}_n^2),$ we have
\begin{eqnarray}\label{puppups2vitakolja13}
\int\limits_{\Omega_n^0} \psi_n(\omega_1, \ldots,\omega_n) d P_n^0(\omega_n)=1, \quad n=\overline{1,N}.
\end{eqnarray} 
for $ \psi_n(\omega_1, \ldots,\omega_n),$ given by the formula (\ref{pupets2vitakolja11}). The proof of the  equalities (\ref{puppups2vitakolja13}) is the same as in Theorem \ref{witka2}.

\begin{te}\label{10vitakolja121}
Suppose that the conditions of Lemma \ref{witka0} are true.
Then,  the set of strictly positive  random values  $\alpha_n(\{\omega\}_n^1; \{\omega\}_n^2), n=\overline{1,N},$ satisfying the conditions 
$$E^{\mu_0}|\Delta S_n(\omega_1, \ldots,\omega_{n-1}, \omega_n)|=$$
\begin{eqnarray}\label{2pupsavitasja2}
\int\limits_{\Omega_N}\prod\limits_{i=1}^N\psi_i(\omega_1, \ldots,\omega_i) |\Delta S_n(\omega_1, \ldots,\omega_{n-1}, \omega_n)| \prod\limits_{i=1}^Nd P_i^0(\omega_i)<\infty, \quad n=\overline{1,N},
\end{eqnarray}
is a non empty set   for the  measures  $\mu_0(A),$  given by the formula (\ref{vitasja122}). The measure $\mu_0(A),$ constructed by the
 strictly positive  random values  $\alpha_n(\{\omega\}_n^1; \{\omega\}_n^2), n=\overline{1,N},$ satisfying the conditions (\ref{pupwitakolja1}), (\ref{2pupsavitasja2}) is  a   martingale measure for the  evolution of risky asset, given by the formula (\ref{tin1vitasika1ja9}). Every measure, belonging to the convex linear span of such measures, is also
martingale measure for the  evolution of risky asset, given by the formula (\ref{tin1vitasika1ja9}).  They are  equivalent  to the measure $P_N.$ 
 The set of spot measures  $\mu_{\{\omega_1^1,\omega_1^2\},\ldots ,\{\omega_N^1, \omega_N^2\}}(A)$  is a set of martingale measures for the  evolution of risky asset, given by the formula (\ref{tin1vitasika1ja9}).
\end{te}
\begin{proof}
The first fact, that the set of random values  $\alpha_n(\{\omega\}_n^1; \{\omega\}_n^2), n=\overline{1,N},$  satisfying the conditions (\ref{pupwitakolja1}), (\ref{2pupsavitasja2}) is a non empty one, follows from Lemma \ref{witka0}.
From the representation (\ref{2vitakolja10}) for the set of measures $\mu_0(A)$, given by the formula (\ref{vitasja122}), as in the proof of Theorem \ref{witka2}, it is proved that this set of measures  is a set of martingale measures being    equivalent  to the measure $P_N.$ 

Let us prove the last statement of Theorem \ref{10vitakolja121}.
Since for the spot measure $\mu_{\{\omega_1^1,\omega_1^2\},\ldots ,\{\omega_N^1, \omega_N^2\}}(A)$ the representation 
$$\mu_{\{\omega_1^1,\omega_1^2\},\ldots ,\{\omega_N^1, \omega_N^2\}}(A)=$$
\begin{eqnarray}\label{4vitakolja115}
\sum\limits_{i_1=1}^2\ldots \sum\limits_{i_N=1}^2\prod\limits_{j=1}^N\psi_j(\omega_1^{i_1}, \ldots, \omega_j^{i_j})\chi_{A}(\omega_1^{i_1}, \ldots, \omega_N^{i_N}), \quad A \in {\cal F}_N,
\end{eqnarray}
is true, let us calculate
$$\sum\limits_{i_j=1}^2\psi_j(\omega_1^{i_1}, \ldots, \omega_j^{i_j})=\psi_j(\omega_1^{i_1}, \ldots, \omega_{j-1}^{i_{j-1}},\omega_j^{1})+
\psi_j(\omega_1^{i_1}, \ldots,\omega_{j-1}^{i_{j-1}}, \omega_j^{2})=$$
$$ \chi_{\Omega_j^-}(\omega_1^{i_1}, \ldots, \omega_{j-1}^{i_{j-1}}, \omega_j^{1})\psi_j^1(\omega_1^{i_1}, \ldots, \omega_{j-1}^{i_{j-1}} \omega_j^{1})+$$
$$\chi_{\Omega_n^+}(\omega_1^{i_1}, \ldots, \omega_{j-1}^{i_{j-1}}, \omega_j^{1}) \psi_j^2(\omega_1^{i_1}, \ldots, \omega_{j-1}^{i_{j-1}} \omega_j^{1})+$$
$$ \chi_{\Omega_j^-}(\omega_1^{i_1}, \ldots, \omega_{j-1}^{i_{j-1}}, \omega_j^{2})\psi_j^1(\omega_1^{i_1}, \ldots, \omega_{j-1}^{i_{j-1}} \omega_j^{2})+$$
$$\chi_{\Omega_n^+}(\omega_1^{i_1}, \ldots, \omega_{j-1}^{i_{j-1}}, \omega_j^{2}) \psi_j^2(\omega_1^{i_1}, \ldots, \omega_{j-1}^{i_{j-1}} \omega_j^{2})=$$

$$ \chi_{\Omega_j^-}(\omega_1^{i_1}, \ldots, \omega_{j-1}^{i_{j-1}}, \omega_j^{1})\chi_{\Omega_j^+}(\omega_1^{i_1}, \ldots, \omega_{j-1}^{i_{j-1}}, \omega_j^{2}) \frac{\Delta S_j^+(\omega_1^{i_1}, \ldots, \omega_{j-1}^{i_{j-1}}, \omega_j^{2})}{V_j(\omega_1^{i_1}, \ldots, \omega_{j-1}^{i_{j-1}}, \omega_j^1,  \omega_j^2)}+ $$

$$ \chi_{\Omega_j^+}(\omega_1^{i_1}, \ldots, \omega_{j-1}^{i_{j-1}}, \omega_j^{1})\chi_{\Omega_j^-}(\omega_1^{i_1}, \ldots, \omega_{j-1}^{i_{j-1}}, \omega_j^{1}) \frac{\Delta S_j^-(\omega_1^{i_1}, \ldots, \omega_{j-1}^{i_{j-1}}, \omega_j^{1})}{V_j(\omega_1^{i_1}, \ldots, \omega_{j-1}^{i_{j-1}}, \omega_j^1,  \omega_j^1)}+ $$

$$ \chi_{\Omega_j^-}(\omega_1^{i_1}, \ldots, \omega_{j-1}^{i_{j-1}}, \omega_j^{2})\chi_{\Omega_j^+}(\omega_1^{i_1}, \ldots, \omega_{j-1}^{i_{j-1}}, \omega_j^{2}) \frac{\Delta S_j^+(\omega_1^{i_1}, \ldots, \omega_{j-1}^{i_{j-1}}, \omega_j^{2})}{V_j(\omega_1^{i_1}, \ldots, \omega_{j-1}^{i_{j-1}}, \omega_j^1,  \omega_j^2)}+ $$

$$ \chi_{\Omega_j^+}(\omega_1^{i_1}, \ldots, \omega_{j-1}^{i_{j-1}}, \omega_j^{2})\chi_{\Omega_j^-}(\omega_1^{i_1}, \ldots, \omega_{j-1}^{i_{j-1}}, \omega_j^{1}) \frac{\Delta S_j^-(\omega_1^{i_1}, \ldots, \omega_{j-1}^{i_{j-1}}, \omega_j^{1})}{V_j(\omega_1^{i_1}, \ldots, \omega_{j-1}^{i_{j-1}}, \omega_j^1,  \omega_j^1)}= $$

$$ \chi_{\Omega_j^-}(\omega_1^{i_1}, \ldots, \omega_{j-1}^{i_{j-1}}, \omega_j^{1})\chi_{\Omega_j^+}(\omega_1^{i_1}, \ldots, \omega_{j-1}^{i_{j-1}}, \omega_j^{2}) \frac{\Delta S_j^+(\omega_1^{i_1}, \ldots, \omega_{j-1}^{i_{j-1}}, \omega_j^{2})}{V_j(\omega_1^{i_1}, \ldots, \omega_{j-1}^{i_{j-1}}, \omega_j^1,  \omega_j^2)}+ $$

$$ \chi_{\Omega_j^+}(\omega_1^{i_1}, \ldots, \omega_{j-1}^{i_{j-1}}, \omega_j^{2})\chi_{\Omega_j^-}(\omega_1^{i_1}, \ldots, \omega_{j-1}^{i_{j-1}}, \omega_j^{1}) \frac{\Delta S_j^-(\omega_1^{i_1}, \ldots, \omega_{j-1}^{i_{j-1}}, \omega_j^{1})}{V_j(\omega_1^{i_1}, \ldots, \omega_{j-1}^{i_{j-1}}, \omega_j^1,  \omega_j^1)}= $$

$$ \chi_{\Omega_j^-}(\omega_1^{i_1}, \ldots, \omega_{j-1}^{i_{j-1}}, \omega_j^{1})\chi_{\Omega_j^+}(\omega_1^{i_1}, \ldots, \omega_{j-1}^{i_{j-1}}, \omega_j^{2})=\chi_{\Omega_j^{0-}}(\omega_j^{1})\chi_{\Omega_j^{0+}}(\omega_j^{2})=$$
\begin{eqnarray}\label{apm2}
 \left\{\begin{array}{l l l l} 1,  &  \omega_j^{1} \in \Omega_j^{0-}&  \omega_j^{2} \in \Omega_j^{0+}, \\ 
0, & \mbox{otherwise,}
\end{array} \right. , \quad j=\overline{1,N}.
\end{eqnarray}
Further, 
$$\sum\limits_{i_j=1}^2\psi_j(\omega_1^{i_1}, \ldots, \omega_j^{i_j})\Delta S_j(\omega_1^{i_1}, \ldots, \omega_j^{i_j})=$$ 
$$\psi_j(\omega_1^{i_1}, \ldots, \omega_{j-1}^{i_{j-1}},\omega_j^{1})\Delta S_j(\omega_1^{i_1}, \ldots,\omega_{j-1}^{i_{j-1}}, \omega_j^{1})+$$
$$\psi_j(\omega_1^{i_1}, \ldots,\omega_{j-1}^{i_{j-1}}, \omega_j^{2})\Delta S_j(\omega_1^{i_1}, \ldots, \omega_{j-1}^{i_{j-1}},\omega_j^{2})=$$

$$ \chi_{\Omega_j^-}(\omega_1^{i_1}, \ldots, \omega_{j-1}^{i_{j-1}}, \omega_j^{1})\chi_{\Omega_j^+}(\omega_1^{i_1}, \ldots, \omega_{j-1}^{i_{j-1}}, \omega_j^{2})\times$$

$$\left[ - \frac{\Delta S_j^+(\omega_1^{i_1}, \ldots, \omega_{j-1}^{i_{j-1}}, \omega_j^{2})}{V_j(\omega_1^{i_1}, \ldots, \omega_{j-1}^{i_{j-1}}, \omega_j^1,  \omega_j^2)}\Delta S_j^-(\omega_1^{i_1}, \ldots,\omega_{j-1}^{i_{j-1}}, \omega_j^{1})+\right. $$
\begin{eqnarray}\label{1apm2}
\left. \frac{\Delta S_j^-(\omega_1^{i_1}, \ldots, \omega_{j-1}^{i_{j-1}}, \omega_j^{1})}{V_j(\omega_1^{i_1}, \ldots, \omega_{j-1}^{i_{j-1}}, \omega_j^1,  \omega_j^1)}\Delta S_j^+(\omega_1^{i_1}, \ldots,\omega_{j-1}^{i_{j-1}}, \omega_j^{2})\right]=0, \quad j=\overline{1,N}. 
\end{eqnarray}
Let us prove that the set of measures $\mu_{\{\omega_1^1,\omega_1^2\},\ldots ,\{\omega_N^1, \omega_N^2\}}(A)$ is a set of martingale measures.
Really, for $A,$ belonging to the $\sigma$-algebra ${\cal F}_{n-1}$ of the  filtration we have   $A=B\times \prod\limits_{i=n}^N \Omega_i^0,$ where $B$ belongs to $\sigma$-algebra ${\cal F}_{n-1}$ of the measurable space $\{\Omega_{n-1}, {\cal F}_{n-1}\}.$ Then,
$$\int\limits_{A}\Delta S_n(\omega_1, \ldots, \omega_n) d \mu_{\{\omega_1^1,\omega_1^2\},\ldots ,\{\omega_N^1, \omega_N^2\}}=$$
$$\sum\limits_{i_1=1}^2\ldots \sum\limits_{i_N=1}^2\prod\limits_{j=1}^N\psi_j(\omega_1^{i_1}, \ldots, \omega_j^{i_j}) \chi_{B}(\omega_1^{i_1}, \ldots, \omega_{n-1}^{i_{n-1}})\Delta S_n(\omega_1^{i_1}, \ldots, \omega_n^{i_n})=$$

$$\sum\limits_{i_1=1}^2\ldots \sum\limits_{i_n=1}^2\prod\limits_{j=1}^n\psi_j(\omega_1^{i_1}, \ldots, \omega_j^{i_j}) \chi_{B}(\omega_1^{i_1}, \ldots, \omega_{n-1}^{i_{n-1}})\Delta S_n(\omega_1^{i_1}, \ldots, \omega_n^{i_n})=$$

$$\sum\limits_{i_1=1}^2\ldots \sum\limits_{i_{n-1}=1}^2\prod\limits_{j=1}^{n-1}\psi_j(\omega_1^{i_1}, \ldots, \omega_j^{i_j}) \chi_{B}(\omega_1^{i_1}, \ldots, \omega_{n-1}^{i_{n-1}}) \times $$
\begin{eqnarray}\label{vita123}
\sum\limits_{i_n=1}^2 \psi_n(\omega_1^{i_1}, \ldots, \omega_n^{i_n}) 
\Delta S_n(\omega_1^{i_1}, \ldots, \omega_n^{i_n})=0, \quad A \in {\cal F}_{n-1}.\end{eqnarray}
The last means the needed statement.
Theorem \ref{10vitakolja121} is proved.

\end{proof}

Below, in Theorem \ref{100vitakolja121}, we present the consequence of Theorems  \ref{vitasja121},   \ref{10vitakolja121}.

\begin{te}\label{100vitakolja121}
Let the evolution of risky asset be given by the formula  (\ref{tin1vitasika1ja9}) and let Lemma \ref{witka0} conditions  be true.
Suppose that the random value 
 $\alpha_N(\{\omega\}_N^1;\{\omega\}_N^2),$ given on the 
probability space $\{\Omega_N^-\times \Omega_N^+, {\cal F}_N^- \times {\cal F}_N^+,  P_N^-\times  P_N^+ \}, $
satisfy the conditions
$$ P_N^-\times P_N^+(  (\{\omega_1^1, \ldots,\omega_N^1\};\{\omega_1^2, \ldots,\omega_N^2\}),  \alpha_N(\{\omega_1^1, \ldots,\omega_N^1\};\{\omega_1^2, \ldots,\omega_N^2\})>0)=$$
\begin{eqnarray}\label{1vitasikkolja7}
\prod\limits_{n=1}^NP_n^0(\Omega_n^{0-})\times P_n^0(\Omega_n^{0+});
\end{eqnarray}

 $$  \int\limits_{\Omega_n^{0-}\times \Omega_n^{0+}}\alpha_n(\{\omega_1^1, \ldots,\omega_{n-1}^1, \omega_n^1\};\{\omega_1^2, \ldots, \omega_{n-1}^2,\omega_n^2\})\times $$

$$\frac{\Delta S_n^+(\omega_1, \ldots,\omega_{n-1}, \omega_n^2) \Delta S_n^-(\omega_1, \ldots,\omega_{n-1}, \omega_n^1)}{V_n(\omega_1, \ldots,\omega_{n-1}, \omega_n^1,  \omega_n^2)}dP_n^0(\omega_n^1) dP_n^0(\omega_n^2)< \infty, $$
\begin{eqnarray}\label{2vitasikkolja7}
(\omega_1, \ldots,\omega_{n-1}) \in \Omega_{n-1};
\end{eqnarray}

\begin{eqnarray}\label{3vitasikkolja7} 
   \int\limits_{\prod\limits_{i=1}^N[\Omega_i^{0-}\times \Omega_i^{0+}]}\alpha_N(\{\omega_1^1, \ldots,\omega_N^1\};\{\omega_1^2, \ldots,\omega_N^2\})\prod\limits_{i=1}^N d P_i^0(\omega_i^1) d P_i^0(\omega_i^2)=1,  
\end{eqnarray}
where
\begin{eqnarray}\label{pupapup7} 
 \alpha_n(\{\omega_1^1, \ldots,\omega_{n-1}^1, \omega_n^1\};\{\omega_1^2, \ldots, \omega_{n-1}^2,\omega_n^2\})=
\end{eqnarray}
$$ \frac{\int\limits_{\prod\limits_{i=n+1}^N [\Omega_i^{0-}\times \Omega_i^{0+}]}\alpha_N(\{\omega_1^1, \ldots,\omega_N^1\};\{\omega_1^2, \ldots,\omega_N^2\})\prod\limits_{i=n+1}^N d P_i^0(\omega_i^1) d P_i^0(\omega_i^2)}{\int\limits_{\prod\limits_{i=n}^N [\Omega_i^{0-}\times \Omega_i^{0+}]}\alpha_N(\{\omega_1^1, \ldots,\omega_N^1\};\{\omega_1^2, \ldots,\omega_N^2\})\prod\limits_{i=n}^N d P_i^0(\omega_i^1) d P_i^0(\omega_i^2)  }, \ n=\overline{1,N}. $$

 If the set of strictly positive  random values  $\alpha_n(\{\omega\}_n^1; \{\omega\}_n^2), n=\overline{1,N},$ given by the formula (\ref{pupapup7}), satisfies the condition 
$$E^{\mu_0}|\Delta S_n(\omega_1, \ldots,\omega_{n-1}, \omega_n)|=$$
\begin{eqnarray}\label{pupapup8} 
\int\limits_{\Omega_N}\prod\limits_{i=1}^N\psi_i(\omega_1, \ldots,\omega_i) |\Delta S_n(\omega_1, \ldots,\omega_{n-1}, \omega_n)| \prod\limits_{i=1}^Nd P_i^0(\omega_i)<\infty, \quad n=\overline{1,N},
\end{eqnarray}
then, for the  martingale  measure $\mu_0(A)$  
the representation
$$\mu_0(A)=$$
\begin{eqnarray}\label{vitasikja122}
\int\limits_{\Omega_N\times \Omega_N}
\alpha_N(\{\omega_1^1, \ldots,\omega_N^1\};\{\omega_1^2, \ldots,\omega_N^2\}) \mu_{\{\omega_1^1,\omega_1^2\},\ldots ,\{\omega_N^1, \omega_N^2\}}(A) d P_N \times d P_N
\end{eqnarray}
is true.
\end{te} 
\begin{proof}
The random values 
$ \alpha_n(\{\omega_1^1, \ldots,\omega_{n-1}^1, \omega_n^1\};\{\omega_1^2, \ldots, \omega_{n-1}^2,\omega_n^2\}), \ n=\overline{1,N},$ satisfy the conditions
 (\ref{1vitasja7}) - (\ref{3vitasja7}), due to the conditions of Theorem \ref{100vitakolja121}.  
It is evident that
\begin{eqnarray}\label{vitasikja222}
\alpha_N(\{\omega_1^1, \ldots,\omega_N^1\};\{\omega_1^2, \ldots,\omega_N^2\})=\prod\limits_{n=1}^N \alpha_n(\{\omega_1^1, \ldots,\omega_n^1\};\{\omega_1^2, \ldots,\omega_n^2\}).
\end{eqnarray}
Due to Theorem \ref{10vitakolja121},  $\mu_0(A),$ given by the formula (\ref{vitasikja122}),  is a martingale measure being equivalent   to the measure $P_N.$

Let us indicate how to construct the random values $\alpha_N(\{\omega\}_N^1;\{\omega\}_N^2),$ since these random values determine the set of all martingale measures. Suppose that the random value $\alpha_i^k(\omega_i^1,\omega_i^2), \   k=\overline{1,K},$ is a bounded strictly positive random value, given on the measurable space $\{\Omega_i^{0-}\times \Omega_i^{0+}, {\cal F}_i^{0-}\times {\cal F}_i^{0+}\},$ $  i=\overline{1,N},$ and satisfying the conditions
\begin{eqnarray}\label{5vitasikaja500}
 \int\limits_{ \Omega_i^{0-}\times \Omega_i^{0+}}  \alpha_i^k(\omega_i^1,\omega_i^2) d P_i^0(\omega_i^1) d P_i^0(\omega_i^2)=1, \quad i=\overline{1,N}, \quad k=\overline{1,K}.  
\end{eqnarray}
Let us denote
\begin{eqnarray}\label{6vitasikaja500}
\alpha_N^k(\{\omega_1^1, \ldots,\omega_N^1\};\{\omega_1^2, \ldots,\omega_N^2\})=\prod\limits_{i=1}^N\alpha_i^k(\omega_i^1,\omega_i^2), \quad k=\overline{1,K},
\end{eqnarray}
where  $K$ runs natural numbers. If $\gamma_k, \ k=\overline{1,K}, $ are strictly positive real numbers such that $\sum\limits_{k=1}^K \gamma_k=1,$ then
\begin{eqnarray}\label{7vitasikaja500}
\alpha_N(\{\omega_1^1, \ldots,\omega_N^1\};\{\omega_1^2, \ldots,\omega_N^2\})=\sum\limits_{k=1}^K \gamma_k \alpha_N^k(\{\omega_1^1, \ldots,\omega_N^1\};\{\omega_1^2, \ldots,\omega_N^2\}) 
\end{eqnarray}
satisfy the conditions of Theorem \ref{100vitakolja121}.  The set of random values (\ref{7vitasikaja500}) is dense in the set of random values $\alpha_N(\{\omega_1^1, \ldots,\omega_N^1\};\{\omega_1^2, \ldots,\omega_N^2\}),$ satisfying the condition (\ref{1vitasikkolja7}) - (\ref{3vitasikkolja7}).
Theorem \ref{100vitakolja121} is proved.
\end{proof}
Another way to construct $\alpha_N(\{\omega_1^1, \ldots,\omega_N^1\};\{\omega_1^2, \ldots,\omega_N^2\})$ is to use the equalities (\ref{pupwitakolja1}).
The set of $\alpha_n(\{\omega_1^1, \ldots,\omega_{n-1}^1, \omega_{n}^1\};\{\omega_1^2, \ldots, \omega_{n-1}^2, \omega_{n}^{2}\})$ can construct  as follows: suppose that $\alpha_n^1(\{\omega_1^1, \ldots,\omega_{n-1}^1, \omega_{n}^1\};\{\omega_1^2, \ldots, \omega_{n-1}^2, \omega_{n}^{2}\})$
satisfies the inequalities
\begin{eqnarray}\label{pupwitakolja3}
0< h_n \leq \alpha_n^1(\{\omega_1^1, \ldots,\omega_{n-1}^1, \omega_{n}^1\};\{\omega_1^2, \ldots, \omega_{n-1}^2, \omega_{n}^{2}\})\leq H_n<\infty
\end{eqnarray}
for a certain  real positive numbers $ h_n, H_n.$  If to put
$$\alpha_n(\{\omega_1^1, \ldots,\omega_{n-1}^1, \omega_{n}^1\};\{\omega_1^2, \ldots, \omega_{n-1}^2, \omega_{n}^{2}\})=$$
\begin{eqnarray}\label{pupwitakolja4}
\frac{\alpha_n^1(\{\omega_1^1, \ldots,\omega_{n-1}^1, \omega_{n}^1\};\{\omega_1^2, \ldots, \omega_{n-1}^2, \omega_{n}^{2}\})}{ \int\limits_{ \Omega_n^{0-}\times \Omega_n^{0+}}  \alpha_n^1(\{\omega_1^1, \ldots,\omega_{n-1}^1, \omega_{n}^1\};\{\omega_1^2, \ldots, \omega_{n-1}^2, \omega_{n}^{2}\})d P_n^0(\omega_n^1) d P_n^0(\omega_n^2)},
\end{eqnarray}
then the set of random values $\alpha_n(\{\omega_1^1, \ldots,\omega_{n-1}^1, \omega_{n}^1\};\{\omega_1^2, \ldots, \omega_{n-1}^2, \omega_{n}^{2}\}),\ n=\overline{1,N},$ is bounded and    satisfy the conditions (\ref{1vitasja7}) - (\ref{3vitasja7}) under the conditions of Theorem \ref{10vitakolja121}.
We can put
 $$\alpha_N(\{\omega_1^1, \ldots,\omega_N^1\};\{\omega_1^2, \ldots,\omega_N^2\})=$$
\begin{eqnarray}\label{pupwitakolja5}
\prod\limits_{n=1}^N\alpha_n(\{\omega_1^1, \ldots,\omega_{n-1}^1, \omega_{n}^1\};\{\omega_1^2, \ldots, \omega_{n-1}^2, \omega_{n}^{2}\}).
\end{eqnarray}
It is evident that $\alpha_n(\{\omega_1^1, \ldots,\omega_{n-1}^1, \omega_{n}^1\};\{\omega_1^2, \ldots, \omega_{n-1}^2, \omega_{n}^{2}\}), \ n=\overline{1,N},$
must satisfy the conditions (\ref{pupapup8}). 

\section{Derivatives assessment.}

In the papers \cite{Eberlein},  \cite{Bellamy}, the range of non arbitrage prices are established. In the paper \cite{Eberlein}, for the Levy exponential model, the price of super-hedge for call option coincides with the price of the underlying asset under the assumption that the Levy process has unlimited variation, does not contain a Brownian component, with negative jumps of arbitrary magnitude.
The same result is true, obtained in the paper \cite{Bellamy}, if the process describing the evolution of the underlying asset is a diffusion process with the jumps described by Poisson jump process.
In these papers the evolution is described by continuous processes. Below, we consider the  discrete evolution of risky assets that is more realistic from the practical point of view.
Two types of risky asset evolutions are considered: 1)  the price of an asset can take any non negative  value; 2) the price of the risky asset may not fall below a given positive  value for finite time of evolution. For each of these types  of evolutions of risky assets, the bounds of non-arbitrage prices  for a wide class of contingent liabilities are found, among which are the payoff functions of standard call and put options.  

Below, on the probability space  $\{\Omega_N, {\cal F}_N, P_N\},$ where $\Omega_N=\prod\limits_{i=1}^N\Omega_i^0,$   ${\cal F}_N=\prod\limits_{i=1}^N {\cal F}_i^0, $ $P_N=\prod\limits_{i=1}^NP_i^0,$
$\Omega_i^0$ is a complete separable metric space, ${\cal F}_i^0$ is a Borel $\sigma$-algebra on $\Omega_i^0,$ $P_i^0 $ is a probability measure  on ${\cal F}_i^0, \  i=\overline{1,N}, $ we consider the evolution of risky asset given by the formula
$$S_n(\omega_1, \ldots,\omega_{n})=$$
\begin{eqnarray}\label{tinwau1}
S_0 \prod\limits_{i=1}^n(1+
a_i(\omega_1, \ldots,\omega_{i-1})(e^{\sigma_i(\omega_1, \ldots,\omega_{i-1})\varepsilon_i(\omega_i)} -1)),\quad n=\overline{1,N}, 
\end{eqnarray}
where $a_i(\omega_1, \ldots,\omega_{i-1}), \sigma_i(\omega_1, \ldots,\omega_{i-1})$ are ${\cal F}_{i-1}$-measurable   random values, satisfying the conditions $0 <a_i(\omega_1, \ldots,\omega_{i-1})\leq 1, \  \sigma_i (\omega_1, \ldots,\omega_{i-1})\geq \sigma_i>0, $ where $ \sigma_i, \ i=\overline{1,N},$ are real positive numbers. Further, we assume that the random value $\varepsilon_i(\omega_i)$ satisfy the conditions: there exists $\omega_i^1 \in \Omega_i^0$ such that  $\varepsilon_i(\omega_i^1)=0, \ i=\overline{1,N},$
and  for every real number $t>0,$ \ $P_i^0(\varepsilon_i(\omega_i)<- t)>0, $ \
$P_i^0(\varepsilon_i(\omega_i)>t)>0, $ \ $i=\overline{1,N}.$
 
For the evolution of risky asset (\ref{tinwau1}),  we have
$$ \Delta S_n(\omega_1, \ldots,\omega_{n-1},\omega_{n} )=$$
\begin{eqnarray}\label{1tinwau1}
  S_{n-1}(\omega_1, \ldots,\omega_{n-1})a_n(\omega_1, \ldots,\omega_{n-1}) 
(e^{\sigma_n(\omega_1, \ldots,\omega_{n-1})\varepsilon_n(\omega_n)} -1)=
\end{eqnarray}
$$  d_n(\omega_1, \ldots,\omega_{n-1},\omega_{n})(e^{\sigma_n\varepsilon_n(\omega_n)} -1),    $$
where 
$$   d_n(\omega_1, \ldots,\omega_{n-1},\omega_{n})=$$
\begin{eqnarray}\label{vito2tinwau1chka}
 S_{n-1}(\omega_1, \ldots,\omega_{n-1})a_n(\omega_1, \ldots,\omega_{n-1}) 
\frac{(e^{\sigma_n(\omega_1, \ldots,\omega_{n-1}) \varepsilon_n(\omega_n)} -1)}{(e^{\sigma_n\varepsilon_n(\omega_n)} -1)}.  
\end{eqnarray}
It is evident that $ d_n(\omega_1, \ldots,\omega_{n-1},\omega_{n})>0,$ therefore for $\Delta S_n(\omega_1, \ldots,\omega_{n-1},\omega_{n} )$ the representation 
(\ref{vitasja34}) is true with $ \eta_n(\omega_n)=(e^{\sigma_n\varepsilon_n(\omega_n)} -1).$ 
Therefore,
$$\frac{ \Delta S_n^+(\omega_1, \ldots,\omega_{n-1},\omega_{n}^2 )}{V_n(\omega_1, \ldots,\omega_{n-1}, \omega_{n}^1,\omega_{n}^2 )}=$$

\begin{eqnarray}\label{2tinwau1}
\frac{e^{\sigma_n(\omega_1, \ldots,\omega_{n-1})\varepsilon_n(\omega_n^2)} -1}{e^{\sigma_n(\omega_1, \ldots,\omega_{n-1})\varepsilon_n(\omega_n^2)}-e^{\sigma_n(\omega_1, \ldots,\omega_{n-1})\varepsilon_n(\omega_n^1)}},\quad \omega_n^2 \in \Omega_n^{0+}, \quad (\omega_1, \ldots,\omega_{n-1}) \in \Omega_{n-1},
\end{eqnarray}
$$\frac{ \Delta S_n^-(\omega_1, \ldots,\omega_{n-1},\omega_{n^1} )}{V_n(\omega_1, \ldots,\omega_{n-1}, \omega_{n}^1,\omega_{n}^2 )}=$$
\begin{eqnarray}\label{3tinwau1}
\frac{1 - e^{\sigma_n(\omega_1, \ldots,\omega_{n-1})\varepsilon_n(\omega_n^1)} }{e^{\sigma_n(\omega_1, \ldots,\omega_{n-1})\varepsilon_n(\omega_n^2)}-e^{\sigma_n(\omega_1, \ldots,\omega_{n-1})\varepsilon_n(\omega_n^1)}},\quad \omega_n^1 \in \Omega_n^{0-}, \quad (\omega_1, \ldots,\omega_{n-1}) \in \Omega_{n-1},
\end{eqnarray}
where we denoted
$$\Omega_n^{0-}=\{\omega_n \in \Omega_n^0, \varepsilon_n (\omega_n)\leq 0\}, \quad
\Omega_n^{0+}=\{\omega_n \in \Omega_n^0, \varepsilon_n (\omega_n)> 0\}, $$
\begin{eqnarray}\label{4tinwau1}
 \Omega_{n}^-=\Omega_n^{0-}\times \Omega_{n-1}, \quad    \Omega_{n}^+=\Omega_n^{0+}\times \Omega_{n-1}.
\end{eqnarray}

From the formulas (\ref{2tinwau1}),  (\ref{3tinwau1}) and Theorem \ref{witka2},
it follows that the set of martingale measures $M$ do not depend on the random values 
$a_i (\omega_1, \ldots,\omega_{i-1}), \ i=\overline{1,N}.$
If to put $a_i (\omega_1, \ldots,\omega_{i-1})=1,  \ i=\overline{1,N},$ in the formula  (\ref{tinwau1}), then  for the risky asset evolution we obtain the formula
\begin{eqnarray}\label{tochka4tinwau1vito}
S_n(\omega_1, \ldots,\omega_{n-1},\omega_{n})=S_0 \prod\limits_{i=1}^n e^{\sigma_i(\omega_1, \ldots,\omega_{i-1})\varepsilon_i(\omega_i)},\quad n=\overline{1,N}. 
\end{eqnarray}
The evolution of risky assets, given by the formula (\ref{tochka4tinwau1vito}), includes a wide class of evolutions of risky assets, used in economics. 
For example, under an appropriate choice of probability spaces $\{\Omega_i^0, {\cal F}_i^0, P_i^0\}$ and a choice of  sequence of independent  random values $\varepsilon_i(\omega_i)$ with the normal  distribution $N(0,1),$  we obtain  ARCH model (Autoregressive Conditional  Heteroskedastic Model) introduced by Engle in \cite{Engle}  and  GARCH model (Generalized Autoregressive Conditional  Heteroskedastic Model) introduced later by Bollerslev in \cite{Bollerslev}. In these models, the random variables $\sigma_i (\omega_1, \ldots,\omega_{i-1}), \ i=\overline{1,N},$ are called the  volatilities which satisfy the nonlinear set of  equations.

The very important case of  evolution of risky assets  (\ref{tinwau1}) is when $a_i(\omega_1, \ldots,\omega_{i-1})=a_i, \ i=\overline{1,N}, $ are constants, that is, 
\begin{eqnarray}\label{tinwau2}
S_n(\omega_1, \ldots,\omega_{n-1},\omega_{n})=S_0 \prod\limits_{i=1}^n(1+
a_i (e^{\sigma_i(\omega_1, \ldots,\omega_{i-1})\varepsilon_i(\omega_i)} -1)),\quad n=\overline{1,N}, 
\end{eqnarray}
where $0 \leq a_i \leq 1.$

If $0 < a_i < 1, i=\overline{1,N},$ then the evolution of risky asset, given by the formula  (\ref{tinwau2}), we call the evolution of relatively stable asset.

Further, we assume that   the evolution of risky asset given by the formulas (\ref{tinwau1}), (\ref{tochka4tinwau1vito}), (\ref{tinwau2}) satisfy the conditions
\begin{eqnarray}\label{tinvitochka4tinwau1}
\int\limits_{\Omega_N}S_n(\omega_1, \ldots,\omega_{n-1},\omega_{n}) d P_N< \infty, \quad n=\overline{1,N}.
\end{eqnarray}
From the conditions (\ref{tinvitochka4tinwau1}), it follows the inequalities
\begin{eqnarray}\label{vitunjawau2}
\int\limits_{\Omega_N}\Delta S_{n}^- (\omega_1, \ldots,\omega_{n})d P_N< \infty, \quad n=\overline{1,N}. 
\end{eqnarray}
Taking into account that  
$$ \Delta S_n^-(\omega_1, \ldots,\omega_{n-1},\omega_{n}^1 )=$$
\begin{eqnarray}\label{tinvitochka4tinwau2}
S_{n-1}(\omega_1, \ldots,\omega_{n-1})a_n(\omega_1, \ldots,\omega_{n-1}) 
(1- e^{\sigma_n(\omega_1, \ldots,\omega_{n-1})\varepsilon_n(\omega_n^1)}), \quad \omega_n^1 \in \Omega_n^{0-},
\end{eqnarray}
$$ \Delta S_n^+(\omega_1, \ldots,\omega_{n-1},\omega_{n}^2 )=$$
\begin{eqnarray}\label{tinvitochka4tinwau2}
S_{n-1}(\omega_1, \ldots,\omega_{n-1})a_n(\omega_1, \ldots,\omega_{n-1}) 
(e^{\sigma_n(\omega_1, \ldots,\omega_{n-1})\varepsilon_n(\omega_n^2)} -1), \quad \omega_n^2 \in \Omega_n^{0+},
\end{eqnarray}
we have
\begin{eqnarray}\label{tinvitochka4tinwau4}
\frac{1}{ \Delta S_n^-(\omega_1, \ldots,\omega_{n-1},\omega_{n}^1 )}\leq \frac{1}{\prod\limits_{i=1}^{n-1}(1-a_i^1) a_n^0(1- e^{\sigma_n \varepsilon_n(\omega_n^1)})}<\infty, \quad \varepsilon_n(\omega_n^1)<0,
\end{eqnarray}
\begin{eqnarray}\label{tinvitochka4tinwau5}
\frac{1}{ \Delta S_n^+(\omega_1, \ldots,\omega_{n-1},\omega_{n}^2 )}\leq \frac{1}{\prod\limits_{i=1}^{n-1}(1-a_i^1) a_n^0(e^{\sigma_n \varepsilon_n(\omega_n^2)}-1)}<\infty, \quad \varepsilon_n(\omega_n^2)>0,
\end{eqnarray}
under the conditions that
\begin{eqnarray}\label{tinvitochka4tinwau3}
0< a_n^0 \leq a_n(\omega_1, \ldots,\omega_{n-1})\leq a_n^1<1, \quad n=\overline{1,N}.
\end{eqnarray}
\begin{te}\label{vitochkakiss1}
On the probability space $\{\Omega_N, {\cal F}_N, P_N\},$  let the evolution of risky asset be given by one of the formula (\ref{tinwau1}), (\ref{tochka4tinwau1vito}),  (\ref{tinwau2}) that satisfies  the conditions  (\ref{tinvitochka4tinwau1}).

 If the inequalities $0<a_n^0 \leq a_n(\omega_1,\ldots, \omega_{n-1})\leq a_n^1<1, \ 0<a_i<1, \  i=\overline{1,N},$  are true,  then the set of martingale measures $M$ is the same for every evolution of risky assets, given by the formulas  (\ref{tinwau1}), (\ref{tinwau2}).
 For every non-negative super-martingale relative to the set of martingale measures $M$ the optional decomposition is valid.
 Every measure of $M$ is an integral over the spot measures. The fair price $f_0$ of super-hedge for the nonnegative payoff function $f(x)$ is given by the formula
\begin{eqnarray}\label{vitochkakiss2}
f_0=\sup\limits_{P \in M}E^P f(S_N)=\sup\limits_{\omega_i^1 \in \Omega_i^{0-}, \omega_i^2 \in \Omega_i^{0+}, i=\overline{1,N}}\int\limits_{\Omega_N}f(S_N) d \mu_{\{\omega_1^1,\omega_1^2\},\ldots ,\{\omega_N^1, \omega_N^2\}}.
\end{eqnarray}
The set of martingale measures $M_1$ for the evolution of risky asset, given by the formula (\ref{tochka4tinwau1vito}), is contained in the set $M.$
\end{te}
\begin{proof}
From the equalities (\ref{2tinwau1}) - (\ref{3tinwau1}) and the inequalities (\ref{tinvitochka4tinwau1}), it follows that the set $ M$ 
is a nonempty one and every martingale  measure constructed by the set of random values 
$\alpha_n(\omega_1^1,\ldots,\omega_n^1;\omega_1^2,\ldots,\omega_n^2), \ n=\overline{1,N},$ belongs to the set $M,$ if the  inequalities (\ref{vitasja22}) are true.
To prove that the set of martingale measures, defined by the evolutions (\ref{tinwau1}),   (\ref{tinwau2}),  coincide it is necessary to prove the inequalities
\begin{eqnarray}\label{kisskiss3}
0 <  A_n^1\leq \frac {S_n^1(\omega_1,\ldots,\omega_n)}{S_n^2(\omega_1,\ldots,\omega_n)}\leq B_n^1<\infty, \quad n=\overline{1,N},
\end{eqnarray}
where we denoted by $S_n^1 (\omega_1,\ldots,\omega_n) $ the evolution, given by the formula (\ref{tinwau1}), and by $S_n^2 (\omega_1,\ldots,\omega_n) $  the evolution, given by the formula  (\ref{tinwau2}). Under the conditions of Theorem \ref{vitochkakiss1}, we have
$$ \frac {S_n^1(\omega_1,\ldots,\omega_n)}{S_n^2(\omega_1,\ldots,\omega_n)} =$$
\begin{eqnarray}\label{kisskiss4}
\frac{S_0 \prod\limits_{i=1}^n(1+
a_i(\omega_1, \ldots,\omega_{i-1})(e^{\sigma_i(\omega_1, \ldots,\omega_{i-1})\varepsilon_i(\omega_i)} -1))}{S_0 \prod\limits_{i=1}^n(1+
a_i (e^{\sigma_i(\omega_1, \ldots,\omega_{i-1})\varepsilon_i(\omega_i)} -1))}, \quad n=\overline{1,N}. 
\end{eqnarray}
Since 
$$ \frac{1+
a_i(\omega_1, \ldots,\omega_{i-1})(e^{\sigma_i(\omega_1, \ldots,\omega_{i-1})\varepsilon_i(\omega_i)} -1)}{1+
a_i(e^{\sigma_i(\omega_1, \ldots,\omega_{i-1})\varepsilon_i(\omega_i)} -1)}=$$
\begin{eqnarray}\label{kisskiss5}
 \frac{1-a_i(\omega_1, \ldots,\omega_{i-1})+
a_i(\omega_1, \ldots,\omega_{i-1})e^{\sigma_i(\omega_1, \ldots,\omega_{i-1})\varepsilon_i(\omega_i)}}{1-a_i+
a_i e^{\sigma_i(\omega_1, \ldots,\omega_{i-1})\varepsilon_i(\omega_i)}}=D_i, \quad i=\overline{1,N},
\end{eqnarray}
we have
$$\frac{1-a_i^1+a_i^0 e^{\sigma_i(\omega_1, \ldots,\omega_{i-1})\varepsilon_i(\omega_i)}}{1-a_i+
a_i e^{\sigma_i(\omega_1, \ldots,\omega_{i-1})\varepsilon_i(\omega_i)}}\leq  D_i    \leq$$
\begin{eqnarray}\label{kisskiss6}
 \frac{1-a_i^0+a_i^1 e^{\sigma_i(\omega_1, \ldots,\omega_{i-1})\varepsilon_i(\omega_i)}}{1-a_i+
a_i e^{\sigma_i(\omega_1, \ldots,\omega_{i-1})\varepsilon_i(\omega_i)}}, \quad i=\overline{1,N}.  
\end{eqnarray}
Let us denote
$$ A_i=\inf\limits_{(\omega_1, \ldots,\omega_{i}) \in \Omega_i}\frac{1-a_i^1+a_i^0 e^{\sigma_i(\omega_1, \ldots,\omega_{i-1})\varepsilon_i(\omega_i)}}{1-a_i+
a_i e^{\sigma_i(\omega_1, \ldots,\omega_{i-1})\varepsilon_i(\omega_i)}}, \quad i=\overline{1,N},  $$
\begin{eqnarray}\label{kisskiss7}
 B_i=\sup\limits_{(\omega_1, \ldots,\omega_{i}) \in \Omega_i}\frac{1-a_i^0+a_i^1 e^{\sigma_i(\omega_1, \ldots,\omega_{i-1})\varepsilon_i(\omega_i)}}{1-a_i+
a_i e^{\sigma_i(\omega_1, \ldots,\omega_{i-1})\varepsilon_i(\omega_i)}}, \quad i=\overline{1,N}.  
\end{eqnarray}
It is evident that $0<A_i, B_i<\infty, \ i=\overline{1,N},$ and 
\begin{eqnarray}\label{kisskiss8}
A_i \leq D_i \leq B_i, \quad i=\overline{1,N},
\end{eqnarray}
therefore
\begin{eqnarray}\label{kisskiss9}
A_n^1=\prod\limits_{i=1}^n A_i \leq \frac {S_n^1(\omega_1,\ldots,\omega_n)}{S_n^2(\omega_1,\ldots,\omega_n)}\leq \prod\limits_{i=1}^n B_i=B_n^1, \quad  n=\overline{1,N}.  
\end{eqnarray}
So,
\begin{eqnarray}\label{kisskiss10} 
A_N^2  \leq \frac {S_n^1(\omega_1,\ldots,\omega_n)}{S_n^2(\omega_1,\ldots,\omega_n)}\leq B_N^2, \quad  n=\overline{1,N},
\end{eqnarray}
where we put 
$A_N^2=\min\limits_{1\leq n \leq N}A_n^1, \  B_N^2=\max\limits_{1\leq n \leq N}B_n^1.$
Since 
$$| \Delta S_n^1(\omega_1, \ldots,\omega_{n-1},\omega_{n} )|=$$
\begin{eqnarray}\label{kisskiss1}
S_{n-1}^1(\omega_1, \ldots,\omega_{n-1})a_n(\omega_1, \ldots,\omega_{n-1}) 
|(e^{\sigma_n(\omega_1, \ldots,\omega_{n-1})\varepsilon_n(\omega_n)} -1)|, 
\end{eqnarray}

$$| \Delta S_n^2(\omega_1, \ldots,\omega_{n-1},\omega_{n} )|=$$
\begin{eqnarray}\label{kisskiss2}
S_{n-1}^2(\omega_1, \ldots,\omega_{n-1})a_n 
|(e^{\sigma_n(\omega_1, \ldots,\omega_{n-1})\varepsilon_n(\omega_n)} -1)|, 
\end{eqnarray}
we have
$$ \frac{| \Delta S_n^1(\omega_1, \ldots,\omega_{n-1},\omega_{n} )|}{| \Delta S_n^2(\omega_1, \ldots,\omega_{n-1},\omega_{n} )|} = $$
\begin{eqnarray}\label{kisskiss11}
 \frac{S_{n-1}^1(\omega_1, \ldots,\omega_{n-1}) a_n(\omega_1, \ldots,\omega_{n-1})}{ S_{n-1}^2(\omega_1, \ldots,\omega_{n-1})a_n}.  
\end{eqnarray}
Taking into account the obtained inequalities, we have the inequalities
\begin{eqnarray}\label{kisskiss12}
A_N^2\frac {\min\limits_{1\leq n\leq N}a_n^0}{\max\limits_{1\leq n\leq N} a_n} \leq \frac{| \Delta S_n^1(\omega_1, \ldots,\omega_{n-1},\omega_{n} )|}{| \Delta S_n^2(\omega_1, \ldots,\omega_{n-1},\omega_{n} )|} \leq B_N^2 \frac {\max\limits_{1\leq n\leq N}a_n^1}{\min\limits_{1\leq n\leq N} a_n},\quad n=\overline{1,N}.
\end{eqnarray}
The  inequalities (\ref{kisskiss12}) proves that the set of martingale measures for the evolutions of risky assets given by the formulas   (\ref{tinwau1}),   (\ref{tinwau2}) are the same, since the inequalities (\ref{vitasja22}) for the evolutions of risky assets, given by formulas (\ref{tinwau1}),   (\ref{tinwau2}), are fulfilled simultaneously.

For the evolution of risky assets (\ref{tinwau2}), satisfying the conditions (\ref{tinvitochka4tinwau3}), the inequalities (\ref{tinvitochka4tinwau4}), (\ref{tinvitochka4tinwau5}) are true. From this, it follows  that the conditions of Theorem \ref{Tinnna1} are valid. This proves the optional decomposition for every nonnegative super-martingale relative to the family of martingale measures $M.$ From \cite{GoncharNick}, it follows  the formula for the fair price $f_0$ of super-hedge
\begin{eqnarray}\label{vitochkakiss3}
f_0=\sup\limits_{P \in M}E^P f(S_N).
\end{eqnarray}
Further, the conditions of Theorem \ref{100vitakolja121} is also true. Therefore, the formula 
\begin{eqnarray}\label{vitochkakiss4}
\sup\limits_{P \in M}E^P f(S_N)=\sup\limits_{\omega_i^1 \in \Omega_i^{0-}, \omega_i^2 \in \Omega_i^{0+}, i=\overline{1,N}}\int\limits_{\Omega_N}f(S_N) d \mu_{\{\omega_1^1,\omega_1^2\},\ldots ,\{\omega_N^1, \omega_N^2\}}
\end{eqnarray}
is valid. 

To complete  the proof of Theorem \ref{vitochkakiss1}, it needs to show that the set
$M_1\subseteq M.$
Let us denote $S_n^3(\omega_1,\ldots ,\omega_n)$ the evolution of risky asset, given by the formula  (\ref{tochka4tinwau1vito}). Then, as above
\begin{eqnarray}\label{kisskiss13}
\frac{S_n^3(\omega_1,\ldots ,\omega_n)}{S_n^2(\omega_1,\ldots ,\omega_n)}\leq \prod\limits_{i=1}^n\frac{1}{a_i}=C_n,\quad n=\overline{1,N}.
\end{eqnarray}
Therefore,
$$ \frac{| \Delta S_n^3(\omega_1, \ldots,\omega_{n-1},\omega_{n} )|}{| \Delta S_n^2(\omega_1, \ldots,\omega_{n-1},\omega_{n} )|} = $$
\begin{eqnarray}\label{kisskiss14}
 \frac{S_{n-1}^3(\omega_1, \ldots,\omega_{n-1}) }{S_{n-1}^2(\omega_1, \ldots,\omega_{n-1}) a_n}\leq \frac{\max\limits_{1\leq n \leq N}C_n}{\min\limits_{1\leq n \leq N} a_n}, \quad n=\overline{1,N}.
\end{eqnarray}
The inequality (\ref{kisskiss14}) proves the needed statement.
Theorem \ref{vitochkakiss1} is proved.
\end{proof}

 \begin{te}\label{tinvika1}
On the probability space $\{\Omega_N, {\cal F}_N, P_N\},$  let the evolution of risky asset be given by the formula (\ref{tinwau1}). Suppose that $0 \leq a_i(\omega_1, \ldots,\omega_{i-1})\leq  1,$ $ \sigma_i(\omega_1, \ldots,\omega_{i-1})>\sigma_i>0, \  i=\overline{1,N},$
 and $a_n=1$ for a certain $1 \leq n \leq N.$
If the  nonnegative   payoff function $f(x),   \ x \in [0,\infty),$ satisfies the conditions:\\
1) $f(0)=0, \ f(x) \leq a x, \  \lim\limits_{x \to \infty}\frac{f(x)}{x}=a,\ a >0,$
then 
\begin{eqnarray}\label{tinvika3}
\sup\limits_{P \in M}E^P f(S_N)=a S_0.
\end{eqnarray}
If, in addition,  the nonnegative payoff function $f(x)$ is a convex down one, then
\begin{eqnarray}\label{tinvika4}
\inf\limits_{P \in M}E^P f(S_N)=f(S_0), 
\end{eqnarray}
where $M$ is a  set of equivalent martingale measures for the  evolution of  risky asset,  given by the formula (\ref{tinwau1}). The interval of non-arbitrage prices of contingent liability $f(S_N)$ lies in the set $ [f(S_0), a S_0].$
\end{te}
\begin{proof} 
Since the conditions of  Theorem \ref{vitochkakiss1} are satisfied, then  the formula
\begin{eqnarray}\label{tinmmarstina2}
\sup\limits_{Q \in M}\int\limits_{\Omega_N}f(S_N) dQ=\sup\limits_{\omega_i^1 \in \Omega_i^{0-}, \omega_i^2 \in \Omega_i^{0+}, i=\overline{1,N}}\int\limits_{\Omega_N}f(S_N) d \mu_{\{\omega_1^1,\omega_1^2\},\ldots ,\{\omega_N^1, \omega_N^2\}}
\end{eqnarray}
is true, where for the spot measure $\mu_{\{\omega_1^1,\omega_1^2\},\ldots ,\{\omega_N^1, \omega_N^2\}}(A)$ the representation 

$$\mu_{\{\omega_1^1,\omega_1^2\},\ldots ,\{\omega_N^1, \omega_N^2\}}(A)=$$
\begin{eqnarray}\label{tinvitasja115}
\sum\limits_{i_1=1}^2\ldots \sum\limits_{i_N=1}^2\prod\limits_{j=1}^N\psi_j(\omega_1^{i_1}, \ldots, \omega_j^{i_j})\chi_{A}(\omega_1^{i_1}, \ldots, \omega_N^{i_N}), \quad A \in {\cal F}_N,
\end{eqnarray}
is valid, and 
$$\sup\limits_{\omega_i^1 \in \Omega_i^{0-}, \omega_i^2 \in \Omega_i^{0+}, i=\overline{1,N}}\int\limits_{\Omega_N}f(S_N) d \mu_{\{\omega_1^1,\omega_1^2\},\ldots ,\{\omega_N^1, \omega_N^2\}}=$$

$$\sup\limits_{\omega_i^1 \in \Omega_i^{0-}, \omega_i^2 \in \Omega_i^{0+}, i=\overline{1,N}}\sum\limits_{i_1=1, \ldots,i_N=1}^2 \prod\limits_{j=1}^N\psi_j(\omega_1^{i_1}, \ldots, \omega_j^{i_j}) \times $$ 
\begin{eqnarray}\label{tinvika5}
  f\left(S_0\prod\limits_{s=1}^N\left(1+a_s(\omega_1^{i_{1}}, \ldots, \omega_{s-1}^{i_{s-1}}) 
\left(e^{\sigma_s(\omega_1^{i_{1}}, \ldots, \omega_{s-1}^{i_{s-1}}) \varepsilon_s (\omega_{s}^{i_{s}})}-1\right) \right)\right),
\end{eqnarray}
where we denoted 
 $ \Omega_s^{0-}=\{\omega_s \in \Omega_s^0, \  \varepsilon_s(\omega_{s})\leq 0\},$  
 $ \Omega_s^{0+}=\{\omega_s \in \Omega_s^0, \  \varepsilon_s(\omega_{s})> 0\}.$  
 From the inequality, $f(S_N) \leq a S_N,$ we have 
\begin{eqnarray}\label{tinvika6}
\sup\limits_{Q \in M}\int\limits_{\Omega}f(S_N) dQ \leq a S_0.
\end{eqnarray}
To prove the inverse inequality, we use the inequality
$$\sup\limits_{Q \in M}\int\limits_{\Omega}f(S_N) dQ \geq$$
$$\sum\limits_{i_1=1, \ldots,i_N=1}^2 \prod\limits_{j=1}^N\psi_j(\omega_1^{i_1}, \ldots, \omega_j^{i_j}) \times $$ 
\begin{eqnarray}\label{pupochka5}
f\left(S_0\prod\limits_{s=1}^N\left(1+a_s(\omega_1^{i_{1}}, \ldots, \omega_{s-1}^{i_{s-1}}) 
\left(e^{\sigma_s(\omega_1^{i_{1}}, \ldots, \omega_{s-1}^{i_{s-1}}) \varepsilon_s(\omega_{s}^{i_{s}})}-1\right) \right)\right).
\end{eqnarray}
In the right hand  side of the  last inequality, let us put $\varepsilon_s (\omega_{s}^{1})=0, \ s \neq n.$ Such elementary events $\omega_s^1$ exist, due to the conditions relative to  the random values $ \varepsilon_s(\omega_{s}),\ s=\overline{1,N}.$ 
We obtain 
$$\sum\limits_{i_1=1, \ldots,i_N=1}^2 \prod\limits_{j=1}^N\psi_j(\omega_1^{i_1}, \ldots, \omega_j^{i_j}) \times $$ 
$$f\left(S_0\prod\limits_{s=1}^N\left(1+a_s(\omega_1^{i_{1}}, \ldots, \omega_{s-1}^{i_{s-1}}) 
\left(e^{\sigma_s(\omega_1^{i_{1}}, \ldots, \omega_{s-1}^{i_{s-1}}) \varepsilon_s(\omega_{s}^{i_{s}})}-1\right) \right)\right)=$$
\begin{eqnarray}\label{tin5vika7}
\sum\limits_{i_n=1}^2 \psi_n(\omega_1^{1}, \ldots, \omega_{n-1}^{1},  \omega_n^{i_n})
f\left(S_0
e^{\sigma_n(\omega_1^{1}, \ldots, \omega_{n-1}^{1}) \varepsilon_n(\omega_{n}^{i_n})}  \right).
\end{eqnarray}
 Therefore, 
$$\sup\limits_{Q \in M}\int\limits_{\Omega}f(S_N) dQ \geq$$
\begin{eqnarray}\label{tinpupa5vika7}
\sup\limits_{\omega_n^1 \in \Omega_n^{0-}, \omega_n^2 \in \Omega_n^{0+}}\sum\limits_{i_n=1}^2 \psi_n(\omega_1^{1}, \ldots, \omega_{n-1}^{1},  \omega_n^{i_n})
f\left(S_0
e^{\sigma_n(\omega_1^{1}, \ldots, \omega_{n-1}^{1}) \varepsilon_n(\omega_{n}^{i_n})}  \right).
\end{eqnarray}
Further,
$$\sup\limits_{\omega_n^1 \in \Omega_n^{0-}, \omega_n^2 \in \Omega_n^{0+}} \sum\limits_{i_n=1}^2 \psi_n(\omega_1^{1}, \ldots, \omega_n^{i_n})\times$$
$$ f\left(S_0
e^{\sigma_n(\omega_1^{1}, \ldots, \omega_{n-1}^{1}) \varepsilon_n(\omega_{n}^{i_{n}})}  \right)=$$

$$\sup\limits_{\omega_n^1 \in \Omega_n^{0-}, \omega_n^2 \in \Omega_n^{0+}}\left[\frac{\Delta S_n^+(\omega_1^{1}, \ldots, \omega_{n-1}^{1}, \omega_n^2)}{V_n(\omega_1^{1}, \ldots, \omega_{n-1}^{1}, \omega_n^1, \omega_n^2)} f\left(S_0
e^{\sigma_n(\omega_1^{1}, \ldots, \omega_{n-1}^{1}) \varepsilon_n(\omega_{n}^{1})}\right)+\right.$$
$$\left.\frac{\Delta S_n^-(\omega_1^{1}, \ldots, \omega_{n-1}^{1}, \omega_n^1)}{V_n(\omega_1^{1}, \ldots, \omega_{n-1}^{1}, \omega_n^1, \omega_n^2)} f\left(S_0
e^{\sigma_n(\omega_1^{1}, \ldots, \omega_{n-1}^{1}) \varepsilon_n (\omega_{n}^{2})}\right)\right]\geq$$

$$\lim\limits_{\varepsilon_n(\omega_{n}^{2}) \to \infty} \lim\limits_{\varepsilon_n(\omega_{n}^{1}) \to -\infty}\left[\frac{e^{\sigma_n(\omega_1^{1}, \ldots, \omega_{n-1}^{1}) \varepsilon_n(\omega_{n}^{2})}-1}{e^{\sigma_n(\omega_1^{1}, \ldots, \omega_{n-1}^{1}) \varepsilon_n(\omega_{n}^{2})}- e^{\sigma_n(\omega_1^{1}, \ldots, \omega_{n-1}^{1}) \varepsilon_n(\omega_{n}^{1})}}\times \right. $$ 
$$ f\left(S_0e^{\sigma_n(\omega_1^{1}, \ldots, \omega_{n-1}^{1}) \varepsilon_n(\omega_{n}^{1})}\right)+$$

$$\left.\frac{1-e^{\sigma_n(\omega_1^{1}, \ldots, \omega_{n-1}^{1}) \varepsilon_n(\omega_{n}^{1})}}{e^{\sigma_n (\omega_1^{1}, \ldots, \omega_{n-1}^{1}) \varepsilon_n(\omega_{n}^{2})}- e^{\sigma_n(\omega_1^{1}, \ldots, \omega_{n-1}^{1}) \varepsilon_n(\omega_{n}^{1})}} f\left(S_0
e^{\sigma_n(\omega_1^{1}, \ldots, \omega_{n-1}^{1}) \varepsilon_n(\omega_{n}^{2})}\right)\right]=$$

\begin{eqnarray}\label{tin15vika9}
\lim\limits_{\varepsilon_n (\omega_{n}^{2}) \to \infty}\frac{1}{e^{\sigma_n(\omega_1^{1}, \ldots, \omega_{n-1}^{1}) \varepsilon_n(\omega_{n}^{2})}} f\left(S_0 e^{\sigma_n(\omega_1^{1}, \ldots, \omega_{n-1}^{1}) \varepsilon_n(\omega_{n}^{2})}
\right)=a S_0.
\end{eqnarray}
Substituting the inequality (\ref{tin15vika9}) into the inequality (\ref{tin5vika7}), we obtain the needed inequality.

Let us prove the equality (\ref{tinvika4}).
Using the Jensen inequality, we obtain
\begin{eqnarray}\label{tinvika8}
\inf\limits_{P \in M}E^P f(S_N) \geq f(E^PS_N)=f(S_0).
\end{eqnarray}
Let us prove the inverse inequality. It is  evident that
$$\sum\limits_{i_1=1, \ldots,i_N=1}^2 \prod\limits_{j=1}^N\psi_j(\omega_1^{i_1}, \ldots, \omega_j^{i_j}) \times $$ 

$$f\left(S_0\prod\limits_{s=1}^N\left(1+a_s(\omega_1^{i_{1}}, \ldots, \omega_{s-1}^{i_{s-1}}) 
\left(e^{\sigma_s(\omega_1^{i_{1}}, \ldots, \omega_{s-1}^{i_{s-1}}) \varepsilon_s (\omega_{s}^{i_{s}})}-1\right) \right)\right)\geq	$$

\begin{eqnarray}\label{tinvika9}
 \inf\limits_{P \in M}E^P f(S_N).
\end{eqnarray}
Putting in this inequality $\varepsilon_i(\omega_{i}^{1})=0, \ i=\overline{1,N},$ we obtain the needed.  The last statement about the  interval of non-arbitrage prices  follows from \cite{WalterSchacher} and  \cite{DMW90}. Theorem \ref{tinvika1} is proved.
\end{proof}

\begin{te}\label{tinvika10}
On the probability space $\{\Omega_N, {\cal F}_N, P_N\},$  let the evolution of risky asset be given by the formula (\ref{tinwau1}). Suppose that $0 \leq a_i(\omega_1, \ldots,\omega_{i-1})\leq  1,$ $ \sigma_i(\omega_1, \ldots,\omega_{i-1})>\sigma_i>0, \  i=\overline{1,N},$
 and $a_n=1$ for a certain $1 \leq n \leq N.$
If the nonnegative  payoff function $f(x),   \ x \in [0,\infty),$ satisfies the conditions:\\
1) $f(0)=K, \ f(x) \leq K, $
then 
\begin{eqnarray}\label{tinvika12}
\sup\limits_{P \in M}E^P f(S_N)=K.
\end{eqnarray}
If, in addition, the nonnegative payoff function $f(x)$ is a convex down one, then
\begin{eqnarray}\label{tinvika13}
\inf\limits_{P \in M}E^P f(S_N)=f(S_0), 
\end{eqnarray}
where $M$ is a set of equivalent maqtingale measures for the  evolution of  risky asset,  given by the formula (\ref{tinwau1}).  The interval of non-arbitrage prices of contingent liability $f(S_N)$ coincides with the set $[f(S_0), K].$
\end{te}
\begin{proof}
Due to Theorem \ref{vitochkakiss1}, the equality
\begin{eqnarray}\label{12tinmmarstina2}
\sup\limits_{Q \in M}\int\limits_{\Omega_N}f(S_N) dQ=\sup\limits_{\omega_i^1 \in \Omega_i^{0-}, \omega_i^2 \in \Omega_i^{0+}, i=\overline{1,N}}\int\limits_{\Omega_N}f(S_N) d \mu_{\{\omega_1^1,\omega_1^2\},\ldots ,\{\omega_N^1, \omega_N^2\}}
\end{eqnarray}
is valid, where for the spot measure $\mu_{\{\omega_1^1,\omega_1^2\},\ldots ,\{\omega_N^1, \omega_N^2\}}(A)$ the representation 
$$\mu_{\{\omega_1^1,\omega_1^2\},\ldots ,\{\omega_N^1, \omega_N^2\}}(A)=$$
\begin{eqnarray}\label{tinvitasja115}
\sum\limits_{i_1=1}^2\ldots \sum\limits_{i_N=1}^2\prod\limits_{j=1}^N\psi_j(\omega_1^{i_1}, \ldots, \omega_j^{i_j})\chi_{A}(\omega_1^{i_1}, \ldots, \omega_N^{i_N}), \quad A \in {\cal F}_N,
\end{eqnarray}
is true, and
$$\sup\limits_{\omega_i^1 \in \Omega_i^{0-}, \omega_i^2 \in \Omega_i^{0+}, i=\overline{1,N}}\int\limits_{\Omega_N}f(S_N) d \mu_{\{\omega_1^1,\omega_1^2\},\ldots ,\{\omega_N^1, \omega_N^2\}}=$$

$$\sup\limits_{\omega_i^1 \in \Omega_i^{0-}, \omega_i^2 \in \Omega_i^{0+}, i=\overline{1,N}}\sum\limits_{i_1=1, \ldots,i_N=1}^2 \prod\limits_{j=1}^N\psi_j(\omega_1^{i_1}, \ldots, \omega_j^{i_j}) \times $$ 
\begin{eqnarray}\label{tinvika5}
  f\left(S_0\prod\limits_{s=1}^N\left(1+a_s(\omega_1^{i_{1}}, \ldots, \omega_{s-1}^{i_{s-1}}) 
\left(e^{\sigma_s(\omega_1^{i_{1}}, \ldots, \omega_{s-1}^{i_{s-1}}) \varepsilon_s(\omega_{s}^{i_{s}})}-1\right) \right)\right).
\end{eqnarray}
It is evident that
\begin{eqnarray}\label{tinvika14}
\sup\limits_{P \in M}E^P f(S_N)\leq K.
\end{eqnarray}
Further,
$$\sup\limits_{Q \in M}\int\limits_{\Omega}f(S_N) dQ \geq$$
$$\sum\limits_{i_1=1, \ldots,i_N=1}^2 \prod\limits_{j=1}^N\psi_j(\omega_1^{i_1}, \ldots, \omega_j^{i_j}) \times $$ 
\begin{eqnarray}\label{putinvikapu14}
  f\left(S_0\prod\limits_{s=1}^N\left(1+a_s(\omega_1^{i_{1}}, \ldots, \omega_{s-1}^{i_{s-1}}) 
\left(e^{\sigma_s(\omega_1^{i_{1}}, \ldots, \omega_{s-1}^{i_{s-1}}) \varepsilon_s(\omega_{s}^{i_{s}})}-1\right) \right)\right).
\end{eqnarray}
In the right hand side of the last inequality, let us put $\varepsilon_s (\omega_{s}^{1})=0, \ s \neq n.$
We obtain 
$$\sum\limits_{i_1=1, \ldots,i_N=1}^2 \prod\limits_{j=1}^N\psi_j(\omega_1^{i_1}, \ldots, \omega_j^{i_j}) \times $$ 
$$f\left(S_0\prod\limits_{s=1}^N\left(1+a_s(\omega_1^{i_{1}}, \ldots, \omega_{s-1}^{i_{s-1}}) 
\left(e^{\sigma_s(\omega_1^{i_{1}}, \ldots, \omega_{s-1}^{i_{s-1}}) \varepsilon_s(\omega_{s}^{i_{s}})}-1\right) \right)\right)=$$
\begin{eqnarray}\label{1tin5vika7}
\sum\limits_{i_n=1}^2 \psi_n(\omega_1^{1}, \ldots, \omega_{n-1}^{1}, \omega_n^{i_n})f\left(S_0
e^{\sigma_n (\omega_1^{1}, \ldots, \omega_{n-1}^{1}) \varepsilon_n (\omega_{n}^{i_{n}})}\right).
\end{eqnarray}
From the last equality, we obtain
$$\sup\limits_{Q \in M}\int\limits_{\Omega}f(S_N) dQ \geq$$
\begin{eqnarray}\label{1tinpupas5vika7}
\sup\limits_{\omega_n^1 \in \Omega_n^{0-}, \omega_n^2 \in \Omega_n^{0+}}\sum\limits_{i_n=1}^2 \psi_n(\omega_1^{1}, \ldots, \omega_{n-1}^{1}, \omega_n^{i_n})f\left(S_0
e^{\sigma_n(\omega_1^{1}, \ldots, \omega_{n-1}^{1}) \varepsilon_n(\omega_{n}^{i_{n}})}\right).
\end{eqnarray}
Further,
$$\sup\limits_{\omega_n^1 \in \Omega_n^{0-}, \omega_n^2 \in \Omega_n^{0+}} \sum\limits_{i_n=1}^2 \psi_n(\omega_1^{1}, \ldots,  \omega_{n-1}^{1}, \omega_n^{i_n})f\left(S_0
e^{\sigma_n(\omega_1^{1}, \ldots, \omega_{n-1}^{1}) \varepsilon_n(\omega_{n}^{i_{n}})}\right)=$$ 
$$\sup\limits_{\omega_n^1 \in \Omega_n^{0-}, \omega_n^2 \in \Omega_n^{0+}}\left[\frac{\Delta S_n^+(\omega_1^{1}, \ldots, \omega_{n-1}^{1}, \omega_n^2)}{V_n(\omega_1^{1}, \ldots, \omega_{n-1}^{1}, \omega_n^1, \omega_n^2)} f\left(S_0
e^{\sigma_n (\omega_1^{1}, \ldots, \omega_{n-1}^{1}) \varepsilon_n (\omega_{n}^{1})}\right)+\right.$$
$$\left.\frac{\Delta S_n^-(\omega_1^{1}, \ldots, \omega_{n-1}^{1}, \omega_n^1)}{V_n(\omega_1^{1}, \ldots, \omega_{n-1}^{1}, \omega_n^1, \omega_n^2)} f\left(S_0
e^{\sigma_n(\omega_1^{1}, \ldots, \omega_{n-1}^{1}) \varepsilon_n(\omega_{n}^{2})}\right)\right]\geq$$

$$\lim\limits_{\varepsilon(\omega_{n}^{2}) \to \infty} \lim\limits_{\varepsilon(\omega_{n}^{1}) \to -\infty}\left[\frac{e^{\sigma_n(\omega_1^{1}, \ldots, \omega_{n-1}^{1}) \varepsilon_n(\omega_{n}^{2})}-1}{e^{\sigma_n(\omega_1^{1}, \ldots, \omega_{n-1}^{1}) \varepsilon_n (\omega_{n}^{2})}- e^{\sigma_n(\omega_1^{1}, \ldots, \omega_{n-1}^{1}) \varepsilon_n(\omega_{n}^{1})}} f\left(S_0
e^{\sigma_n(\omega_1^{1}, \ldots, \omega_{n-1}^{1}) \varepsilon_n(\omega_{n}^{1})}\right)+\right.$$

$$\left.\frac{1-e^{\sigma_n(\omega_1^{1}, \ldots, \omega_{n-1}^{1}) \varepsilon_n(\omega_{n}^{1})}}{e^{\sigma_n(\omega_1^{1}, \ldots, \omega_{n-1}^{1}) \varepsilon_n(\omega_{n}^{2})}- e^{\sigma_n(\omega_1^{1}, \ldots, \omega_{n-1}^{1}) \varepsilon_n(\omega_{n}^{1})}} f\left(S_0
e^{\sigma_n(\omega_1^{1}, \ldots, \omega_{n-1}^{1}) \varepsilon_n(\omega_{n}^{2})}\right)\right]=$$
\begin{eqnarray}\label{1tin15vika9}
f(0)=K.
\end{eqnarray}
Substituting the inequality (\ref{1tin15vika9}) into the inequality (\ref{1tin5vika7}), we obtain the needed inequality.

Let us prove the equality (\ref{tinvika13}). Due to the convexity of the  payoff function $f(x),$  using the Jensen  inequality, we obtain
\begin{eqnarray}\label{tinvika17}
\inf\limits_{P \in M}E^P f(S_N) \geq f(E^PS_N)=f(S_0).
\end{eqnarray}
Let us prove the inverse inequality. It is  evident that
$$\sum\limits_{i_1=1, \ldots,i_N=1}^2 \prod\limits_{j=1}^N\psi_j(\omega_1^{i_1}, \ldots, \omega_j^{i_j}) \times $$ 

$$  f\left(S_0\prod\limits_{s=1}^N\left(1+a_s(\omega_1^{i_{1}}, \ldots, \omega_{s-1}^{i_{s-1}}) 
\left(e^{\sigma_s(\omega_1^{i_{1}}, \ldots, \omega_{s-1}^{i_{s-1}}) \varepsilon_s(\omega_{s}^{i_{s}})}-1\right) \right)\right)\geq $$
\begin{eqnarray}\label{tinvika18}
 \inf\limits_{P \in M}E^P f(S_N).
\end{eqnarray}
Putting in this inequality $\varepsilon_i(\omega_{i}^{1})=0, \  i=\overline{1,N},$ we obtain the needed.  The last statement about the  interval of non-arbitrage prices  follows from \cite{WalterSchacher} and  \cite{DMW90}.
Theorem \ref{tinvika10} is proved.
\end{proof}

\begin{te}\label{tinnavika36}
On the probability space $\{\Omega_N, {\cal F}_N, P_N\},$  let the evolution of risky asset be given by the formula (\ref{tinwau2}). Suppose that $0 \leq a_i\leq  1,$ $ \sigma_i(\omega_1, \ldots,\omega_{i-1})>\sigma_i>0,$
$ \   i=\overline{1,N}.$ 
If the nonnegative payoff function $f(x),   \ x \in [0,\infty),$ satisfies the conditions:\\
1) $f(0)=0, \ f(x) \leq a x, \  \lim\limits_{x \to \infty}\frac{f(x)}{x}=a,\ a >0,$
then the inequalities
\begin{eqnarray}\label{tinnavika37}
f\left(S_0 \prod\limits_{i=1}^N(1-a_i)\right)+ a S_0 \left(1- \prod\limits_{i=1}^N(1-a_i)\right) \leq \sup\limits_{P \in M}E^P f(S_N)\leq a S_0
\end{eqnarray}
are true.
If, in addition,  the nonnegative payoff function $f(x)$ is a convex down one, then
\begin{eqnarray}\label{tinnavika38}
\inf\limits_{P \in M}E^P f(S_N)=f(S_0), 
\end{eqnarray}
where $M$ is the set of equivalent martingale measures for the  evolution of  risky asset,  given by the formula (\ref{tinwau2}).
\end{te}
\begin{proof} As before,
$$a S_0 \geq \sup\limits_{\omega_i^1 \in \Omega_i^{0-}, \omega_i^2 \in \Omega_i^{0+}, i=\overline{1,N}}\int\limits_{\Omega_N}f(S_N) d \mu_{\{\omega_1^1,\omega_1^2\},\ldots ,\{\omega_N^1, \omega_N^2\}}=$$

$$\sup\limits_{\omega_i^1 \in \Omega_i^{0-}, \omega_i^2 \in \Omega_i^{0+}, i=\overline{1,N}}\sum\limits_{i_1=1, \ldots,i_N=1}^2 \prod\limits_{j=1}^N\psi_j(\omega_1^{i_1}, \ldots, \omega_j^{i_j}) \times $$ 
\begin{eqnarray}\label{1tinvika5}
  f\left(S_0\prod\limits_{s=1}^N\left(1+a_s 
\left(e^{\sigma_s(\omega_1^{i_{1}}, \ldots, \omega_{s-1}^{i_{s-1}}) \varepsilon_s(\omega_{s}^{i_{s}})}-1\right) \right)\right).
\end{eqnarray}

$$\sup\limits_{\omega_N^1 \in \Omega_N^{0-}, \omega_N^2 \in \Omega_N^{0+}} \sum\limits_{i_N=1}^2 \psi_N(\omega_1^{i_1}, \ldots, \omega_N^{i_N})  \times $$ 
$$f\left(S_0\prod\limits_{s=1}^N\left(1+a_s 
\left(e^{\sigma_s(\omega_1^{i_{1}}, \ldots, \omega_{s-1}^{i_{s-1}}) \varepsilon_s(\omega_{s}^{i_{s}})}-1\right) \right)\right)
=$$ 
$$\sup\limits_{\omega_N^1 \in \Omega_N^{0-}, \omega_N^2 \in \Omega_N^{0+}}\left[\frac{\Delta S_N^+(\omega_1^{i_{1}}, \ldots, \omega_{N-1}^{i_{N-1}}, \omega_N^2)}{V_N(\omega_1^{i_{1}}, \ldots, \omega_{N-1}^{i_{N-1}}, \omega_N^1, \omega_N^2)}\times\right.$$ 
$$
 f\left(S_{N-1}\left(1+a_N 
\left(e^{\sigma_N(\omega_1^{i_{1}}, \ldots, \omega_{N-1}^{i_{N-1}}) \varepsilon_N(\omega_{N}^{1})}-1\right) \right)\right)+$$

$$\left.\frac{\Delta S_N^-(\omega_1^{i_{1}}, \ldots, \omega_{N-1}^{i_{N-1}}, \omega_N^1)}{V_N(\omega_1^{i_{1}}, \ldots, \omega_{N-1}^{i_{N-1}}, \omega_N^1, \omega_N^2)}
 f\left(S_{N-1}\left(1+a_N 
\left(e^{\sigma_N(\omega_1^{i_{1}}, \ldots, \omega_{N-1}^{i_{N-1}}) \varepsilon_N(\omega_{N}^{2})}-1\right) \right)\right)
\right]\geq$$

$$\lim\limits_{\varepsilon_N(\omega_{N}^{2}) \to \infty} \lim\limits_{\varepsilon_N(\omega_{N}^{1}) \to -\infty}\left[\frac{e^{\sigma_N(\omega_1^{i_{1}}, \ldots, \omega_{N-1}^{i_{N-1}}) \varepsilon_N(\omega_{N}^{2})}-1}{e^{\sigma_N(\omega_1^{i_{1}}, \ldots, \omega_{N-1}^{i_{N-1}}) \varepsilon_N(\omega_{N}^{2})}- e^{\sigma_N(\omega_1^{i_{1}}, \ldots, \omega_{N-1}^{i_{N-1}}) \varepsilon_N(\omega_{N}^{1})}}\times\right.$$
$$ f\left(S_{N-1}\left(1+a_N 
\left(e^{\sigma_N(\omega_1^{i_{1}}, \ldots, \omega_{N-1}^{i_{N-1}}) \varepsilon_N(\omega_{N}^{1})}-1\right) \right)\right)
+$$

$$\left.\frac{1-e^{\sigma_N(\omega_1^{i_{1}}, \ldots, \omega_{N-1}^{i_{N-1}}) \varepsilon_N(\omega_{N}^{1})}}{e^{\sigma_N(\omega_1^{i_{1}}, \ldots, \omega_{N-1}^{i_{N-1}}) \varepsilon_N(\omega_{N}^{2})}- e^{\sigma_N(\omega_1^{i_{1}}, \ldots, \omega_{N-1}^{i_{N-1}}) \varepsilon_N(\omega_{N}^{1})}} \times\right.$$
$$ \left. f\left(S_{N-1}\left(1+a_N 
\left(e^{\sigma_N(\omega_1^{i_{1}}, \ldots, \omega_{N-1}^{i_{N-1}}) \varepsilon_N(\omega_{N}^{2})}-1\right) \right)\right)\right]=$$
\begin{eqnarray}\label{tinna1tinvika5}
 f(S_{N-1}(1-a_N))+a a_N S_{N-1},
\end{eqnarray}
where we put
\begin{eqnarray}\label{pupsyk5}
S_{N-1}=S_0\prod\limits_{s=1}^{N-1}\left(1+a_s 
\left(e^{\sigma_s(\omega_1^{i_{1}}, \ldots, \omega_{s-1}^{i_{s-1}}) \varepsilon_s(\omega_{s}^{i_{s}})}-1\right) \right).
\end{eqnarray}
Substituting  the inequality (\ref{tinna1tinvika5}) into  (\ref{1tinvika5}), we obtain the inequality

$$\sup\limits_{\omega_i^1 \in \Omega_i^{0-}, \omega_i^2 \in \Omega_i^{0+}, i=\overline{1,N}}\sum\limits_{i_1=1, \ldots,i_N=1}^2 \prod\limits_{j=1}^N\psi_j(\omega_1^{i_1}, \ldots, \omega_j^{i_j}) \times $$ 
$$ f\left(S_0\prod\limits_{s=1}^N\left(1+a_s 
\left(e^{\sigma_s(\omega_1^{i_{1}}, \ldots, \omega_{s-1}^{i_{s-1}}) \varepsilon_s(\omega_{s}^{i_{s}})}-1\right) \right)\right)\geq$$

$$\sup\limits_{\omega_i^1 \in \Omega_i^{0-}, \omega_i^2 \in \Omega_i^{0+}, i=\overline{1,N-1}}\sum\limits_{i_1=1, \ldots,i_{N-1}=1}^2 \prod\limits_{j=1}^{N-1}\psi_j(\omega_1^{i_1}, \ldots, \omega_j^{i_j}) \times $$ 
\begin{eqnarray}\label{tinna12vika1}
 f\left(S_0(1-a_N)\prod\limits_{s=1}^{N-1}\left(1+a_s 
\left(e^{\sigma_s(\omega_1^{i_{1}}, \ldots, \omega_{s-1}^{i_{s-1}}) \varepsilon_s(\omega_{s}^{i_{s}})}-1\right) \right)\right)+  a a_N S_0. 
\end{eqnarray}
Applying $(N-1)$ times the inequality (\ref{tinna12vika1}), we obtain the inequality

$$\sup\limits_{Q \in M}\int\limits_{\Omega}f(S_N) dQ \geq f(S_0 \prod\limits_{i=1}^N(1-a_i))+ a S_0 \sum\limits_{i=1}^N a_i \prod\limits_{s=i+1}^N(1-a_s)=$$
\begin{eqnarray}\label{tinna3vika1}
f\left(S_0 \prod\limits_{i=1}^N(1-a_i)\right)+ a S_0\left(1-\prod\limits_{i=1}^N(1-a_i)\right).
\end{eqnarray}

Let us prove the equality (\ref{tinnavika38}). Using the Jensen inequality, we obtain
\begin{eqnarray}\label{tinna1vika38}
\inf\limits_{P \in M}E^P f(S_N)\geq f(S_0). 
\end{eqnarray}
Let us prove the inverse inequality. It is  evident that 
$$\sum\limits_{i_1=1, \ldots,i_N=1}^2 \prod\limits_{j=1}^N\psi_j(\omega_1^{i_1}, \ldots, \omega_j^{i_j}) \times $$ 
$$  f\left(S_0\prod\limits_{s=1}^N\left(1+a_s 
\left(e^{\sigma_s(\omega_1^{i_{1}}, \ldots, \omega_{s-1}^{i_{s-1}}) \varepsilon_s(\omega_{s}^{i_{s}})}-1\right) \right)\right)\geq $$
\begin{eqnarray}\label{tinna1vika3}
 \inf\limits_{P \in M}E^P f(S_N).
\end{eqnarray}
Putting in the inequality (\ref{tinna1vika3}) $\varepsilon_n(\omega_n)=0, n=\overline{1,N}, $ we obtain the inverse inequality.
\end{proof}

\begin{te}\label{tinnvika39}
On the probability space $\{\Omega_N, {\cal F}_N, P_N\},$  let the evolution of risky asset be given by the formula (\ref{tinwau2}). Suppose that $0 \leq a_i\leq  1,$ $ \sigma_i(\omega_1, \ldots,\omega_{i-1})>\sigma_i>0,$
$ \   i=\overline{1,N}.$ 
If the nonnegative payoff function $f(x),   \ x \in [0,\infty),$ satisfies the  conditions:\\
1) $f(0)=K, \ f(x) \leq K, $
then 
\begin{eqnarray}\label{tinnvika40}
f\left(S_0\prod\limits_{i=1}^N(1-a_i)\right)  \leq \sup\limits_{P \in M}E^P f(S_N)\leq K.
\end{eqnarray}
If, in addition,  the nonnegative payoff function $f(x)$ is a convex down one, then
\begin{eqnarray}\label{pusstinnvika41}
\inf\limits_{P \in M}E^P f(S_N)=f(S_0), 
\end{eqnarray}
where $M$ is the set of equivalent martingale measures for the  evolution of  risky asset, given by the formula (\ref{tinwau2}).
\end{te}
\begin{proof}
Let us obtain the estimate from below. Really,
$$a K \geq \sup\limits_{\omega_i^1 \in \Omega_i^{0-}, \omega_i^2 \in \Omega_i^{0+}, i=\overline{1,N}}\int\limits_{\Omega_N}f(S_N) d \mu_{\{\omega_1^1,\omega_1^2\},\ldots ,\{\omega_N^1, \omega_N^2\}}=$$

$$\sup\limits_{\omega_i^1 \in \Omega_i^{0-}, \omega_i^2 \in \Omega_i^{0+}, i=\overline{1,N}}\sum\limits_{i_1=1, \ldots,i_N=1}^2 \prod\limits_{j=1}^N\psi_j(\omega_1^{i_1}, \ldots, \omega_j^{i_j}) \times $$ 
\begin{eqnarray}\label{puss1tinvika5}
  f\left(S_0\prod\limits_{s=1}^N\left(1+a_s 
\left(e^{\sigma_s(\omega_1^{i_{1}}, \ldots, \omega_{s-1}^{i_{s-1}}) \varepsilon_s(\omega_{s}^{i_{s}})}-1\right) \right)\right).
\end{eqnarray}

Further,
$$\sup\limits_{\omega_N^1 \in \Omega_N^{0-}, \omega_N^2 \in \Omega_N^{0+}} \sum\limits_{i_N=1}^2 \psi_N(\omega_1^{i_1}, \ldots, \omega_N^{i_N})  f\left(S_0\prod\limits_{s=1}^N\left(1+a_s 
\left(e^{\sigma_s(\omega_1^{i_{1}}, \ldots, \omega_{s-1}^{i_{s-1}}) \varepsilon_s(\omega_{s}^{i_{s}})}-1\right) \right)\right)
=$$ 
$$\sup\limits_{\omega_N^1 \in \Omega_N^{0-}, \omega_N^2 \in \Omega_N^{0+}}\left[\frac{\Delta S_N^+(\omega_1^{i_{1}}, \ldots, \omega_{N-1}^{i_{N-1}}, \omega_N^2)}{V_N(\omega_1^{i_{1}}, \ldots, \omega_{N-1}^{i_{N-1}}, \omega_N^1, \omega_N^2)}\times\right.$$ $$
 f\left(S_{N-1}\left(1+a_N 
\left(e^{\sigma_N(\omega_1^{i_{1}}, \ldots, \omega_{N-1}^{i_{N-1}}) \varepsilon_N(\omega_{N}^{1})}-1\right) \right)\right)+$$

$$\left.\frac{\Delta S_N^-(\omega_1^{i_{1}}, \ldots, \omega_{N-1}^{i_{N-1}}, \omega_N^1)}{V_N(\omega_1^{i_{1}}, \ldots, \omega_{N-1}^{i_{N-1}}, \omega_N^1, \omega_N^2)}
 f\left(S_{N-1}\left(1+a_N 
\left(e^{\sigma_N(\omega_1^{i_{1}}, \ldots, \omega_{N-1}^{i_{N-1}}) \varepsilon_N(\omega_{N}^{2})}-1\right) \right)\right)
\right]\geq$$

$$\lim\limits_{\varepsilon_N(\omega_{N}^{2}) \to \infty} \lim\limits_{\varepsilon_N(\omega_{N}^{1}) \to - \infty}\left[\frac{e^{\sigma_N(\omega_1^{i_{1}}, \ldots, \omega_{N-1}^{i_{N-1}}) \varepsilon_N(\omega_{N}^{2})}-1}{e^{\sigma_N(\omega_1^{i_{1}}, \ldots, \omega_{N-1}^{i_{N-1}}) \varepsilon_N(\omega_{N}^{2})}- e^{\sigma_N(\omega_1^{i_{1}}, \ldots, \omega_{N-1}^{i_{N-1}}) \varepsilon_N(\omega_{N}^{1})}}\times\right.$$ $$
 f\left(S_{N-1}\left(1+a_N 
\left(e^{\sigma_N(\omega_1^{i_{1}}, \ldots, \omega_{N-1}^{i_{N-1}}) \varepsilon_N(\omega_{N}^{1})}-1\right) \right)\right)
+$$

$$\left.\frac{1-e^{\sigma_N(\omega_1^{i_{1}}, \ldots, \omega_{N-1}^{i_{N-1}}) \varepsilon_N(\omega_{N}^{1})}}{e^{\sigma_N(\omega_1^{i_{1}}, \ldots, \omega_{N-1}^{i_{N-1}}) \varepsilon_N(\omega_{N}^{2})}- e^{\sigma_N(\omega_1^{i_{1}}, \ldots, \omega_{N-1}^{i_{N-1}}) \varepsilon_N(\omega_{N}^{1})}} \times\right.$$ $$\left.
 f\left(S_{N-1}\left(1+a_N 
\left(e^{\sigma_N(\omega_1^{i_{1}}, \ldots, \omega_{N-1}^{i_{N-1}}) \varepsilon_N(\omega_{N}^{2})}-1\right) \right)\right)\right]=$$
\begin{eqnarray}\label{pusstinna1tinvika5}
 f(S_{N-1}(1-a_N)),
\end{eqnarray}
where we put
\begin{eqnarray}\label{pupsyk15}
S_{N-1}=S_0\prod\limits_{s=1}^{N-1}\left(1+a_s 
\left(e^{\sigma_s(\omega_1^{i_{1}}, \ldots, \omega_{s-1}^{i_{s-1}}) \varepsilon_s(\omega_{s}^{i_{s}})}-1\right) \right).
\end{eqnarray}
Substituting  the inequality (\ref{pusstinna1tinvika5}) into  (\ref{puss1tinvika5}), we obtain the inequality
$$\sup\limits_{\omega_i^1 \in \Omega_i^{0-}, \omega_i^2 \in \Omega_i^{0+}, i=\overline{1,N}}\sum\limits_{i_1=1, \ldots,i_N=1}^2 \prod\limits_{j=1}^N\psi_j(\omega_1^{i_1}, \ldots, \omega_j^{i_j}) \times $$ 
$$ f\left(S_0\prod\limits_{s=1}^N\left(1+a_s 
\left(e^{\sigma_s(\omega_1^{i_{1}}, \ldots, \omega_{s-1}^{i_{s-1}}) \varepsilon_s(\omega_{s}^{i_{s}})}-1\right) \right)\right)\geq$$

$$\sup\limits_{\omega_i^1 \in \Omega_i^{0-}, \omega_i^2 \in \Omega_i^{0+}, i=\overline{1,N-1}}\sum\limits_{i_1=1, \ldots,i_{N-1}=1}^2 \prod\limits_{j=1}^{N-1}\psi_j(\omega_1^{i_1}, \ldots, \omega_j^{i_j}) \times $$ 
\begin{eqnarray}\label{pusstinna12vika1}
 f\left(S_0(1-a_N)\prod\limits_{s=1}^{N-1}\left(1+a_s 
\left(e^{\sigma_s(\omega_1^{i_{1}}, \ldots, \omega_{s-1}^{i_{s-1}}) \varepsilon_s(\omega_{s}^{i_{s}})}-1\right) \right)\right). 
\end{eqnarray}
Applying $(N-1)$ times the inequality (\ref{pusstinna12vika1}), we obtain the inequality
\begin{eqnarray}\label{pusstinna3vika1}
\sup\limits_{Q \in M}\int\limits_{\Omega}f(S_N) dQ \geq f(S_0 \prod\limits_{i=1}^N (1-a_i)).
\end{eqnarray}

Let us prove the equality (\ref{pusstinnvika41}). Using the Jensen inequality we obtain
\begin{eqnarray}\label{pusstinna1vika38}
\inf\limits_{P \in M}E^P f(S_N)\geq f(S_0). 
\end{eqnarray}
Let us prove the inverse inequality. It is  evident that
 $$\sum\limits_{i_1=1, \ldots,i_N=1}^2 \prod\limits_{j=1}^N\psi_j(\omega_1^{i_1}, \ldots, \omega_j^{i_j}) \times $$ 
$$  f\left(S_0\prod\limits_{s=1}^N\left(1+a_s 
\left(e^{\sigma_s(\omega_1^{i_{1}}, \ldots, \omega_{s-1}^{i_{s-1}}) \varepsilon_s(\omega_{s}^{i_{s}})}-1\right) \right)\right)\geq $$
\begin{eqnarray}\label{pusstinna1vika3}
 \inf\limits_{P \in M}E^P f(S_N).
\end{eqnarray}
Putting in the inequality (\ref{pusstinna1vika3}) $\varepsilon_n(\omega_n)=0, \ n=\overline{1,N}, $ we obtain the inverse inequality.
\end{proof}

\begin{te}\label{pussytinavika1}
On the probability space $\{\Omega_N, {\cal F}_N, P_N\},$  let the evolution of risky asset be given by the formula (\ref{tinwau2}). Suppose that $0 \leq a_i\leq  1,$ $ \sigma_i(\omega_1, \ldots,\omega_{i-1})>\sigma_i>0,$
$ \   i=\overline{1,N}.$ 
For the  payoff function  $f(x)=(x-K)^+, \ x \in (0, \infty), \ K >0,  $   the fair price of super-hedge is given by the  formula 
$$\sup\limits_{Q \in M} E^Qf(S_N)=$$
 \begin{eqnarray}\label{pussytinavika2}
 \left\{\begin{array}{l l} (S_0 - K)^+, & \mbox{if} \quad  S_0 \prod\limits_{ i=1}^N(1-a_i))  \geq  K,\\
S_0\left(1- \prod\limits_{ i=1}^N(1-a_i)\right), & \mbox{if} \quad S_0 \prod\limits_{ i=1}^N(1-a_i) < K.
\end{array} \right.
\end{eqnarray}
If $S_0 \prod\limits_{ i=1}^N(1-a_i))  \geq  K,$
 then the set of non arbitrage prices coincides with the point $(S_0-K)^+, $ in case if $S_0 \prod\limits_{ i=1}^N(1-a_i) < K$ the set of non arbitrage prices coincides with the set $\left[(S_0-K)^+, S_0\left(1-\prod\limits_{ i=1}^N(1-a_i)\right)\right].$
\end{te}
\begin{proof} 
Let us introduce the denotations 

$$I_N=\sum\limits_{i_1=1, \ldots,i_N=1}^2 \prod\limits_{j=1}^N\psi_j(\omega_1^{i_1}, \ldots, \omega_j^{i_j}) \times $$
\begin{eqnarray}\label{pussy22tinavika2}
 f\left(S_0\prod\limits_{s=1}^N\left(1+a_s 
\left(e^{\sigma_s(\omega_1^{i_{1}}, \ldots, \omega_{s-1}^{i_{s-1}}) \varepsilon_s(\omega_{s}^{i_{s}})}-1\right) \right)\right),
\end{eqnarray}
$$I_N^1=\sum\limits_{i_1=1, \ldots,i_N=1}^2 \prod\limits_{j=1}^N\psi_j(\omega_1^{i_1}, \ldots, \omega_j^{i_j}) \times $$
\begin{eqnarray}\label{pussy222tinavika2}
 f_1\left(S_0\prod\limits_{s=1}^N\left(1+a_s 
\left(e^{\sigma_s(\omega_1^{i_{1}}, \ldots, \omega_{s-1}^{i_{s-1}}) \varepsilon_s(\omega_{s}^{i_{s}})}-1\right) \right)\right),
\end{eqnarray}
$$I_N^0=\sup\limits_{\omega_i^1 \in \Omega_i^{0-}, \omega_i^2 \in \Omega_i^{0+}, i=\overline{1,N}}\sum\limits_{i_1=1, \ldots,i_N=1}^2 \prod\limits_{j=1}^N\psi_j(\omega_1^{i_1}, \ldots, \omega_j^{i_j}) \times $$
\begin{eqnarray}\label{pussy220tinavika2}
 f\left(S_0\prod\limits_{s=1}^N\left(1+a_s 
\left(e^{\sigma_s(\omega_1^{i_{1}}, \ldots, \omega_{s-1}^{i_{s-1}}) \varepsilon_s(\omega_{s}^{i_{s}})}-1\right) \right)\right),
\end{eqnarray}
where we put $f_1(x)=(K- x)^+.$
Let us estimate from above the value $I_N.$ For this we use the equality
\begin{eqnarray}\label{pussytinavika2}
I_N = I_N^1+S_0-K,
\end{eqnarray}
which follows from the identity: $f(x)=f_1(x) +x-K, \ x \geq 0.$
Since 
\begin{eqnarray}\label{pussytinavika7}
  f_1\left(S_0\prod\limits_{s=1}^N\left(1+a_s 
\left(e^{\sigma_s(\omega_1^{i_{1}}, \ldots, \omega_{s-1}^{i_{s-1}}) \varepsilon_s(\omega_{s}^{i_{s}})}-1\right) \right)\right) \leq f_1\left(S_0\prod\limits_{s=1}^N (1-a_s)\right),
\end{eqnarray}
we obtain the inequality
\begin{eqnarray}\label{pussy12tinavika8}
I_N \leq S_0 -K + f_1\left(S_0\prod\limits_{s=1}^N (1-a_s)\right).
\end{eqnarray}
From the inequality (\ref{pussy12tinavika8}),   we have
$$I_N^0 \leq S_0 -K + f_1\left(S_0\prod\limits_{s=1}^N (1-a_s))\right)=
$$
 \begin{eqnarray}\label{pussytinavika8}
 \left\{\begin{array}{l l} (S_0 - K)^+, & \mbox{if} \quad  S_0 \prod\limits_{ i=1}^N(1-a_i))  \geq  K,\\
S_0\left(1- \prod\limits_{ i=1}^N(1-a_i)\right), & \mbox{if} \quad S_0 \prod\limits_{ i=1}^N(1-a_i) <  K.
\end{array} \right.
\end{eqnarray}
Due to  the  inequality  (\ref{tinnavika37}) of Theorem \ref{tinnavika36},
\begin{eqnarray}\label{pussytinavika5}
I_N^0 \geq f\left(S_0 \prod\limits_{i=1}^N(1-a_i)\right)+  S_0 \left(1- \prod\limits_{i=1}^N(1-a_i)\right)
\end{eqnarray}
 and the inequality
\begin{eqnarray}\label{pussy5tinavika5}
I_N^0  \geq (S_0 -K)^+,
\end{eqnarray}
which follows from the Jensen inequality, we have
$$I_N^0 \geq  \max\left \{ S_0 - K)^+,  f\left(S_0 \prod\limits_{i=1}^N(1-a_i)\right)+  S_0 \left(1- \prod\limits_{i=1}^N(1-a_i) \right)\right\}=$$
\begin{eqnarray}\label{pussy88tinavika8}
 \left\{\begin{array}{l l} (S_0 - K)^+, & \mbox{if} \quad  S_0 \prod\limits_{ i=1}^N(1-a_i))  \geq  K,\\
S_0\left(1- \prod\limits_{ i=1}^N(1-a_i)\right), & \mbox{if} \quad S_0 \prod\limits_{ i=1}^N(1-a_i) <  K.
\end{array} \right.
\end{eqnarray}
This proves Theorem  \ref{pussytinavika1}.
\end{proof}

\begin{te}\label{pussytinavika9}
On the probability space $\{\Omega_N, {\cal F}_N, P_N\},$  let the evolution of risky asset be given by the formula (\ref{tinwau2}). Suppose that $0 \leq a_i\leq  1,$ $ \sigma_i(\omega_1, \ldots,\omega_{i-1})>\sigma_i>0,$
$\   i=\overline{1,N}.$ 
For the  payoff function  $f_1(x)=(K-x)^+, \ x \in (0, \infty), \ K >0, $    the fair price of super-hedge is given by the  formula 
 \begin{eqnarray}\label{pussytinavika10}
\sup\limits_{Q \in M} E^Qf_1(S_N)=f_1\left(S_0 \prod\limits_{i=1}^N(1-a_i)\right).
\end{eqnarray}
 The set of non arbitrage prices coincides  with the  interval \\ $\left[(K-S_0)^+, f_1\left(S_0\prod\limits_{ i=1}^N(1-a_i)\right)\right].$
\end{te}
\begin{proof} The inequality
$$I_N^1=\sum\limits_{i_1=1, \ldots,i_N=1}^2 \prod\limits_{j=1}^N\psi_j(\omega_1^{i_1}, \ldots, \omega_j^{i_j}) \times $$
\begin{eqnarray}\label{pussytinavika11}
 f_1\left(S_0\prod\limits_{s=1}^N\left(1+a_s 
\left(e^{\sigma_s(\omega_1^{i_{1}}, \ldots, \omega_{s-1}^{i_{s-1}}) \varepsilon_s(\omega_{s}^{i_{s}})}-1\right) \right)\right)\leq  f_1\left(S_0\prod\limits_{ i=1}^N(1-a_i)\right)
\end{eqnarray}
is true.
Taking into account the  inequality (\ref{tinnvika40}) of Theorem \ref{tinnvika39}, we prove Theorem \ref{pussytinavika9}.
\end{proof}

\begin{te}\label{pussytinavika12}
On the probability space $\{\Omega_N, {\cal F}_N, P_N\},$  let the evolution of risky asset be given by the formula (\ref{tinwau2}). Suppose that $0 \leq a_i\leq  1,$ $ \sigma_i(\omega_1, \ldots,\omega_{i-1})>\sigma_i>0,$
$ \   i=\overline{1,N}.$ 
 For the payoff function  $f_1(S_0,S_1, \ldots, S_N)=\left(K-\frac{\sum\limits_{i=0}^N S_i}{N+1}\right)^+, \ K >0,  $  the fair price of super-hedge is given  by the  formula 
 \begin{eqnarray}\label{pussytinavika13}
\sup\limits_{Q \in M} E^Qf_1(S_0,S_1, \ldots, S_N)=\left(K-\frac{S_0\sum\limits_{i=0}^N\prod\limits_{s=1}^i(1-a_s)}{N+1}\right)^+.
\end{eqnarray}
 The set of non arbitrage prices coincides with the  interval \\ $\left[(K-S_0)^+, \left(K-\frac{S_0\sum\limits_{i=0}^N\prod\limits_{s=1}^i(1-a_s)}{N+1}\right)^+\right],$ if $K>\frac{S_0\sum\limits_{i=0}^N\prod\limits_{s=1}^i(1-a_s)}{N+1}.$ \\ For $K\leq \frac{S_0\sum\limits_{i=0}^N\prod\limits_{s=1}^i(1-a_s)}{N+1}$ the set of non arbitrage prices coincides with the point $0.$
\end{te}
\begin{proof}
Let us denote
$$S_n(\omega_1^1, \ldots, \omega_n^1)=S_0\prod\limits_{s=1}^n\left(1+a_s 
\left(e^{\sigma_s(\omega_1^{1}, \ldots, \omega_{s-1}^{1}) \varepsilon_s(\omega_{s}^{1})}-1\right)\right), \quad n=\overline{1,N},$$
\begin{eqnarray}\label{1pusytinavika13}
t_N(\omega_1^1, \ldots, \omega_N^1)=\prod\limits_{s=1}^N\frac{ e^{\sigma_s(\omega_1^{1}, \ldots, \omega_{s-1}^{1})\varepsilon_s(\omega_{s}^{2})}-1}{e^{\sigma_s(\omega_1^{1}, \ldots, \omega_{s-1}^{1})\varepsilon_s(\omega_{s}^{2})}- e^{\sigma_s(\omega_1^{1}, \ldots, \omega_{s-1}^{1})\varepsilon_s(\omega_{s}^{1})}}.
\end{eqnarray}
It is evident that 
$$I_N^2=\sup\limits_{\omega_i^1 \in \Omega_i^{0-}, \omega_i^2 \in \Omega_i^{0+}, i=\overline{1,N}}\sum\limits_{i_1=1, \ldots,i_N=1}^2 \prod\limits_{j=1}^N\psi_j(\omega_1^{i_1}, \ldots, \omega_j^{i_j}) \times $$
\begin{eqnarray}\label{pussy220tinavika2}
 f_1\left(S_0\prod\limits_{s=1}^N\left(1+a_s 
\left(e^{\sigma_s(\omega_1^{i_{1}}, \ldots, \omega_{s-1}^{i_{s-1}}) \varepsilon_s(\omega_{s}^{i_{s}})}-1\right) \right)\right)\geq
\end{eqnarray}
$$ \lim\limits_{\varepsilon_s(\omega_s^{1})=-\infty, \ \varepsilon_s(\omega_s^{2}) \to \infty, s=\overline{1,N}}
 f_1\left(S_0, S_1(\omega_1^{1}), \ldots, S_N(\omega_1^1, \ldots, \omega_N^1 )\right)\times \\$$
$$t_N(\omega_1^1, \ldots, \omega_N^1)
=f_1\left(S_0,S_0(1-a_1), \ldots, S_0\prod\limits_{s=1}^N(1-a_s)\right) =$$
\begin{eqnarray}\label{pussytinavika14}
\left(K-\frac{S_0\sum\limits_{i=0}^N\prod\limits_{s=1}^i(1-a_s)}{N+1}\right)^+.
\end{eqnarray}
Let us prove the inverse inequality.  We have 
$$I_N^2 \leq  \sup\limits_{\omega_i^1 \in \Omega_i^{0-}, \omega_i^2 \in \Omega_i^{0+}, i=\overline{1,N}}\sum\limits_{i_1=1, \ldots,i_N=1}^2 \prod\limits_{j=1}^N\psi_j(\omega_1^{i_1}, \ldots, \omega_j^{i_j}) \times $$
$$f_1\left(S_0,S_0(1-a_1), \ldots, S_0\prod\limits_{s=1}^N(1-a_s)\right) =$$
\begin{eqnarray}\label{pussytinavika15}
f_1\left(S_0,S_0(1-a_1), \ldots, S_0\prod\limits_{s=1}^N(1-a_s)\right)=\left(K-\frac{S_0\sum\limits_{i=0}^N\prod\limits_{s=1}^N(1-a_s)}{N+1}\right)^+.
\end{eqnarray}  
Therefore,
\begin{eqnarray}\label{pussy17tinavika15}
 I_N^2 \leq  \left(K-\frac{S_0\sum\limits_{i=0}^N\prod\limits_{s=1}^i(1-a_s)}{N+1}\right)^+.
\end{eqnarray}

The inequalities (\ref{pussytinavika14}), (\ref{pussy17tinavika15}) prove Theorem \ref{pussytinavika12}.
\end{proof}

\begin{te}\label{pupsiktinavika16}
On the probability space $\{\Omega_N, {\cal F}_N, P_N\},$  let the evolution of risky asset be given by the formula (\ref{tinwau2}). Suppose that $0 \leq a_i\leq  1,$ $ \sigma_i(\omega_1, \ldots,\omega_{i-1})>\sigma_i>0,$
$ \   i=\overline{1,N}.$ 
For the payoff function  $f(S_0,S_1, \ldots, S_N)=\left(\frac{\sum\limits_{i=0}^N S_i}{N+1}- K\right)^+, \ K >0,  $     the fair price of super-hedge is given by the  formula 
$$ \sup\limits_{Q \in M} E^Qf(S_0, S_1, \ldots,S_N)=  $$
 \begin{eqnarray}\label{pupsiktinavika17}
\left\{\begin{array}{l l} (S_0 - K)^+, & \mbox{if} \quad \frac{ S_0\sum\limits_{i=0}^N \prod\limits_{ s=1}^i (1-a_i)}{N+1} \geq K,\\
S_0\left(1- \frac{\sum\limits_{i=0}^N\prod\limits_{ s=1}^i (1-a_s)}{N+1}\right), & \mbox{if} \quad S_0 \frac{\sum\limits_{i=0}^N\prod\limits_{ s=1}^i(1-a_s)}{N+1} < K.
\end{array} \right.
\end{eqnarray}
If $\frac{ S_0\sum\limits_{i=0}^N \prod\limits_{ s=1}^i (1-a_i)}{N+1} \geq K,$ then the set of non arbitrage prices coincides with the point $(S_0-K)^+, $ in case if $ S_0 \frac{\sum\limits_{i=0}^N\prod\limits_{ s=1}^i(1-a_s)}{N+1} < K$ the set of non arbitrage prices coincides with the   interval $\left[(S_0-K)^+, 
S_0\left(1- \frac{\sum\limits_{i=0}^N\prod\limits_{ s=1}^i (1-a_s)}{N+1}\right)           \right].$
\end{te}
\begin{proof}
We have 
$$\sup\limits_{\omega_i^1 \in \Omega_i^{0-}, \omega_i^2 \in \Omega_i^{0+}, i=\overline{1,N}}\sum\limits_{i_1=1, \ldots,i_N=1}^2 \prod\limits_{j=1}^N\psi_j(\omega_1^{i_1}, \ldots, \omega_j^{i_j}) \times $$
\begin{eqnarray}\label{pupsiktinavika18}
 f\left(S_0\prod\limits_{s=1}^N\left(1+a_s 
\left(e^{\sigma_s(\omega_1^{i_{1}}, \ldots, \omega_{s-1}^{i_{s-1}}) \varepsilon_s(\omega_{s}^{i_{s}})}-1\right) \right)\right)=
\end{eqnarray}

$$\sup\limits_{\omega_i^1 \in \Omega_i^{0-}, \omega_i^2 \in \Omega_i^{0+}, i=\overline{1,N}}\sum\limits_{i_1=1, \ldots,i_N=1}^2 \prod\limits_{j=1}^N\psi_j(\omega_1^{i_1}, \ldots, \omega_j^{i_j}) \times $$
\begin{eqnarray}\label{pupsiktinavika18}
 f_1\left(S_0\prod\limits_{s=1}^N\left(1+a_s 
\left(e^{\sigma_s(\omega_1^{i_{1}}, \ldots, \omega_{s-1}^{i_{s-1}}) \varepsilon_s(\omega_{s}^{i_{s}})}-1\right) \right)\right)+S_0 -K=
\end{eqnarray}
$$(S_0-K) + \left(K-\frac{S_0\sum\limits_{i=0}^N\prod\limits_{s=1}^i(1-a_s)}{N+1}\right)^+=$$
 \begin{eqnarray}\label{1pupsiktinavika18}
\left\{\begin{array}{l l} (S_0 - K)^+, & \mbox{if} \quad \frac{ S_0\sum\limits_{i=0}^N \prod\limits_{ s=1}^i (1-a_i)}{N+1} \geq K,\\
S_0\left(1- \frac{\sum\limits_{i=0}^N\prod\limits_{ s=1}^i (1-a_s)}{N+1}\right), & \mbox{if} \quad S_0 \frac{\sum\limits_{i=0}^N\prod\limits_{ s=1}^i(1-a_s)}{N+1} < K.
\end{array} \right.
\end{eqnarray}
In the formula (\ref{pupsiktinavika18}) we introduced the denotation 
 \begin{eqnarray}\label{pupsiktinavika19}
f_1(S_0,S_1, \ldots, S_N)=\left(K - \frac{\sum\limits_{i=0}^N S_i}{N+1}\right)^+.
\end{eqnarray}
The proof of Theorem  \ref{pupsiktinavika16} follows from the equality (\ref{pupsiktinavika18}).
\end{proof}

\section{Estimation of parameters.}
 Suppose that   $\{ g_i(X_N)\}_{i=1}^N$ is a mapping from the set  $[0,1]^N$ into itself, where $X_N=\{x_1,\ldots,x_N\}, \  0 \leq \ x_i \leq 1, \ i=\overline{1,N}.$ If $S_0, S_1, \ldots, S_N$ is a sample of the process (\ref{tinwau2}), let us  denote the  order statistic $S_{(0)}, S_{(1)},\ldots, S_{(N)}$  of this sample. Introduce also  the denotation  $g_i\left([S]_{N}\right)=g_i\left(\frac{S_{(0)}}{S_{(N)}},\ldots, \frac{S_{(N-1)}}{S_{(N)}}\right), \ i=\overline{1,N}.$

\begin{te}\label{puptanijalib5}
Suppose that $S_0, S_1, \ldots, S_N$ is 
a sample of the random process   (\ref{tinwau2}). Then, for the parameters 
$a_1, \ldots, a_N$ the estimation
$$a_1=1 - \tau_0 \frac{S_{(0)}}{S_0}g_1\left([S]_{N}\right),\quad 0< \tau_0\leq 1,$$
\begin{eqnarray}\label{tanijalib6}
a_i=1- \frac{g_i\left([S]_{N}\right)}{ g_{i-1}\left([S]_{N}\right)}, \quad i=\overline{2,N},
\end{eqnarray}
is valid,  if  for $ \ g_N([S]_{N})>0, \ [S]_{N} \in [0,1]^N, $ the  inequalities $ g_1([S]_{N}) \geq g_2([S]_{N}) \geq \ldots \geq g_N([S]_{N}) $   are true.
If $\tau_0=0,$ then $a_i=1, \ i=\overline{1,N}.$
\end{te}
\begin{proof}
The estimation of the parameters $a_1, \ldots, a_N$ we do using the representation of random process $S_n, \ n=\overline{1,N}.$ 
The smallest value of  the random variable $S_n$ is equal $S_0\prod\limits_{i=1}^n(1-a_i), \ n=\overline{1,N}.$ 
Let us determine the parameters  $a_i$ from the relations

$$S_0 \prod\limits_{i=1}^N(1-a_i)=\tau g_N\left([S]_{N}\right), \ldots, S_{0}\prod\limits_{i=1}^{N-k}(1-a_i)=\tau g_{N-k}\left([S]_{N}\right), \ldots,$$
\begin{eqnarray}\label{tanijalib7}
 S_0\prod\limits_{i=1}^{N-k-1}(1-a_i)=\tau g_{N-k-1}\left([S]_{N}\right),
\ldots, S_0(1-a_1)=\tau g_1\left([S]_{N}\right),
\end{eqnarray}
where $\tau >0.$ Taking into account  the relations (\ref{tanijalib7}), we obtain 

$$  S_0(1-a_1)=\tau g_1\left([S]_{N}\right), $$
\begin{eqnarray}\label{100tanijalib7}
\tau g_{N-k-1}\left([S]_{N}\right)(1- a_{N-k})=\tau g_{N-k}\left([S]_{N}\right), \quad k=\overline{2,N}. 
\end{eqnarray}
Solving the relations (\ref{100tanijalib7}), we have
\begin{eqnarray}\label{101tanijalib7}
a_1=1-\frac{\tau}{S_0}g_1\left([S]_{N}\right), \quad
a_{N-k}=1-\frac{g_{N-k}\left([S]_{N}\right)}{g_{N-k-1}\left([S]_{N}\right)},\quad k=\overline{2,N}.
\end{eqnarray}
It is evident that $a_{N-k} \geq 0, \  k=\overline{2,N}.$ To provide the positiveness
 of $a_1$ and the inequalities $ \tau g_{N-n}\left([S]_{N}\right) \leq S_{N-n}, \ n=\overline{0,N-1}, \ S_0 \geq S_{(0)},$ meaning that  the random process   (\ref{tinwau2}) takes all  the values  from the sample $S_n, \ n=\overline{0,N},$  we must to put $\tau =\tau_0 S_{(0)}, \ 0< \tau_0 \leq 1.$  It is evident that, if $\tau_0=0,$ then $a_i=1, \  i=\overline{1,N}$ Theorem  \ref{puptanijalib5} is proved.
\end{proof}
\begin{remark}\label{vitochkapupsyk1}
It is evident that 

$$a_i=1, \quad i=\overline{N-k,N},  \ 1<k \leq N-1, \
a_i=1- \frac{g_i([S]_N)}{g_{i-1}([S]_N)}, \ i=\overline{2,N-k-1},$$
\begin{eqnarray}\label{vitochkapupsyk2}
 a_1=1-\frac{\tau_0 S_{(0)}}{S_0}g_1([S]_N), \quad 0< \tau_0\leq 1,
\end{eqnarray}
is also estimation of the parameters $a_1, \ldots, a_N$ if 
$$0< g_{N-k-1}([S]_{N}) \leq g_{N-k-2}([S]_{N}) \ldots \leq  g_1([S]_{N}), \ [S]_{N} \in [0,1]^N.$$ 
Such estimation is not interesting since 
$$\prod\limits_{i=1}^{N-i}(1-a_i)=0, \quad i=\overline{0,k}.  $$
\end{remark}

\begin{remark}\label{vitochka1}
If 
\begin{eqnarray}\label{vitochka2}
g(x)=\left\{\begin{array}{l l} \frac{S_0}{S_{(0)}} x, & \mbox{if} \quad  0 \leq x \leq \frac{S_{(0)}}{S_0},\\
1, & \mbox{if} \quad  \frac{S_{(0)}}{S_0}< x \leq 1,
\end{array} \right.
\end{eqnarray}
 $$ g_{i}([S]_N)=g\left(\frac{S_{(N-i)}}{S_{(N)}}\right), \quad  i=\overline{1, N}, \ \tau_0=1,$$
then for the parameters 
$a_1, \ldots, a_N$ the estimation

\begin{eqnarray}\label{vitochka3}
a_i=\left\{\begin{array}{l l} 1- \frac{S_{(N-i)}}{S_{(N-i+1)}},  & \mbox{if} \quad   \frac{S_{(N-i+1)}}{S_{(N)}} \leq \frac{S_{(0)}}{S_0},\\
1 - \frac{S_{(N-i)}}{S_{(N)}}\frac{S_{0}}{S_{(0)}}, & \mbox{if} \quad  \frac{S_{(N-i+1)}}{S_{(N)}}> \frac{S_{(0)}}{S_0}, \  \frac{S_{(N-i)}}{S_{(N)}}\leq  \frac{S_{(0)}}{S_0},\\
0, & \mbox{if} \quad \frac{S_{(N-i)}}{S_{(N)}}> \frac{S_{(0)}}{S_0}.
\end{array} \right. \quad  i=\overline{2,N},
\end{eqnarray}
\begin{eqnarray}\label{vitochka4}
a_1=\left\{\begin{array}{l l} 1- \frac{S_{(N-1)}}{S_{(N)}},  & \mbox{if} \quad   \frac{S_{(N-1)}}{S_{(N)}} \leq \frac{S_{(0)}}{S_0},\\
1 - \frac{S_{(0)}}{S_{0}}, & \mbox{if} \quad  \frac{S_{(N-1)}}{S_{(N)}}> \frac{S_{(0)}}{S_0}
\end{array} \right.
\end{eqnarray}
is true. 
The following equalities 
$$ \prod\limits_{i=1}^N(1-a_i)=\frac{S_{(0)}}{S_{0}} g\left(\frac{S_{(0)}}{S_{(N)}}\right)=\frac{S_{(0)}}{S_{(N)}}, $$
\begin{eqnarray}\label{vitochka5}
 \prod\limits_{i=1}^{N-k}(1-a_i)=\left\{\begin{array}{l l} \frac{S_{(k)}}{S_{(N)}},  & \mbox{if} \quad   \frac{S_{(k)}}{S_{(N)}} \leq \frac{S_{(0)}}{S_0},\\
 \frac{S_{(0)}}{S_0}, & \mbox{if} \quad  \frac{S_{(k)}}{S_{(N)}}> \frac{S_{(0)}}{S_0},
\end{array} \right. \quad  k=\overline{1,N-1},
\end{eqnarray}
are valid.
\end{remark}

\begin{remark}\label{vitochkapus1}
Suppose that $g(x)=x, \ x \in [0,1].$ Let us put $g_{N-i}([S]_N)=g(\frac{S_{(i)}}{S_{(N)}})= \frac{S_{(i)}}{S_{(N)}}, \ i=\overline{0,k}, \ g_{N-i}([S]_N)=1, \ i=\overline{k+1, N-1}.  $
Then, 

$$a_1=1-\tau_0 \frac{S_{(0)}}{S_0}, \quad 0< \tau_0 \leq 1,  \quad  a_i=0,  \quad  i=\overline{2, N-k-1}, $$
\begin{eqnarray}\label{vitochkapus2}
 a_i=1-\frac{g_i([S]_N)}{g_{i-1}([S]_N)}, \quad  i=\overline{ N-k, N},
\end{eqnarray}
is an estimation  for the parameters $a_1, \ldots, a_N.$
\end{remark}

In the next Theorems we put $\tau_0=1.$ This corresponds to the fact that fair price of super-hedge is minimal for the considered statistic.

\begin{te}\label{tanijalib8}
On the probability space $\{\Omega_N, {\cal F}_N, P_N\},$  let the evolution of risky asset be given by the formula (\ref{tinwau2}),
with parameters  $ a_i, \ i=\overline{1,N},$ given by the formula (\ref{tanijalib6}).  For the payoff function  $f(x)=(x-K)^+, \ x \in (0, \infty), \ K >0,  $   the fair price of super-hedge is given by the  formula 
$$\sup\limits_{Q \in M} E^Qf(S_N)=$$
 \begin{eqnarray}\label{tanijalib9}
 \left\{\begin{array}{l l} (S_0 - K)^+, & \mbox{if} \quad    S_{(0)}  g_N\left([S]_{N}\right) \geq  K,\\
S_0\left(1- \frac{S_{(0)}  g_N\left([S]_{N}\right)}{S_{0} }\right), & \mbox{if} \quad  S_{(0)}  g_N\left([S]_{N}\right)  < K.
\end{array} \right.
\end{eqnarray}
If $ S_{(0)}  g_N\left([S]_{N}\right)  \geq  K,$
 then the set of non arbitrage prices coincides with the point $(S_0-K)^+, $ in case if $ S_{(0)}  g_N\left([S]_{N}\right) < K$ the set of non arbitrage prices coincides with the closed set $\left[(S_0-K)^+, S_0\left(1- \frac{S_{(0)}  g_N\left([S]_{N}\right)}{S_{0} }\right)\right].$

The fair price of super-hedge  for the statistic  (\ref{vitochka3}), (\ref{vitochka4}) is given by the formula 
\begin{eqnarray}\label{vitochka6}
\sup\limits_{Q \in M} E^Qf(S_N)= \left\{\begin{array}{l l} (S_0 - K)^+, & \mbox{if} \quad    S_0\frac{S_{(0)}}{S_{(N)}}   \geq  K,\\
S_0\left(1- \frac{S_{(0)}}{S_{(N)}}\right), & \mbox{if} \quad   S_0\frac{S_{(0)}}{S_{(N)}} < K.
\end{array} \right.
\end{eqnarray}
If $  S_0\frac{S_{(0)}}{S_{(N)}}   \geq  K,$
 then the set of non arbitrage prices coincides with the point $(S_0-K)^+, $ in case if $  S_0\frac{S_{(0)}}{S_{(N)}} < K$ the set of non arbitrage prices coincides with the closed set $\left[(S_0-K)^+, S_0\left(1- \frac{S_{(0)}}{S_{(N)}}\right)\right].$

The fair price of super-hedge is minimal one for the statistic (\ref{tanijalib6}) with $g_i(X_N)=g_N(X_N)=1, \ i=\overline{1,N-1},$ and is given by the formula
\begin{eqnarray}\label{!00tanijalib9}
\sup\limits_{Q \in M} E^Qf(S_N)= \left\{\begin{array}{l l} (S_0 - K)^+, & \mbox{if} \quad    S_{(0)}   \geq  K,\\
S_0- S_{(0)}, & \mbox{if} \quad  S_{(0)}   < K.
\end{array} \right.
\end{eqnarray}
If $ S_{(0)}   \geq  K,$
 then the set of non arbitrage prices coincides with the point $(S_0-K)^+, $ in case if $ S_{(0)} < K$ the set of non arbitrage prices coincides with the closed set $[(S_0-K)^+, S_0 - S_{(0)}].$
\end{te}

\begin{te}\label{tanijalib9}
On the probability space $\{\Omega_N, {\cal F}_N, P_N\},$  let the evolution of risky asset be given by the formula (\ref{tinwau2})
 with   the parameters  $ a_i, \ i=\overline{1,N},$ given by the formula (\ref{tanijalib6}).    For the payoff function  $f_1(x)=(K-x)^+, \ x \in (0, \infty), \ K >0, $    the fair price of super-hedge is given by the  formula 
 \begin{eqnarray}\label{tanijalib10}
\sup\limits_{Q \in M} E^Q f_1(S_N)=f_1\left( S_{(0)}  g_N\left([S]_{N}\right) \right).
\end{eqnarray}
 The set of non arbitrage prices coincides  with the closed interval \\ $\left[(K-S_0)^+, f_1\left( S_{(0)}  g_N\left([S]_{N}\right) \right)\right].$

The fair price of super-hedge  for the statistic  (\ref{vitochka3}), (\ref{vitochka4}) is given by the formula
 \begin{eqnarray}\label{vitochka7}
\sup\limits_{Q \in M} E^Q f_1(S_N)=f_1\left(S_0\frac{ S_{(0)}}{S_{(N)}}\right).
\end{eqnarray}
The set of non arbitrage prices coincides  with the closed interval  $\left[(K-S_0)^+, f_1\left(S_0\frac{ S_{(0)}}{S_{(N)}}\right)\right].$

The fair price of super-hedge is minimal one for the statistic (\ref{tanijalib6}) with $g_i(X_N)=g_N(X_N)=1, \ i=\overline{1,N-1},$ and is given by the formula
 \begin{eqnarray}\label{100tanijalib10}
\sup\limits_{Q \in M} E^Q f_1(S_N)=f_1\left( S_{(0)}\right).
\end{eqnarray}
The set of non arbitrage prices coincides  with the closed interval  $\left[(K-S_0)^+, f_1\left( S_{(0)}\right)\right].$
\end{te}

\begin{te}\label{tanijalib12}
On the probability space $\{\Omega_N, {\cal F}_N, P_N\},$  let the evolution of risky asset be given by the formula (\ref{tinwau2})
 with   the parameters  $ a_i, \ i=\overline{1,N},$ given by the formula (\ref{tanijalib6}).     For the payoff function  $f_1(S_0,S_1, \ldots, S_N)=\left(K-\frac{\sum\limits_{i=0}^N S_i}{N+1}\right)^+, \ K >0,  $  the fair price of super-hedge is given by the  formula 
 \begin{eqnarray}\label{tanijalib13}
\sup\limits_{Q \in M} E^Q f_1(S_0,S_1, \ldots, S_N)=\left(K-\frac{S_0+S_{(0)}\sum\limits_{i=1}^N g_i\left([S]_{N}\right) }{(N+1)}\right)^+.
\end{eqnarray}
 The set of non arbitrage prices coincides with the closed  interval \\ $\left[(K-S_0)^+, \left(K-\frac{S_0+S_{(0)}\sum\limits_{i=1}^N g_i\left([S]_{N}\right) }{(N+1)}\right)^+\right],$ if $K>\frac{S_0+S_{(0)}\sum\limits_{i=1}^N g_i\left([S]_{N}\right) }{(N+1)}.$ \\ For $K\leq \frac{S_0+S_{(0)}\sum\limits_{i=1}^N g_i\left([S]_{N}\right) }{(N+1)}$ the set of non arbitrage prices coincides with the point $0.$

The fair price of super-hedge is minimal one for the statistic (\ref{tanijalib6}) with $g_i(X_N)=g_N(X_N)=1, \ i=\overline{1,N-1},$ and is given by the  formula
 \begin{eqnarray}\label{100tanijalib13}
\sup\limits_{Q \in M} E^Q f_1(S_0,S_1, \ldots, S_N)=\left(K-\frac{S_0+S_{(0)} N }{(N+1)}\right)^+.
\end{eqnarray}
The set of non arbitrage prices coincides with the closed  interval \\ $\left[(K-S_0)^+, \left(K-\frac{S_0+S_{(0)} N }{(N+1)}\right)^+\right],$ if $K>\frac{S_0+S_{(0)} N }{(N+1)}.$  For $K\leq \frac{S_0+S_{(0)} N}{(N+1)}$ the set of non arbitrage prices coincides with the point $0.$

\end{te}

\begin{te}\label{tanijalib16}
 On the probability space $\{\Omega_N, {\cal F}_N, P_N\},$  let the evolution of risky asset be given by the formula (\ref{tinwau2})
with   the parameters  $ a_i, \ i=\overline{1,N},$ given by the formula (\ref{tanijalib6}).       For the payoff function  $f(S_0,S_1, \ldots, S_N)=\left(\frac{\sum\limits_{i=0}^N S_i}{N+1}- K\right)^+, \ K >0,  $     the fair price of super-hedge is given by the  formula 
$$ \sup\limits_{Q \in M} E^Qf(S_0, S_1, \ldots,S_N)=  $$
 \begin{eqnarray}\label{tanijalib17}
\left\{\begin{array}{l l} (S_0 - K)^+, & \mbox{if} \quad \frac{S_0+S_{(0)}\sum\limits_{i=1}^N g_i\left([S]_{N}\right) }{(N+1)} \geq K,\\
\left(S_0- \frac{S_0+S_{(0)}\sum\limits_{i=1}^N g_i\left([S]_{N}\right) }{(N+1)}\right), & \mbox{if} \quad \frac{S_0+S_{(0)}\sum\limits_{i=1}^N g_i\left([S]_{N}\right) }{(N+1)} < K.
\end{array} \right.
\end{eqnarray}
If $\frac{S_0+S_{(0)}\sum\limits_{i=1}^N g_i\left([S]_{N}\right) }{(N+1)} \geq K,$ then the set of non arbitrage prices coincides with the point $(S_0-K)^+, $ in case if
 $\frac{S_0+S_{(0)}\sum\limits_{i=1}^N g_i\left([S]_{N}\right) }{(N+1)} < K$ the set of non arbitrage prices coincides with the closed  interval $\left[(S_0-K)^+, \left(S_0- \frac{S_0+S_{(0)}\sum\limits_{i=1}^N g_i\left([S]_{N}\right) }{(N+1)}\right)\right].$

The fair price of super-hedge is minimal one for the statistic (\ref{tanijalib6}) with $g_i(X_N)=g_N(X_N)=1, \ i=\overline{1,N-1},$ and is given by the formula
 $$ \sup\limits_{Q \in M} E^Qf(S_0, S_1, \ldots,S_N)=  $$
 \begin{eqnarray}\label{100tanijalib17}
\left\{\begin{array}{l l} (S_0 - K)^+, & \mbox{if} \quad \frac{S_0+S_{(0)} N }{(N+1)} \geq K,\\
\left(S_0- \frac{S_0+S_{(0)} N }{(N+1)}\right), & \mbox{if} \quad \frac{S_0+S_{(0)} N }{(N+1)} < K.
\end{array} \right.
\end{eqnarray}
If $\frac{S_0+S_{(0)} N }{(N+1)} \geq K,$ then the set of non arbitrage prices coincides with the point $(S_0-K)^+, $ in case if
 $\frac{S_0+S_{(0)} N }{(N+1)} < K$ the set of non arbitrage prices coincides with the closed  interval $\left[(S_0-K)^+, \left(S_0- \frac{S_0+S_{(0)} N }{(N+1)}\right)\right].$
\end{te}

\section{Conclusions.}
Section 1 provides an overview of the achievements and formulates the main problem that has been solved.
Section 2 contains the formulation of conditions which must satisfy the evolution of risky assets.
In Section 3, conditions (\ref{1vitasja7}) - (\ref{3vitasja7}) are formulated for the set of nonnegative random variables with the help of which a family of measures is constructed in a recurrent way. In Lemma \ref{witka0}, conditions were found for the existence of bounded nonnegative random variables satisfying the conditions  (\ref{1vitasja7}) - (\ref{3vitasja7}).
In Lemma \ref{witka1}, it was proved that the family of measures introduced in the recurrent way is equivalent to the original measure.

Theorem \ref{witka2} gives sufficient conditions under which the introduced family of measures is the set of martingale measures equivalent to the original measure for the evolution of risky assets considered in Section 1.

In Section 4, relying on the concept of an exhaustive decomposition of a measurable space, in Lemma \ref{witka4}, we prove an integral inequality for a nonnegative random variable for the constructed family of martingale measures.

In Theorem \ref{witka5}, for a special class of evolutions of risky assets for the nonnegative random variable satisfying the integral inequality, obtained in Lemma \ref{witka4}, a pointwise system of inequalities is obtained.

In Lemma \ref{witka6}, on the basis of Lemma  \ref{witka4},  we obtained a pointwise system of inequalities for a nonnegative random variable for the general case of the evolution of risky assets.

Theorem \ref{witka7} contains sufficient conditions under the fulfillment of which   the resulting system of inequalities with respect to the nonnegative random variable has a solution whose right-hand side satisfies the condition: the conditional expectation of the right-hand side of the inequality with respect to the filtration is equal to 1.

Theorem \ref{1witka7} solves the same problem as in Theorem \ref{Tinnna1} for the general case of the evolution of risky assets.

In Section 5, based on the inequalities obtained in Theorems \ref{witka7} and \ref{1witka7}, we prove a theorem on the optional decomposition of nonnegative super-martingales with respect to the family of equivalent martingale measures.

The description of the family of equivalent martingale measures given in Theorem \ref{witka2} is rather general, therefore, in Section 6, a spot set of measures is introduced.
In Lemma \ref{vitasja114}, the representation is obtained for the family of spot measures.

Based on the concept of the spot family of measures, the family   of $\alpha$-spot
  measures based on a set of positive random variables is introduced.
Theorem  \ref{vitasja121} provides sufficient conditions for the integral over the set of $\alpha$-spot measures to be an integral over the set of spot measures.

 In Theorem \ref{10vitakolja121}, sufficient conditions are given when the family of spot measures is a family of martingale measures and the constructed family of measures, that is an integral over the set of $\alpha$-spot  measures, is a family of martingale measures being  equivalent to the original measure.

Theorem \ref{100vitakolja121} describes the class of evolutions of risky assets for which the family of equivalent martingale measures is such that each martingale measure is an integral over the set of spot measures.

Section 7 is devoted to the application of the results obtained in the previous sections. A class of random processes is considered, which contains well-known processes of the type  ARCH and GARCH ones. Two types of random processes are considered, those for which the price of an asset cannot go down to zero and those for which the price can go down to zero during the period under consideration. The first class of processes describes the evolution of well-managed assets. We will call these assets relatively stable.

Theorem \ref{vitochkakiss1} asserts that for the evolution of relatively stable assets in the period under consideration, the family of martingale measures is one and the same.
The family of martingale measures for the evolution of risky assets whose price can drop to zero is contained in the family of martingale measures for the evolution of relatively stable assets. Each of the martingale measures for the considered class of evolutions is an integral over the set of spot martingale measures. On this basis, the fair price of the super hedge is given by the formula (\ref{vitochkakiss2}).
In Theorems \ref{tinvika1} and \ref{tinvika10}, an interval of non-arbitrage prices is found for a wide class of payoff functions in the case when evolution describes relatively unstable assets.
This range is quite wide for the payment functions of standard put and call options. The fair price of the super hedge is in this case the starting price of the underlying asset. In Theorems \ref{tinnavika36}, \ref{tinnvika39} estimates are found for the fair price of the super-hedge for the introduced class of evolutions with respect to stable assets.
In Theorems \ref{pussytinavika1} and \ref{pussytinavika9}, formulas are found for the fair price of contracts with call and put options for the evolution of assets described by parametric processes.

In Theorems \ref{pussytinavika12} and \ref{pupsiktinavika16}, the same formulas are found for Asian-type put and call options. A characteristic feature of these estimates is that for the evolution of relatively stable assets, the fair price of the super hedge is less than the price of the initial price of the asset.

 In Section 8, the estimates of the parameters of risky assets included in the evolution are obtained. This result is contained in Theorem \ref{puptanijalib5}.
In Theorems \ref{tanijalib8} and \ref{tanijalib9}, formulas are found for the fair price of contracts with call and put options for the obtained parameter estimates, and the interval of non-arbitrage prices for different statistics is found. The same results are contained in Theorems \ref{tanijalib12}, \ref{tanijalib16} for Asian-style call and put options.

%\newpage
\vskip 5mm
%\centerline{{\bf References.}}

\end{document}